%% file: k_center_arxiv.tex
\documentclass[a4paper,10pt]{article}
\usepackage[utf8]{inputenc}
 \usepackage{amsmath,amsfonts,amssymb}
 \usepackage{amsthm}
 \usepackage{mathtools}
 \usepackage{amsopn}
 \usepackage{tabularx,lipsum,environ}
 \usepackage{paralist}
 \usepackage{microtype}
 \usepackage{authblk}
 \usepackage{hyperref}
 \usepackage{fancyhdr}
 \usepackage{rotating}
 \usepackage{overpic}
 \usepackage{ucs}
 \usepackage{enumerate}
 \usepackage{graphicx}\usepackage{capt-of}
 \usepackage{booktabs}
 \usepackage{varwidth}

\title{Structural Parameters, Tight Bounds, and Approximation for $(k,r)$-Center\thanks{Research partially supported by the PSL research project MULTIFAC.}}
\author{Ioannis Katsikarelis}
\author{Michael Lampis}
\author{Vangelis Th.\ Paschos}
\affil{Université Paris-Dauphine, PSL Research University, CNRS, UMR 7243 \\ LAMSADE, 75016, Paris, France, \texttt{\{ioannis.katsikarelis|michail.lampis|paschos\}@lamsade.dauphine.fr}}
\theoremstyle{plain}
\newtheorem{theorem}{Theorem}
\newtheorem{lemma}[theorem]{Lemma}
\newtheorem{corollary}[theorem]{Corollary}

\newlength{\defbaselineskip}
\newcommand{\setlinespacing}[2]%
           {\setlength{\baselineskip}{#1 \defbaselineskip}}

\newlength{\btw}
\newlength{\stw}
\newsavebox\tmpbox 

\newcommand{\cw}{\ensuremath\textrm{cw}}
\newcommand{\tw}{\ensuremath\textrm{tw}}
\newcommand{\pw}{\ensuremath\textrm{pw}}
\newcommand{\td}{\ensuremath\textrm{td}}
\newcommand{\vc}{\ensuremath\textrm{vc}}
\newcommand{\fvs}{\ensuremath\textrm{fvs}}
\newcommand{\KC}{$(k,r)$-\textsc{Center}}
\newcommand{\dl}{\ensuremath\textrm{dl}}
\newcommand{\DL}{\ensuremath\textrm{DL}}
\date{}

\begin{document}

\maketitle

\begin{abstract} In \KC\ we are given a (possibly edge-weighted) graph and are
asked to select at most $k$ vertices (centers), so that all other vertices are
at distance at most $r$ from a center. In this paper we provide a number of
tight fine-grained bounds on the complexity of this problem with respect to
various standard graph parameters. Specifically:
\begin{itemize}
\item For any $r\ge 1$, we show an algorithm that solves the problem in
$O^*((3r+1)^{\cw})$ time, where $\cw$ is the clique-width of the input graph,
as well as a tight SETH lower bound matching this algorithm's performance. As a
corollary, for $r=1$, this closes the gap that previously existed on the
complexity of \textsc{Dominating Set} parameterized by $\cw$.
\item We strengthen previously known FPT lower bounds, by showing that \KC\ is
W[1]-hard parameterized by the input graph's vertex cover (if edge weights are
allowed), or feedback vertex set, even if $k$ is an additional parameter.  Our
reductions imply tight ETH-based lower bounds.  
Finally, we devise an algorithm parameterized by vertex cover for unweighted graphs.
\item We show that the complexity of the problem parameterized by tree-depth is
$2^{\Theta(\td^2)}$, by showing an algorithm of this complexity and a tight
ETH-based lower bound.
\end{itemize}
We complement these mostly negative results by providing \emph{FPT
approximation schemes} parameterized by clique-width or treewidth, which work
efficiently independently of the values of $k,r$. In particular, we give
algorithms which, for any $\epsilon>0$, run in time
$O^*((\tw/\epsilon)^{O(\tw)})$, $O^*((\cw/\epsilon)^{O(\cw)})$ and return a
$(k,(1+\epsilon)r)$-center if a $(k,r)$-center exists, thus circumventing the
problem's W-hardness.  \end{abstract}

\input{intro.tex}

\input{cw.tex}

\input{vc-fvs-td.tex}

\input{tw-approx.tex}

\input{cw-approx.tex}

\section{Conclusion}
In this paper we considered the \KC\ problem, a distance-based generalization of \textsc{Dominating Set}. Due to the problem's well-investigated hardness and inapproximability, our focus has been on structural parameterization. In particular, we give tight fine-grained bounds on the complexity of \KC\ with respect to the well-known graph parameters treewidth $\tw$, clique-width $\cw$, tree-depth $\td$, vertex cover $\vc$ and feedback vertex set $\fvs$:
\begin{itemize}
 \item A Dynamic Programming algorithm of running time $O^*((3r+1)^{\cw})$ and a matching lower bound based on the SETH, that for $r=1$ closed a complexity gap for \textsc{Dominating Set} parameterized by $\cw$.
 \item W[1]-hardness and ETH-based lower bounds of $n^{o(\vc+k)}$ and $n^{o(\fvs+k)}$ for edge-weighted and unweighted graphs, respectively, while these are complemented by an $O^*(5^{\vc})$-time FPT algorithm for the unweighted case.
 \item An algorithm solving the problem in time $O^*(2^{O(\td)^2})$, that assuming the ETH would be optimal as well.
 \item Algorithms computing for any $\epsilon>0$ a $(k,(1+\epsilon)r)$-center in time $O^*((\tw/\epsilon)^{O(\tw)})$, or $O^*((\cw/\epsilon)^{O(\cw)})$, if a $(k,r)$-center exists in the graph, assuming a tree decomposition of width $\tw$ is provided along with the input.
\end{itemize}
We again observe that the problem's hardness for $\vc$ on edge-weighted graphs (being in itself uncommon) and $\fvs$ for the unweighted case points to the increasing divergence of the problem's complexity from that of \textsc{Dominating Set} for increasing values of $r$. This in turn implies the necessity of keeping the value of this distance parameter bounded for efficient exact resolution, while the alternative of producing inexact solutions is incorporated in our approximation schemes.

Considering also previous related work (notably \cite{BorradaileL16}), the above results can be seen to complete the picture on the problem's structurally parameterized complexity. Nevertheless, some remaining open questions concern the sharpening of our ETH-based lower bounds using as a starting point the more precise SETH, as well as the complexity status of the problem with respect to other structural parameters, such as rankwidth, modular-width or neighborhood diversity.

\section*{Acknowledgement}
The very useful remarks and suggestions of the two anonymous referees are gratefully acknowledged.

\bibliography{kr_center_citations}

\end{document}

%% file: intro.tex
\section{Introduction}\label{sec:intro}

In this paper we study the \KC\ problem: given a graph $G=(V,E)$ and a weight
function $w: E\to \mathbb{N}^+$ which satisfies the triangle inequality and
defines the length of each edge, we are asked if there exists a set $K$ (the \textit{center-set}) of at
most $k$ vertices of $V$, so that $\forall u\in V\setminus K$ we have
$\min_{v\in K}d(v,u)\le r$, where $d(v,u)$ denotes the shortest-path distance
from $v$ to $u$ under weight function $w$. If $w$ assigns weight $1$ to all
edges we say that we have an instance of \emph{unweighted} \KC.
\KC\ is an extremely well-investigated optimization problem with numerous applications.
It has a long history, especially from the point of view of approximation
algorithms, where the objective is typically to minimize $r$ for a given $k$
\cite{Gonzalez85,HochbaumS86,V01,KhullerS00,ChechikP15,FederG88,PanigrahyV98,Krumke95,KhullerPS00,AgarwalP02,EisenstatKM14}.
The converse objective (minimizing $k$ for a given $r$) has also been
well-studied, with the problem being typically called $r$-\textsc{Dominating
Set} in this case~\cite{Slater76,Bar-IlanKP93,BrandstadtD98,LokshtanovMPRS13,CoelhoMW15}.

Because \KC\ generalizes \textsc{Dominating Set} (which corresponds to the case
$r=1$), the problem can already be seen to be hard, even to approximate (under
standard complexity assumptions).  In particular, the optimal $r$ cannot be
approximated in polynomial time by a factor better than $2$ (even on planar
graphs \cite{FederG88}), while $k$ cannot be approximated by a factor better
than $\ln n$ \cite{Moshkovitz15}. Because of this hardness, we are strongly
motivated to investigate the problem's complexity when the input graph has some
restricted structure. 

In this paper our goal is to perform a complete analysis of the complexity of
\KC\ that takes into account this input structure by using the framework of
parameterized complexity.  In particular, we provide \emph{fine-grained} upper
and lower bound results on the complexity of \KC\ with respect to the most
widely studied parameters that measure a graph's structure: treewidth \textbf{$\tw$},
clique-width \textbf{$\cw$}, tree-depth \textbf{$\td$}, vertex cover
\textbf{$\vc$}, and feedback vertex set \textbf{$\fvs$}. In addition to the
intrinsic value of determining the precise complexity of \KC, this approach is
further motivated by the fact that FPT algorithms for this problem have often
been used as building blocks for more elaborate approximation algorithms
\cite{DemaineFHT05,EisenstatKM14}.  Indeed, (some of) these questions have
already been considered, but we provide a number of new results that build on
and improve the current state of the art.  Along the way, we also close a gap
on the complexity of the flagship \textsc{Dominating Set} problem parameterized
by clique-width.  Specifically, we prove the following:
\begin{itemize}
\item \emph{\KC\ can be solved (on unweighted graphs) in time $O^*((3r+1)^{\cw})$ (if
a clique-width expression is supplied with the input), but it cannot be solved
in time $O^*((3r+1-\epsilon)^{\cw})$ for any fixed $r\ge1$, unless the Strong
Exponential Time Hypothesis (SETH) \cite{ImpagliazzoP01,ImpagliazzoPZ01} fails.}

The algorithmic result relies on standard techniques (dynamic programming on
clique-width, fast subset convolution), as well as several problem-specific
observations which are required to obtain the desired table size. The SETH
lower bound follows from a direct reduction from \textsc{SAT}. A noteworthy
consequence of our lower bound result is that, \emph{for the case of
\textsc{Dominating Set}, it closes the gap between the complexity of the best
known algorithm ($O^*(4^\cw)$ \cite{BodlaenderLRV10}) and the best previously
known lower bound ($O^*((3-\epsilon)^\cw)$ \cite{LokshtanovMS11a})}.
\item \emph{\KC\ cannot be solved in time $n^{o(\vc+k)}$ on edge-weighted graphs, or
time $n^{o(\fvs+k)}$ on unweighted graphs, unless the Exponential Time
Hypothesis (ETH) is false.}

It was already known that an FPT algorithm
parameterized just by $\tw$ (for unbounded $r$) is unlikely to be possible~\cite{BorradaileL16}. These results show that the same holds for the two more
restrictive parameters $\fvs$ and $\vc$, even if $k$ is also added as a
parameter. They are (asymptotically) tight, since it is easy to obtain
$O^*(n^\fvs)$, $O^*(n^\vc)$, and $O^*(n^k)$ algorithms. We remark that \KC\ is
a rare example of a problem that turns out to be hard parameterized by $\vc$.
We complement these lower bounds by an FPT algorithm for the unweighted case,
running in time $O^*(5^\vc)$.
\item \emph{\KC\ can be solved in time $O^*(2^{O(\td^2)})$ for unweighted graphs, but
if it can be solved in time $O^*(2^{o(\td^2)})$, then the ETH is false.}

Here the upper bound follows from known connections between a graph's tree-depth and
its diameter, while the lower bound follows from a reduction from
$3$-\textsc{SAT}. We remark that this is a somewhat uncommon example of a
parameterized problem whose parameter dependence turns out to be exponential in
the \emph{square} of the parameter.
\end{itemize}
The results above, together with the recent work of~\cite{BorradaileL16} which showed tight bounds of $O^*((2r+1)^{\tw})$ regarding
the problem's complexity parameterized by~$\tw$, give a complete, and often
fine-grained, picture on \KC\ for the most important graph parameters. One of
the conclusions that can be drawn is that, as a consequence of the problem's
hardness for $\vc$ (in the weighted case) and $\fvs$, there are few
cases where we can hope to obtain an FPT algorithm without bounding the value
of $r$. In other words, as $r$ increases the
complexity of exactly solving the problem quickly degenerates away from the
case of \textsc{Dominating Set}, which is FPT for all considered parameters. 

A further contribution of this paper is to complement this negative view by
pointing out that it only applies if one insists on solving the problem
\emph{exactly}. If we allow algorithms that return a
$(1+\epsilon)$-approximation to the optimal $r$, for arbitrarily small
$\epsilon>0$ and while respecting the given value of $k$, we obtain the
following:
\begin{itemize}
\item \emph{There exist algorithms which, for any $\epsilon>0$, when given a graph
that admits a $(k,r)$-center, return a $(k,(1+\epsilon)r)$-center in time
$O^*((\tw/\epsilon)^{O(\tw)})$, or $O^*((\cw/\epsilon)^{O(\cw)})$, assuming a tree
decomposition or clique-width expression is given in the input.}
\end{itemize}
The $\tw$ approximation algorithm is based on a technique introduced in
\cite{Lampis14}, while the $\cw$ algorithm relies on a new extension of an idea
from \cite{GurskiW00}, which may be of independent interest.
Thanks to these approximation algorithms, we arrive at an improved
understanding of the complexity of \KC\ by including the question of
approximation, and obtain algorithms which continue to work efficiently even
for large values of $r$.  Figure \ref{fig:param_relations} illustrates the
relationships between parameters and Table \ref{parameter_table} summarizes our results.
\begin{table}[htbp]
\begin{minipage}[b]{0.3\linewidth}
\centering
\includegraphics[width=40mm]{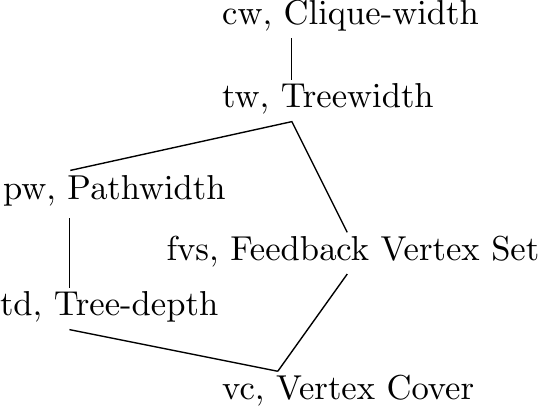}
\captionof{figure}{Relationships of parameters. Algorithmic results are inherited
 downwards, hardness results upwards.}
\label{fig:param_relations}
\end{minipage}
\hfill
\begin{minipage}[b]{0.66\linewidth}
\centering
\begin{tabular}{|c|c|c|c|c|c|c|}
\hline
              & cw            & tw                           & fvs           & td          & vc          \\ \hline
FPT exact     & \ref{thm:cw_dp} (w/u) & \ref{thm:tw-exact} (w/u)           &               & \ref{thm:td-alg} (u) & \ref{thm:vc_algo} (u) \\ \hline
FPT-AS        & \ref{thm_cw_approx} (w/u) & \ref{thm:tw_approx} (w/u)                  &               &             &             \\ \hline
SETH LB       & \ref{cw_SETH_LB} (u)   &                     &               &             &             \\ \hline
ETH LB        &               &                                & \ref{thm:W_hard_FVSk} (w/u) & \ref{thm:eth_TD_LB} (u) & \ref{thm:W_hard_VCk} (w) \\ \hline
W{[}1{]}-hard &               &                                 & \ref{thm:W_hard_FVSk} (w/u) &             & \ref{thm:W_hard_VCk} (w) \\ \hline
\end{tabular}
    \caption{A summary of our results (theorem numbers) for all considered parameters. Initials u/w denote the unweighted/weighted variants of the problem.}
    \label{parameter_table}
\end{minipage}
\end{table}

\subparagraph{Related Work:} Our work follows  upon recent work by~\cite{BorradaileL16}, which showed that \KC\ can be solved in
$O^*((2r+1)^\tw)$, but not faster (under SETH), while its connected variant can
be solved in $O^*((2r+2)^\tw)$, but not faster. This paper in turn generalized
previous results on \textsc{Dominating Set} for which a series of papers had
culminated into an $O^*(3^\tw)$ algorithm~\cite{TelleP93,AlberN02,RooijBR09},
while on the other hand,~\cite{LokshtanovMS11a} showed that
an $O^*((3-\epsilon)^\pw)$ algorithm would violate the SETH, where $\pw$
denotes the input graph's pathwidth.  The complexity of \KC\ by the related
parameter branchwidth had previously been considered in~\cite{DemaineFHT05} where an
$O^*((2r+1)^{\frac{3}{2}\textrm{bw}})$ algorithm is given. Moreover,~\cite{Marx05}
showed the problem parameterized by the number $k$ of centers to be W[1]-hard
in the $L_{\infty}$ metric, in fact analysing \textsc{Covering
Points with Squares}, a geometric variant. On clique-width, a
$O^*(4^{\textrm{cw}})$-time algorithm for \textsc{Dominating Set} was
given in~\cite{BodlaenderLRV10}, while~\cite{OumSV14} notes that the lower
bound of~\cite{LokshtanovMS11a} for pathwidth/treewidth would also imply no
$(3-\epsilon)^{\textrm{cw}}\cdot n^{O(1)}$-time algorithm exists for
clique-width under SETH as well, since clique-width is at most 1 larger than
pathwidth. For the edge-weighted variant, \cite{Feldmann15} shows that a $(2-\epsilon)$-approximation is W[2]-hard for parameter $k$ and NP-hard for graphs of \emph{highway dimension} $h=O(\log^2n)$, while also offering a 3/2-approximation algorithm of running time $2^{O(kh\log(h))}\cdot n^{O(1)}$, exploiting the similarity of this problem with that of solving \textsc{Dominating Set} on graphs of bounded $\vc$. Finally, for unweighted graphs, \cite{Leitert17} provides efficient  (linear/polynomial) algorithms computing $(r+O(\mu))$-dominating sets and $+O(\mu)$-approximations for \KC, where $\mu$ is the \emph{tree-breadth} or \emph{cluster diameter} in a layering partition of the input graph, while \cite{EisenstatKM14} gives a polynomial-time bicriteria approximation scheme for graphs of bounded genus.

\section{Definitions and Preliminaries}\label{sec:prel}
Dealing with \KC\ 
we allow, in general, the weight function $w$ to be non-symmetric. We also require edge weights to
be strictly positive integers but, as we will see (Lemma~\ref{lem:zero-ok}),
this is not a significant restriction.

We use standard graph-theoretic notation. For a graph $G=(V,E)$, $n=|V|$
denotes the number of vertices, $m=|E|$ the number of edges and for a subset $X\subseteq V$,
$G[X]$ denotes the graph induced by $X$. 
Further, we assume the reader has some familiarity with standard definitions
from parameterized complexity theory, such as the classes FPT, W[1] (see
\cite{CyganFKLMPPS15,FlumG06,DowneyF13}). For a parameterized problem with
parameter $k$, an FPT-AS is an algorithm which for any $\epsilon>0$ runs in
time $O^*(f(k,\frac{1}{\epsilon}))$ (i.e.\ FPT time when parameterized by
$k+\frac{1}{\epsilon}$) and produces a solution at most a multiplicative factor
$(1+\epsilon)$ from the optimal (see \cite{Marx08}). We use $O^*(\cdot)$ to
imply omission of factors polynomial in $n$. In this paper we present
approximation schemes with running times of the form $(\log
n/\epsilon)^{O(k)}$. These can be seen to imply an FPT running time by a
well-known win-win argument:
\begin{lemma}\label{lem:fpt-logn}
If a parameterized problem with parameter $k$ admits, for some $\epsilon>0$, an
algorithm running in time $O^*((\log n /\epsilon)^{O(k)})$, then it also admits an
algorithm running in time $O^*((k/\epsilon)^{O(k)})$.
\end{lemma}
 \begin{proof}
 We consider two cases: if $k\le \sqrt{\log n}$ then $(\log n/\epsilon)^{O(k)} =
 (1/\epsilon)^{O(k)} (\log n)^{O(\sqrt{\log n})} = O^*((1/\epsilon)^{O(k)})$. If
 on the other hand, $k> \sqrt{\log n}$, we have $\log n \le k^2$, so $O^*((\log n
 /\epsilon)^{O(k)}) = O^*((k/\epsilon)^{O(k)})$. \end{proof}
A \emph{tree decomposition} of a graph $G=(V,E)$ is a pair $(\mathcal{X},T)$ with $T=(I,F)$ a tree and $\mathcal{X}=\{X_i|i\in I\}$ a family of subsets of $V$ (called \emph{bags}), one for each node of $T$, with the following properties:
 \begin{enumerate}[1)]
  \item $\bigcup_{i\in I}X_i=V$;
  \item for all edges $(v,w)\in E$, there exists an $i\in I$ with $v,w\in X_i$;
  \item for all $i,j,k\in I$, if $j$ is on the path from $i$ to $k$ in $T$, then $X_i\cap X_k\subseteq X_j$.
 \end{enumerate}
 The \emph{width} of a tree decomposition $((I,F),\{X_i|i\in I\})$ is $\max_{i\in I}|X_i|-1$. The \emph{treewidth} of a graph $G$ is the minimum width over all tree decompositions of $G$, denoted by $\textrm{tw}(G)$ (\cite{BodlaenderK08,Bodlaender06,Bodlaender00}).

Moreover, for rooted $T$, let $G_i=(V_i,E_i)$ denote the \emph{terminal subgraph} defined by node $i\in I$, i.e.\ the induced subgraph of $G$ on all vertices in bag $i$ and its descendants in $T$. Also let $N_{i}(v)$ denote the neighborhood of vertex $v$ in $G_i$ and $d_i(u,v)$ denote the distance between vertices $u$ and $v$ in $G_i$, while $d(u,v)$ (absence of subscript) is the distance in $G$.

In addition, a tree decomposition can be converted to a \emph{nice} tree decomposition of the same width (in $O(\textrm{tw}^2\cdot n)$ time and with $O(\textrm{tw}\cdot n)$ nodes): the tree here is rooted and binary, while nodes can be of four types: 
\begin{enumerate}[a)]
 \item Leaf nodes $i$ are leaves of $T$ and have $|X_i|=1$;
 \item Introduce nodes $i$ have one child $j$ with $X_i=X_j\cup\{v\}$ for some vertex $v\in V$ and are said to \emph{introduce} $v$;
 \item Forget nodes $i$ have one child $j$ with $X_i=X_j\setminus\{v\}$ for some vertex $v\in V$ and are said to \emph{forget} $v$;
 \item Join nodes $i$ have two children denoted by $i-1$ and $i-2$, with $X_i=X_{i-1}=X_{i-2}$.
\end{enumerate}
Nice tree decompositions were introduced by Kloks in \cite{Kloks94} and using them does not in general give any additional algorithmic possibilities, yet algorithm design becomes considerably easier.

We will also make use of the notion of \emph{cliquewidth} (see \cite{CourcelleMR00}): the set of graphs of cliquewidth $\textrm{cw}$ is the set of vertex-labelled graphs that can be inductively constructed by using the following operations:
\begin{enumerate}[1)]
  \item Introduce: $i(l)$, for $l\in[1,\textrm{cw}]$ is the graph consisting of a single vertex with label $l$;
  \item Join: $\eta(G,a,b)$, for $G$ having cliquewidth $\textrm{cw}$ and $a,b\in[1,\textrm{cw}]$ is the graph obtained from $G$ by adding all possible edges between vertices of label $a$ and vertices of label $b$;
  \item Rename: $\rho(G,a,b)$, for $G$ having cliquewidth $\textrm{cw}$ and $a,b\in[1,\textrm{cw}]$ is the graph obtained from $G$ by changing the label of all vertices of label $a$ to $b$;
  \item Union: $G_1\cup G_2$, for $G_1,G_2$ having cliquewidth $\textrm{cw}$ is the disjoint union of graphs $G_1,G_2$.
 \end{enumerate} Note we here assume the labels are integers in $[1,\textrm{cw}]$, for ease of exposition.
 
A \emph{cliquewidth expression} of width $\textrm{cw}$ for $G=(V,E)$ is a recipe for constructing a $\textrm{cw}$-labelled graph isomorphic to $G$. More formally, a cliquewidth expression is a rooted binary tree $T_G$, such that each node $t\in T_G$ has one of four possible types, corresponding to the operations given above. In addition, all leaves are introduce nodes, each introduce node has a label associated with it and each join or rename node has two labels associated with it. For each node $t$, the graph $G_t$ is defined as the graph obtained by applying the operation of node $t$ to the graph (or graphs) associated with its child (or children). All graphs $G_t$ are subgraphs of $G$ and for all leaves of label $l$, their associated graph is $i(l)$.

 Additionally, we will use the parameters
\emph{vertex cover number} and \emph{feedback vertex set number}  of a graph
$G$, which are the sizes of the minimum vertex set whose deletion leaves the
graph edgeless, or acyclic, respectively. Finally, we will consider the related
notion of \emph{tree-depth} \cite{NesetrilM06}, which is defined as the minimum
height of a rooted forest whose completion (the graph obtained by connecting
each node to all its ancestors) contains the input graph as a subgraph. We will
denote these parameters for a graph $G$ as
$\tw(G),\pw(G),\cw(G),\vc(G),\fvs(G)$, and $\td(G)$, and will omit~$G$ if it is
clear from the context. We recall the following well-known relations between these parameters~\cite{BodlaenderGHK95,CourcelleO00} which
justify the hierarchy given in Figure \ref{fig:param_relations}:
\begin{lemma}\label{lem:relations} \cite{BodlaenderGHK95,CourcelleO00} For any
graph $G$ we have $\tw(G) \le \pw(G) \le \td(G) \le \vc(G)$, $\tw(G)\le \fvs(G)
\le vc(G)$, $\cw(G)\le \pw(G)+1$, and $\cw(G) \le 2^{\tw(G)+1}+1$.  \end{lemma}
We also recall here the two main complexity assumptions used in this paper
\cite{ImpagliazzoP01,ImpagliazzoPZ01}. The Exponential Time Hypothesis (ETH)
states that 3-\textsc{SAT} cannot be solved in time $2^{o(n+m)}$ on instances
with $n$ variables and $m$ clauses. The Strong Exponential Time Hypothesis
(SETH) states that for all $\epsilon>0$, there exists an integer $k$ such that
$k$-\textsc{SAT} cannot be solved in time $(2-\epsilon)^n$ on instances of
$k$-\textsc{SAT} with $n$ variables.

Finally, \textsc{$k$-Multicolored Independent Set} is defined as follows: we are given a graph $G=(V,E)$, with $V$ partitioned into $k$ cliques $V=V_1\uplus\dots\uplus V_k$, $|V_i|=n, \forall i\in[1,k]$, and are asked to find an $S\subseteq V$, such that $G[S]$ forms an independent set and $|S\cap V_i|=1,\forall i\in[1,k]$.

%% file: cw.tex
\section{Clique-width}

\subsection{Lower bound based on SETH}\label{sec_cw_LB} 

In this section we prove that for any fixed constant
$r\ge1$, the existence of any algorithm for \textsc{$(k,r)$-Center} of running
time $O^*((3r+1-\epsilon)^{\textrm{cw}})$, for some $\epsilon>0$,
would imply the existence of some algorithm for \textsc{SAT} of running time
$O^*((2-\delta)^n)$, for some $\delta>0$.

Before we proceed, let us recall the high-level
idea behind the SETH lower bound for \textsc{Dominating Set} given in
\cite{LokshtanovMS11a}, as well its generalization to \KC\ given in
\cite{BorradaileL16}. In both cases the key to the reduction is the
construction of long paths, which are conceptually divided into blocks of
$2r+1$ vertices. The intended solution consists of selecting, say, the $i$-th
vertex of a block of a path, and repeating this selection in all blocks of this
path. This allows us to encode $(2r+1)^t$ choices, where $t$ is the number of
paths we make, which ends up being roughly equal to the treewidth of the
construction. 

The reason this construction works in the converse direction is that, even
though the optimal \KC\ solution may ``cheat'' by selecting the $i$-th vertex
of a block, and then the $j$-th vertex of the next, one can see that we must
have $j\le i$.  Hence, by making the paths that carry the solution's encoding
long enough we can ensure that the solution eventually settles into a pattern
that encodes an assignment to the original formula (which can be ``read'' with
appropriate gadgets).

In our lower bound construction for clique-width we need to be able to ``pack''
more information per unit of width: instead of encoding $(2r+1)$ choices for
each unit of treewidth, we need to encode $(3r+1)$ choices for each label. Our
high-level plan to achieve this is to use \emph{a pair} of long paths for each
label. Because we only want to invest one label for each pair of paths we are
forced to periodically (every $2r+1$ vertices) add cross edges between them, so
that the connection between blocks can be performed with a single join
operation. See the paths $A_1,B_1$ in Figure \ref{fig:global_construction} for
an illustration.

Our plan now is to encode a solution by selecting a pair of vertices that will
be repeated in each block, for example every $i$-th vertex of $A_1$ and every
$j$-th vertex of $B_1$. One may naively expect that this would allow us to
encode $(2r+1)^2$ choices for each label (which would lead to a SETH lower
bound that would contradict the algorithm of Section \ref{sec_cw_DP}). However,
because of the cross edges, the optimal \KC\ solution is not as well-behaved on
a pair of cross-connected paths as it was on a path, and this makes it much
harder to execute the converse direction of the reduction. Our strategy will
therefore be to identify $(3r+1)$ canonical selection pairs, order them, and
show that any valid solution must be well-behaved with respect to these pairs.
We will then use these pairs to encode $(3r+1)^{\cw}$ choices and
still get the converse direction of the reduction to work (if all paths are
sufficiently long).

We next describe the construction of a graph $G$, given some $\epsilon<1$ and an
instance $\phi$ of \textsc{SAT} with $n$ variables and $m$ clauses. We first
choose an integer $p\ge\dfrac{1}{(1-\lambda)\log_{2}(3r+1)}$, for $\lambda=\log_{3r+1}(3r+1-\epsilon)<1$, for reasons that become
apparent in the proof of Theorem \ref{cw_SETH_LB}. Note that for the results of this section, both $r$ and $p$ are considered constants. We then group the variables
of $\phi$ into $t=\lceil\frac{n}{\gamma}\rceil$ groups $F_1,\dots,F_t$, for
$\gamma=\lfloor\log_2(3r+1)^p\rfloor$, being also the maximum size of any such
group. Our construction uses a main Block gadget $\hat{G}$ and three smaller gadgets $\hat{T}_N$, $\hat{X}_N$ and $\hat{U}_N$ as parts.
\subparagraph{Guard gadget $\hat{T}_N$:} This gadget has $N$ \emph{input}
vertices and its purpose is to allow for any selection of a single input vertex
as a center to cover all vertices within the gadget, while offering no paths of
length $\le r$ between the inputs. Thus if at least two guard gadgets are attached to a set of $N$ input vertices, any minimum-sized center-set will select one of them to cover all vertices in the gadgets, without interfering with whether the other inputs are covered by some outside selection.
 \subparagraph{Construction and size/cw bounds for $\hat{T}_N$:}  Given vertices $v_1,\dots,v_N$, we construct
 the gadget as follows: we first make $\lfloor\frac{r}{2}\rfloor$ vertices
 $u_i^{1},\dots,u_i^{\lfloor\frac{r}{2}\rfloor}$ forming a path for each $v_i$
 and we connect each $u_i^1$ to each $v_i$. We then make another path on
 $\lceil\frac{r}{2}\rceil$ vertices, called
 $w^1,\dots,w^{\lceil\frac{r}{2}\rceil}$, and we make vertices
 $u_i^{\lfloor\frac{r}{2}\rfloor}$ adjacent to the starting vertex $w^1$ of this
 path for all $i\in[1,N]$. Finally, if $r$ is even, we make all vertices
 $u_i^{\lfloor\frac{r}{2}\rfloor}$ into a clique, while for odd $r$ the
 construction is already complete. Figure \ref{fig:guard_gadget} provides an
 illustration. The number of vertices in the gadget is $|\hat{T}_N|=N\lfloor\frac{r}{2}\rfloor+\lceil\frac{r}{2}\rceil$, while the gadget can also be constructed by a clique-width expression using at most this number of labels, by handling each vertex as an individual label.
 \begin{figure}[htbp]
 \centerline{\includegraphics[width=70mm]{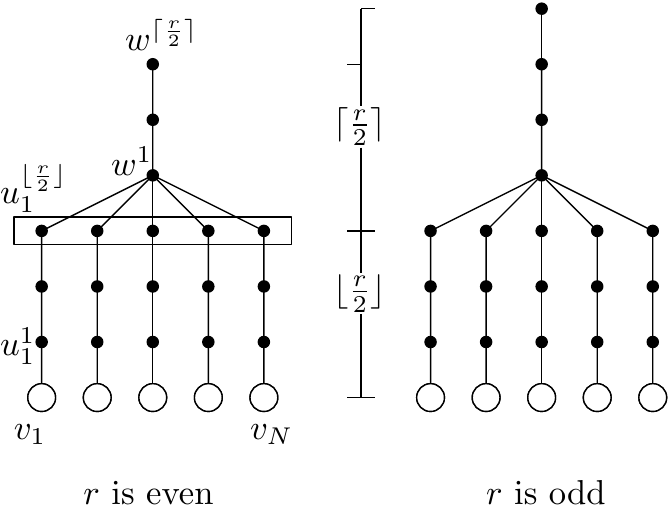}}
 \caption{A general
 picture of the guard gadget $\hat{T}_N$ for even and odd $r$. Note the box
 indicating vertices forming a clique for the case of even $r$.}
 \label{fig:guard_gadget} \end{figure}
\begin{lemma}\label{guard_gadget_proof}
 All vertices of gadget $\hat{T}_N$ are at distance $\le r$ from any input vertex $v_i,i\in[1,N]$, while $d(v_i,v_j)=r+1$, for any $i\not=j\in[1,N]$.
\end{lemma}
 \begin{proof}
  The distance from any input vertex $v_i$, with $i\in[1,N]$, to vertex $u_i^{\lfloor\frac{r}{2}\rfloor}$ is $d(v_i,u_i^{\lfloor\frac{r}{2}\rfloor})=\lfloor\frac{r}{2}\rfloor$, while the distance from $u_i^{\lfloor\frac{r}{2}\rfloor}$ to $w^{\lceil\frac{r}{2}\rceil}$ is $d(u_i^{\lfloor\frac{r}{2}\rfloor},w^{\lceil\frac{r}{2}\rceil})=\lceil\frac{r}{2}\rceil$, thus $d(v_i,w^{\lceil\frac{r}{2}\rceil})=\lfloor\frac{r}{2}\rfloor+\lceil\frac{r}{2}\rceil=r$. Further, the distance from $u_i^{\lfloor\frac{r}{2}\rfloor}$ to any vertex $u_j^1$, for $j\not=i$, is $d(u_i^{\lfloor\frac{r}{2}\rfloor},u_j^1)=\lfloor\frac{r}{2}\rfloor$ for even $r$, and $d(u_i^{\lfloor\frac{r}{2}\rfloor},u_j^1)=1+\lfloor\frac{r}{2}\rfloor$ for odd $r$. Thus $d(v_i,u_j^1)=\lfloor\frac{r}{2}\rfloor+\lfloor\frac{r}{2}\rfloor=r$ for even $r$, and $d(v_i,u_j^1)=\lfloor\frac{r}{2}\rfloor+1+\lfloor\frac{r}{2}\rfloor=r$ for odd $r$. Thus any vertex on these paths is at distance $\le r$ from $v_i$, while any other input vertex $v_j$, being at distance $1$ from its corresponding $u_j^1$, is at distance $r+1$ from $v_i$.
 \end{proof}
\subparagraph{Clique gadget $\hat{X}_N$:} This gadget again has $N$ input
vertices and the aim now is to make sure that any selection of a single input
vertex as a center will cover all other vertices of the gadget, that no
selection of any single vertex that is not an input would suffice instead,
while all paths connecting the inputs are of distance exactly $r$.
 \subparagraph{Construction and size/cw bounds for $\hat{X}_N$:} Given
 vertices $v_1,\dots,v_N$, we first connect them to each other by paths on $r-1$
 new vertices, so that the distances between any two of them are exactly $r$. We
 then make all \emph{middle} vertices that lie at distance
 $\lfloor\frac{r}{2}\rfloor$ from some vertex $v_i$ on these paths adjacent to
 each other (into a clique): for even $r$ the middle vertex of each path is at
 distance $r/2$ from the path's endpoint vertices $v_i,v_j$, while for odd $r$
 the middle vertices are the two vertices at distance $(r-1)/2$ from one of the
 path's endpoints $v_i,v_j$. Thus all vertices on these paths are at distance at
 most $r$ from any vertex $v_i$, as they lie within distance
 $<\lfloor\frac{r}{2}\rfloor$ from a middle vertex on their path, which is at
 distance exactly $\lfloor\frac{r}{2}\rfloor+1$ from any input vertex $v_i$.
 Next, we make another vertex $u_{i,j}^{l}$ for each middle vertex of these
 paths between $v_i,v_j$ (and $l\in[1,2]$ for odd $r$) and we make it adjacent
 to its corresponding middle vertex. We then add another vertex $x$ that we also
 connect to all $v_i$ vertices by paths using $r-1$ new vertices (thus at distance $r$ from each $v_i$), making the
 middle vertices on these new paths adjacent to each other as well (into
 another, disjoint clique). See Figure \ref{fig:clique_gadget} for an
 illustration. The size of the gadget is $|\hat{X}_N|=N+1+{N \choose 2}\cdot r+N(r-1)$ for even $r$ and $|\hat{X}_N|=N+1+{N \choose 2}\cdot(r+1)+N(r-1)$ for odd $r$. The gadget can also be constructed by a clique-width expression using at most this number of labels, by handling each vertex as an individual label.
 \begin{figure}[htbp]
 \centerline{\includegraphics[width=100mm]{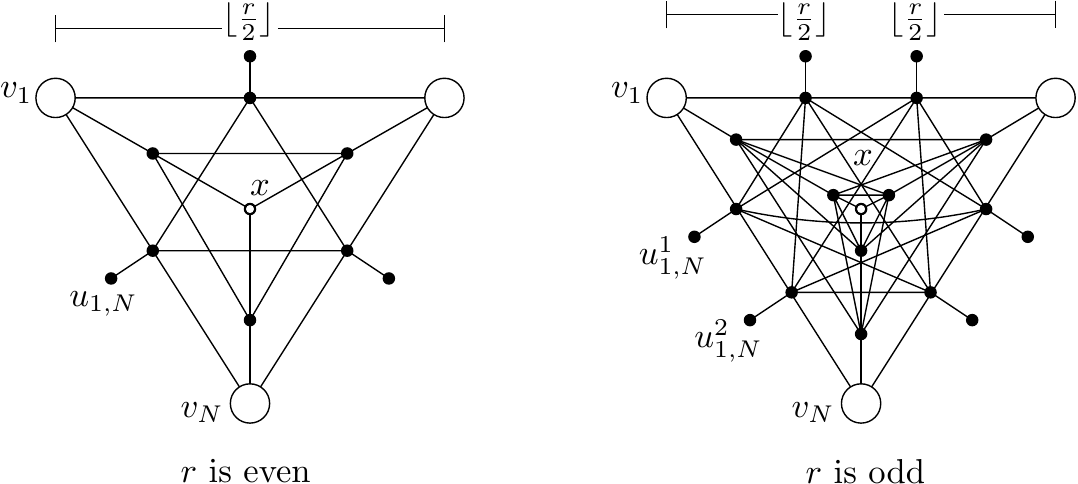}} \caption{A
 general picture of the clique gadget $\hat{X}_N$ for even and odd $r$.}
 \label{fig:clique_gadget} \end{figure}
 \begin{lemma}\label{clique_gadget_proof}
  Any minimum-sized center-set restricted to the vertices of $\hat{X}_N$, will select exactly one of the input vertices $v_i,i\in[1,N]$, to cover all other vertices in $\hat{X}_N$.
 \end{lemma}
  \begin{proof}
  Assume the selection of any single vertex from $\hat{X}_N$ in any minimum-sized center-set: if the selected vertex is on one of the original paths, then vertex $x$ is not
  covered, as its distance is at least $r+1$ from any such vertex (via the
  closest $v_i$), while if the selected vertex, say $w$, is on one of the new
  paths between some $v_i$ and $x$, there will be at least one vertex
  $u_{j,k}^{l}$ on a path between two input vertices $v_j,v_k$ for $l\in[1,2]$
  that is not covered, as the distances from $w$ to $v_j,v_k$ are at least
  $1+\lfloor\frac{r}{2}\rfloor$ and the distances from there to the $u_{j,k}^{l}$
  are also $1+\lfloor\frac{r}{2}\rfloor$. On the other hand, if the selected vertex is one of the inputs $v_i$, then all other vertices are covered: the distance from $v_i$ to any other input $v_j$ or $x$ is exactly $r$, thus all vertices on these paths are covered (including vertices of the type $u_{i,j}^l$), while for vertices on paths not originating at $v_i$, the distance from $v_i$ to some middle vertex on a path adjacent to it is $\lfloor\frac{r}{2}\rfloor$ and the distance from there to any vertex on some other path (or even adjacent to a middle vertex) is $\le\lfloor\frac{r}{2}\rfloor$, giving an overall distance of $\le r$.
 \end{proof}
\subparagraph{Assignment gadget $\hat{U}_N$:} Once more, there are $N$ input
vertices $v_1,\dots,v_N$ for this gadget, while the purpose here is to ensure
that, assuming all input vertices have already been covered, any minimum-sized
center-set will select exactly $N-1$ of them to cover all other vertices in the
gadget.
 \subparagraph{Construction and size/cw bound for $\hat{U}_N$:}  We first connect all input vertices to each other by two
 distinct paths each containing $r$ new vertices, so that all distances between
 any pair of input vertices are exactly $r+1$. Let the vertices on these paths
 between $v_i,v_j$ be $u_{i,j}^{l}$ and $\hat{u}_{i,j}^{l}$ for $l\in[1,r]$.
 Then, for odd $r$, we also attach a path of  $\lfloor\frac{r}{2}\rfloor$
 vertices to the middle vertex of each path, that is, the vertices
 $u_{i,j}^{\lfloor\frac{r}{2}\rfloor+1}$ and
 $\hat{u}_{i,j}^{\lfloor\frac{r}{2}\rfloor+1}$ that are at distance
 $\lfloor\frac{r}{2}\rfloor+1$ from both endpoints $v_i,v_j$ of their paths. We
 call the vertices on these new paths $w_{i,j}^{m}$ and $\hat{w}_{i,j}^{m}$ for
 $m\in[1,\lfloor\frac{r}{2}\rfloor]$. For even $r$, we make a vertex
 $w_{i,j}^{r/2}$ (resp.\ $\hat{w}_{i,j}^{r/2}$) for each path and attach two
 paths of $r/2-1$ vertices to it, naming the vertices on these paths
 $w_{i,j}^{o,m}$ (resp.\ $\hat{w}_{i,j}^{o,m}$) for $o\in[1,2]$ and
 $m\in[1,r/2-1]$, finally attaching the other endpoint vertex $w_{i,j}^{1,1}$
 (resp.\ $\hat{w}_{i,j}^{1,1}$) to $u_{i,j}^{r/2}$ (resp.\
 $\hat{u}_{i,j}^{r/2}$) and also $w_{i,j}^{2,1}$ (resp.\ $\hat{w}_{i,j}^{2,1}$)
 to $u_{i,j}^{r/2+1}$ (resp.\ $\hat{u}_{i,j}^{r/2+1}$), being the vertices at
 distance $r/2$ from one of the two endpoints $v_i$ (and $r/2+1$ from the other
 $v_j$). Thus between any two inputs $v_i,v_j$, there are two vertices
 $w_{i,j}^{\lfloor\frac{r}{2}\rfloor},\hat{w}_{i,j}^{\lfloor\frac{r}{2}\rfloor}$
 at distance exactly $r$ from both. See Figure \ref{fig:assignment_gadget} for an illustration. The size of the gadget is $|\hat{U}_N|=N+2{N \choose 2}\cdot(r+\lfloor\frac{r}{2}\rfloor)$ for odd $r$ and $|\hat{U}_N|=N+2{N \choose 2}\cdot(2r-1)$ for even $r$, while the gadget can also be constructed by a clique-width expression using at most this number of labels, by handling each vertex as an individual label.
 \begin{figure}[htbp]
 \centerline{\includegraphics[width=130mm]{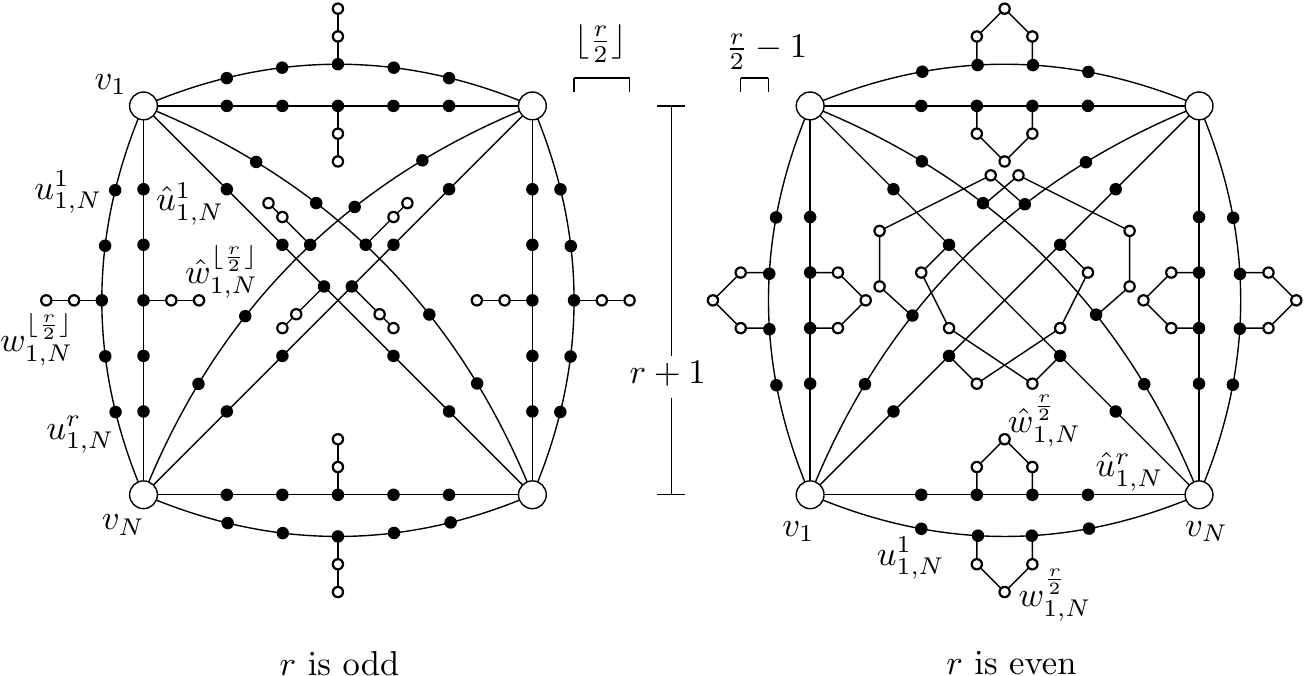}} \caption{A
 general picture of the assignment gadget $\hat{U}_N$ for odd and even $r$.}
 \label{fig:assignment_gadget} \end{figure}
\begin{lemma}\label{assignment_gadget_proof}
 Assuming all input vertices $v_1,\dots,v_N$ need not be covered by this selection, any minimum-sized
center-set restricted to the vertices of $\hat{U}_N$ will select exactly $N-1$ of the input vertices to cover all other vertices in $\hat{U}_N$.
\end{lemma}
 \begin{proof}
  The claim on the gadget's function is shown by induction on $N$: for all $N\ge2$ any minimum-sized center-set will select exactly $N-1$ input vertices, while if any non-input vertex is selected there will be at least $N$ vertices required. The base case is $N=2$ and we have two input vertices $v_1,v_2$ and two paths between them, with selection of either $v_1$ or $v_2$ indeed covering all vertices on these paths. On the other hand, as on each path between $v_1,v_2$ there is a vertex $w_{1,2}^{\lfloor\frac{r}{2}\rfloor}$ or $\hat{w}_{1,2}^{\lfloor\frac{r}{2}\rfloor}$ at distance exactly $r$ from both $v_1,v_2$ (and the distance between them is $2r$), selection of any single vertex on these paths will cover all vertices on its path but not all vertices on the other path and thus at least two selections will be required. For the induction step, assuming the claim holds for $N-1$, we extend it for $N$: the gadget $\hat{U}_N$ is constructed by the gadget $\hat{U}_{N-1}$ by adding a new input vertex $v_N$ and two paths from it to every other input vertex $v_1,\dots,v_{N-1}$. Let $v_i$ with $i\in[1,N-1]$ be the vertex that is not selected from the original $N-1$ input vertices of the $\hat{U}_{N-1}$ gadget. Then in $\hat{U}_N$, all vertices that were already in $\hat{U}_{N-1}$ have been covered, as well as all vertices on the paths between vertices $v_j$ with $j\in[1,N-1],j\neq i$ and $v_{N}$. The only vertices that still need to be covered are the ones on the two paths between $v_i$ and $v_N$. Again, as on each path between $v_i,v_N$ there is a vertex $w_{i,N}^{\lfloor\frac{r}{2}\rfloor}$ or $\hat{w}_{i,N}^{\lfloor\frac{r}{2}\rfloor}$ at distance exactly $r$ from both $v_i,v_N$, any selection of a single vertex from inside these two paths will be insufficient to cover both $w_{i,N}^{\lfloor\frac{r}{2}\rfloor}$ and $\hat{w}_{i,N}^{\lfloor\frac{r}{2}\rfloor}$, while selecting either $v_i$ or $v_N$ indeed covers all vertices.
 \end{proof}
\subparagraph{Block gadget $\hat{G}$:} This gadget is the main building block of our construction and uses the above gadgets as inner components. We first make $p$ pairs of paths $A_1, B_1,\dots,A_p,B_p$ consisting of $2r+1$ vertices each, named $a_i^0,\dots,a_i^{2r}$ and $b_i^0,\dots,b_i^{2r}$ for every pair $A_i,B_i$ with $i\in[1,p]$. We then make two copies of the guard gadget $\hat{T}_N$ for $A_i$, where the $N=2r+1$ inputs are the vertices $a_0,\dots,a_{2r}$ and repeat the same for $B_i$. 
Figure \ref{fig:path_guards} provides an illustration. 
We refer to all vertices in these gadgets $\hat{T}_{2r+1}$ as the \emph{guards}. \\
 \begin{figure}[htbp]
  \centerline{\includegraphics[width=56mm]{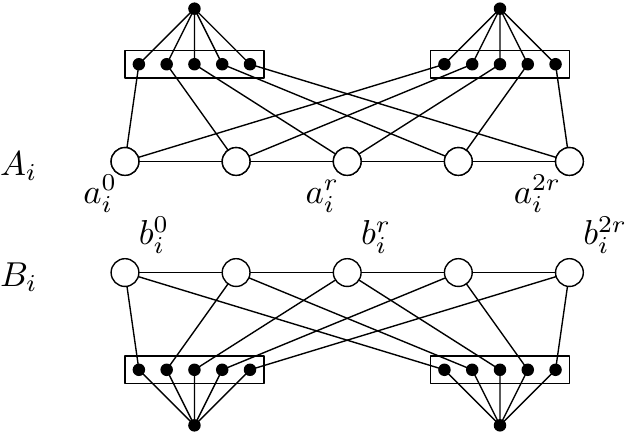}} \caption{The paths
 $A_i,B_i$ along with their attached guard gadgets for $r=2$. Note the boxes
 indicating cliques (for even $r$).} \label{fig:path_guards} \end{figure}
\subparagraph{Canonical pairs:} We next define $3r+1$ \emph{canonical pairs} of numbers $(\alpha_y\in[0,2r],\beta_y\in[0,2r])$, indexed and ordered by $y\in[1,3r+1]$: for $1\le y\le 2r$, the pair is given by $(\alpha_y=\lfloor\frac{y}{2}\rfloor,\beta_y=\lfloor\frac{y}{2}\rfloor)$ if $y$ is odd, while if $y$ is even the pair is given by $(\alpha_y=y/2-1,\beta_y=2r-y/2+1)$. For $2r+1\le y\le 3r+1$ the pair is given by $(\alpha_y=y-r-1,\beta_y=y-r-1)$. As an example, the pairs in the correct order for $r=2$ are (0,0),(0,4),(1,1),(1,3),(2,2),(3,3),(4,4).

In our construction, these pairs will correspond to the indices of the vertices of paths $A_i$ ($\alpha_y$) and $B_i$ ($\beta_y$) that a (canonical) minimum-sized center-set can select. As already mentioned, in the final construction there will be a number of consecutive such pairs of paths $A_i^j,B_i^j$ connected in a path-like manner. We first provide a substitution lemma showing that any minimum-sized center-set does not need to make a selection that does not correspond to one of the canonical pairs followed by a lemma showing that any center-set only making selections based on canonical pairs will have to respect the ordering of the pairs it uses, i.e.\ if some pair with index $y$ is used for selection from paths $A_i^j,B_i^j$, then any pair used in the following paths $A_i^{j+1},B_i^{j+1}$ must be of index $y'\le y$. 
This will form the basis of a crucial argument for showing that any solution will eventually settle into a specific pattern that encodes an assignment without alternations.\footnote{In fact, the subsequent proof of the converse direction (from $(k,r)$-center to assignment, Lemma \ref{cw_SETH_LB_BWD_lem}) does not require the substitution of Lemma \ref{canonical_pair_lemma} given here, as the structure of the graph itself can enforce all choices made to conform to the above canonical pairs. We offer the substitution lemma as further confirmation of the correctness of our approach. }
\begin{lemma}\label{canonical_pair_lemma}
 In a series of $M$ pairs of paths $A_i^j,B_i^j,j\in[1,M]$, where the last vertices of $A_i^j,B_i^j$ are joined with the first vertices of $A_i^{j+1},B_i^{j+1},j\in[1,M-1]$, any center-set $K$ of size $|K|=2m$ that contains one vertex from each path $A_i^j,B_i^j$ can be substituted by a center-set $K'$ of size $|K'|=|K|=2m$, where the indices of each pair of selected vertices is canonical.
\end{lemma}
 \begin{proof}
  First, observe that due to the connection of both final vertices from each pair of paths $A_i^j,B_i^j$ with both first vertices of the following pair $A_i^{j+1},B_i^{j+1}$, any selection $a^l\in A_i^j$ and $b^o\in B_i^j$ is equivalent in terms of the vertices it covers with the opposite selection $a^o\in A_i^j$ and $b^l\in B_i^j$, i.e.\ the same vertices from preceding/succeeding pairs of paths $A_i^{j-1},B_i^{j-1}$/$A_i^{j+1},B_i^{j+1}$ are covered by both selections, while within $A_i^j,B_i^j$ the vertices covered are the opposite. Due to this symmetry, any center-set $K$ in which the index $l\in[0,2r]$ of the selected vertex $a^l\in A_i^j$ is larger than the index $o\in[0,2r]$ of the selected vertex $b^o\in B_i^j$, for any $j\in[1,M]$, can be substituted by a center-set $K'$ of the same size in which $l\le o$ (by replacing one with the other), without affecting the outcome. 

  Given such a center-set, we claim that any pair of selections $a_j^l\in A_i^j,b_j^o\in B_i^j$ can in fact be substituted by a canonical pair, by showing that \emph{all} pairs can be partitioned in equivalence classes, each being represented by a canonical pair: the center-set can alternate between any of the pairs within each class without any change in vertex coverage and we can thus replace any pair of selections by the canonical pair representing the class to which it belongs.

  Consider the canonical pairs for $y\in[1,2r]$ and odd: the pair is given by $(\alpha_y=\lfloor\frac{y}{2}\rfloor,\beta_y=\lfloor\frac{y}{2}\rfloor)$ and we define the corresponding class of non-canonical pairs to include all pairs where $(\alpha_y=\lfloor\frac{y}{2}\rfloor,\lfloor\frac{y}{2}\rfloor\le\beta_y\le2r-\lceil\frac{y}{2}\rceil)$, i.e.\ any pair where $\alpha_y$ is the same as the canonical representative, yet $\beta_y$ can now range from the same value $\lfloor\frac{y}{2}\rfloor$ up to $2r-\lceil\frac{y}{2}\rceil$. To see why any selections within the class are interchangeable, consider the two extreme cases: let $a_j^{\lfloor\frac{y}{2}\rfloor}\in A_i^j, b_j^{\lfloor\frac{y}{2}\rfloor}\in B_i^j$ be a selection followed by $a_{j+1}^{\lfloor\frac{y}{2}\rfloor}\in A_i^{j+1},b_{j+1}^{2r-\lceil\frac{y}{2}\rceil}$ in the subsequent pair of paths. Both selections $a_j^{\lfloor\frac{y}{2}\rfloor},b_j^{\lfloor\frac{y}{2}\rfloor}$ will be at distance $0<d=r-\lfloor\frac{y}{2}\rfloor\le r$ from the middle vertices $a_j^r,b_j^r$ on their paths $A_i^j,B_i^j$ and all vertices $a_j^{\lfloor\frac{y}{2}\rfloor+1},b_j^{\lfloor\frac{y}{2}\rfloor+1},\dots,a_j^{2r-d},b_j^{2r-d}$ from the same paths will be covered by these selections. As the selection from $A_i^{j}$ always matches the one from $A_i^{j+1}$, the remaining vertices on these paths will be covered by the subsequent selection $a_{j+1}^{\lfloor\frac{y}{2}\rfloor}\in A_i^{j+1}$, as well as all vertices on path $B_i^{j+1}$ up to position $r-\lfloor\frac{y}{2}\rfloor-2$, meaning the selection from this path can be up to distance $r+1$ from this ``last'' covered vertex, giving the index of the furthest possible choice from $B_i^{j+1}$ as $2r-\lfloor\frac{y}{2}\rfloor-1=2r-\lceil\frac{y}{2}\rceil$, being exactly the extremal case for this class. Observe also that this selection will not cover more vertices of the subsequent paths $A_i^{j+2},B_i^{j+2}$ as it can reach up to vertices at position $r-\lfloor\frac{y}{2}\rfloor-2$ in both these paths, which are exactly already covered by the selection of $a_{j+2}^{\lfloor\frac{y}{2}\rfloor}\in A_i^{j+2}$, meaning any other intermediate selections would indeed produce the same result as well.

  Next, consider the canonical pairs for $y\in[2r+2,3r+1]$: the pair is given by $(\alpha_y=y-r-1,\beta_y=y-r-1)$ and we define the corresponding class of non-canonical pairs to include all pairs where $(3r+2-y\le\alpha_y\le y-r-1,\beta_y=y-r-1)$, i.e.\ any pair where $\beta_y$ is the same as the canonical representative, yet now $\alpha_y$ can range from the same value $y-r-1$ down to $3r+2-y$. To see why any selections within the class are interchangeable, consider the two extreme cases, as before: let $a_j^{y-r-1}\in A_i^j,b_j^{y-r-1}\in B_i^j$ be a selection followed by $a_{j+1}^{3r+2-y}\in A_i^{j+1}, b_{j+1}^{y-r-1}\in B_i^{j+1}$ in the subsequent pair of paths. Both selections $a_j^{y-r-1},b_j^{y-r-1}$ will be at distance $0\le d=3r+1-y\le r$ from the final vertices $a_j^{2r},b_j^{2r}$ on their paths $A_i^j,B_i^j$ and all vertices $a_{j+1}^0,b_{j+1}^0,\dots,a_{j+1}^{r-d-1},b_{j+1}^{r-d-1}$ from the following paths $A_i^{j+1},B_i^{j+1}$ will be covered by these selections. As the selection from $B_i^j$ always matches the one from $B_i^{j+1}$, all the remaining vertices of $B_i^{j+1}$ are covered, as well as vertices $a_{j+1}^{d'},\dots,a_{j+1}^{2r}$ from $A_i^{j+1}$, where $d'=(y-r-1)-r+2d+2=4r+3-y$ is the maximum distance in $A_i^{j+1}$ that selection $b_{j+1}^{y-r-1}$ can reach, meaning the selection from this path $A_i^{j+1}$ can be up to distance $r+1$ from this ``first'' covered vertex ($a_{j+1}^{4r+3-y}$), giving the index of the nearest possible choice from $A_i^{j+1}$ as $4r+3-y-(r+1)=3r+2-y$, being exactly the extremal case for this class.

  The final observation required for the claim to be shown is that indeed all possible pairs $(l,o)\in[0,2r]^2$, where $l\le o$, either exactly match some canonical pair, or are contained in one of the classes given above. As any opposite selections are symmetrical and within each class the actual selections are interchangeable, any arbitrary center-set $K$ can be substituted by a center-set $K'$ of the same size, where all indices of each pair of selections is canonical.
  \end{proof}
\begin{lemma}\label{order_pair_lemma}
 In a series of $M$ pairs of paths $A_i^j,B_i^j,j\in[1,M]$, where the last vertices of $A_i^j,B_i^j$ are joined with the first vertices of $A_i^{j+1},B_i^{j+1}$ for $j\in[1,M-1]$, in any center-set $K$ that only selects vertices whose indices correspond to canonical pairs, the index $y$ of any canonical pair selected in some pair of paths $A_i^j,B_i^j$ must be larger than, or equal to the index $y'$ of any canonical pair selected in its following pair of paths $A_i^{j+1},B_i^{j+1}$. 
\end{lemma}
 \begin{proof}
  Consider two consecutive pairs of paths $A_i^j,B_i^j$ and $A_i^{j+1},B_i^{j+1}$ and let $a_j^{\alpha_y},b_j^{\beta_y}$ and $a_{j+1}^{\alpha_{y'}},b_{j+1}^{\beta_{y'}}$ be the selections from $A_i^j,B_i^j$ and $A_i^{j+1},B_i^{j+1}$, respectively, with $y,y'\in[1,3r+1]$.

  First, for any pair $(\alpha_y=y-r-1,\beta_y=y-r-1)$, where $y\in[2r+1,3r+1]$, both selections $a_j^{\alpha_y},b_j^{\beta_y}$ will be at distance $0\le d=3r+1-y\le r$ from the final vertices $a_j^{2r},b_j^{2r}$ on their paths $A_i^j,B_i^j$ and all vertices $a_{j+1}^0,b_{j+1}^0,\dots,a_{j+1}^{r-d-1},b_{j+1}^{r-d-1}$ from the following paths $A_i^{j+1},B_i^{j+1}$ will be covered by these selections. Thus for both paths $A_i^{j+1},B_i^{j+1}$, the corresponding distance $d'$ of the selections $a_{j+1}^{\alpha_{y'}},b_{j+1}^{\beta_{y'}}$ from the final vertices $a_{j+1}^{2r},b_{j+1}^{2r}$ will be the same for both and at least equal to $d$, which in turn implies both $\alpha_{y'}\le\alpha_y$ and $\beta_{y'}\le\beta_y$, that gives $y'\le y$ (note that $y'$ can be within $[1,2r]$ as long as both inequalities hold).

  Next, for any pair $(\alpha_y=\lfloor\frac{y}{2}\rfloor,\beta_y=\lfloor\frac{y}{2}\rfloor)$, with $y\in[1,2r]$ and odd, both selections $a_j^{\alpha_y},b_j^{\beta_y}$ will be at distance $0< d=r-\lfloor\frac{y}{2}\rfloor\le r$ from the middle vertices $a_j^{r},b_j^{r}$ on their paths $A_i^j,B_i^j$ and all vertices $a_j^{\alpha_y+1},b_j^{\beta_y+1},\dots,a_j^{2r-d},b_j^{2r-d}$ from the same paths will be covered by these selections. Thus in at least one of the following paths $A_i^{j+1},B_i^{j+1}$ there must be some selection $a_{j+1}^{\alpha_{y'}}$ or $b_{j+1}^{\beta_{y'}}$ at distance $d'$, that is at least equal to $d$, from either $a_{j+1}^{r}$ or $b_{j+1}^{r}$. As for all pairs $(\alpha_y,\beta_y)$ it is always $\beta_y\ge\alpha_y$, if this selection is from $B_i^{j+1}$, then it is $\beta_{y'}\le\beta_y$ and $\alpha_{y'}\le\alpha_y$ which gives $y'\le y$, while if this selection is from $A_i^{j+1}$, we have $\alpha_{y'}\le\alpha_y$, which means either also $\beta_{y'}\le\beta_y$ and thus $y'\le y$, or $\beta_{y'}=2r-y'/2+1$ and $y'=y+1$. In this case, observe that $b_{j+1}^{\beta_{y'}}$ covers all vertices $b_{j+1}^{r-y'/2+1},\dots,b_{j+1}^{2r-y'/2}$, while $a_{j+1}^{\alpha_{y'}}$ can cover all vertices from $b_{j+1}^{0}$ (for $d'>1$) up to $b_{j+1}^{r-y'/2-1}$ from $B_i^{j+1}$, thus leaving vertex $b_{j+1}^{r-y'/2}$ at distance $>r$ from any selected vertex.
 
  Finally, for even $y\in[1,2r]$ and any pair $(\alpha_y=y/2-1,\beta_y=2r-y/2+1)$, selected vertex $a_j^{\alpha_y}$ covers all vertices $a_j^{\alpha_y+1},\dots,a_j^{\alpha_y+r}$, while selected vertex $b_j^{\beta_y}$ can only cover all vertices from $a_j^{\alpha_y+r+2}$ to $a_j^{2r}$ (for $y<2r$), thus requiring at least one selection from the following pair of paths $A_i^{j+1},B_i^{j+1}$, at distance at most $y/2-1$ from the first vertex on its path $a_{j+1}^0$ or $b_{j+1}^0$. In either case, we have $\alpha_{y'}\le\alpha_y$ and thus also $y'\le y$.
 \end{proof}
Now we can continue the description of the block gadget~$\hat{G}$. For each pair of paths $A_i,B_i$ with $i\in[1,p]$ we make $3r+1$ vertices~$u_i^y$, for $y\in[1,3r+1]$, that we connect to vertices~$a_i^{\alpha_y}$ and~$b_i^{\beta_y}$ by paths of length $r+1$, i.e.\ each vertex~$u_i^y$ is at distance $r+1$ from the vertices in $A_i,B_i$, whose indices match the numbers in the pair corresponding to its own index~$y$ (via one path for one such vertex in~$A_i$ or~$B_i$). Let the~$r$ intermediate vertices on the path from each~$u_i^y$ to some vertex in~$A_i$ be called $v_i^{y,1}\dots,v_i^{y,r}$ and the~$r$ intermediate vertices on the other path to some vertex in $B_i$ be called $v_i^{y,r+1}\dots,v_i^{y,2r}$.

Next, we add another vertex $q_i$ that we attach to all vertices $u_i^y$ by paths of length $r-1$ (making the distances between them and $q_i$ equal to $r$) and then we also attach $3r+1$ paths of length $r$ to $q_i$, naming the $(2r-1)(3r+1)$ vertices on these paths $q_i^{1},\dots,q_i^{(2r-1)(3r+1)}$. Let $U_i$ be the set of all $u_i^y$ vertices for all $y\in[1,3r+1]$ and $U$ be the union of all $U_i$ for $i\in[1,p]$. We then make use of the assignment gadget $\hat{U}_N$ described above by making a copy of $\hat{U}_N$ for each $U_i$, where the $N=3r+1$ inputs are identified with the vertices $u_i^y$.

Then, for every set $S\subset U$ that contains exactly one vertex from each~$U_i$ (the number of such sets being $(3r+1)^p$) we make a vertex~$x_S$. Here we make use of the clique gadget~$\hat{X}_N$ described above, where these~$x_S$ vertices act as inputs and $N=(3r+1)^p$. Let~$X$ be the set containing all vertices in the gadget, including all~$x_S$ vertices. Then, for every vertex~$x_S$ we make a copy of the guard gadget~$\hat{T}_N$, where the $N=p+1$ inputs are~$x_S$ and the vertices in~$U$ for which $u_i^y\in S$ (one from each~$U_i$). %See Figure \ref{fig:block_gadget} in Appendix \ref{append_cw_LB} for an illustration. 
This concludes the construction of the block gadget~$\hat{G}$. 
 \begin{figure}[htbp]
  \centerline{\includegraphics[width=110mm]{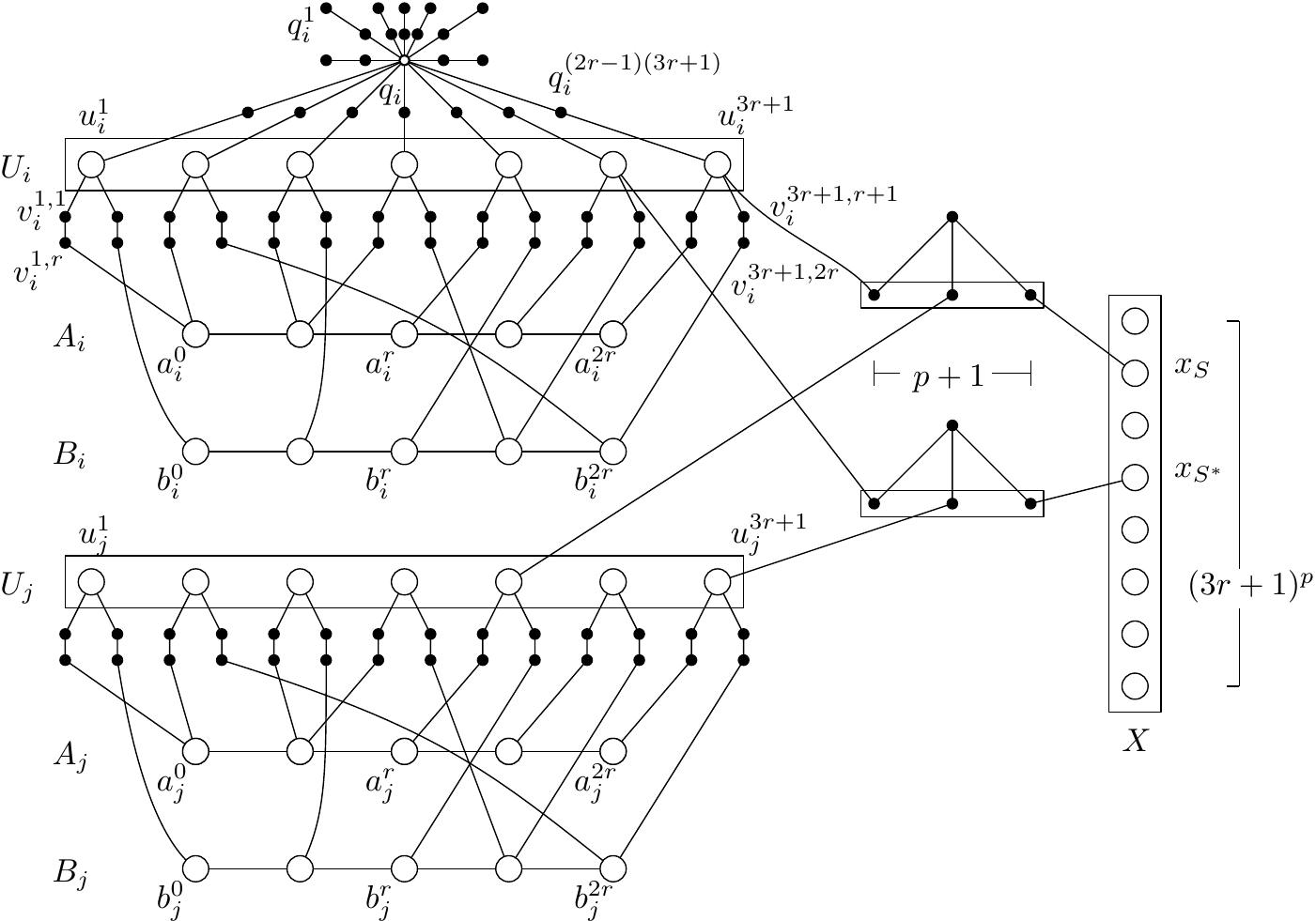}}
  \caption{Inside a block gadget $\hat{G}$: two pairs of paths $A_i,B_i,A_j,B_j$ with attached $U_i,U_j$, clique gadget $X$ and guard gadgets between $U$ and $X$, while guard vertices and all $q_j$ are omitted for clarity. Note boxes around vertices in $U_i,U_j,X$ indicate a gadget, while boxes around the $p+1$ vertices in the guard gadgets indicate cliques (for even $r$).}
  \label{fig:block_gadget}
 \end{figure}
 \begin{figure}[htbp]
 \centerline{\includegraphics[width=100mm]{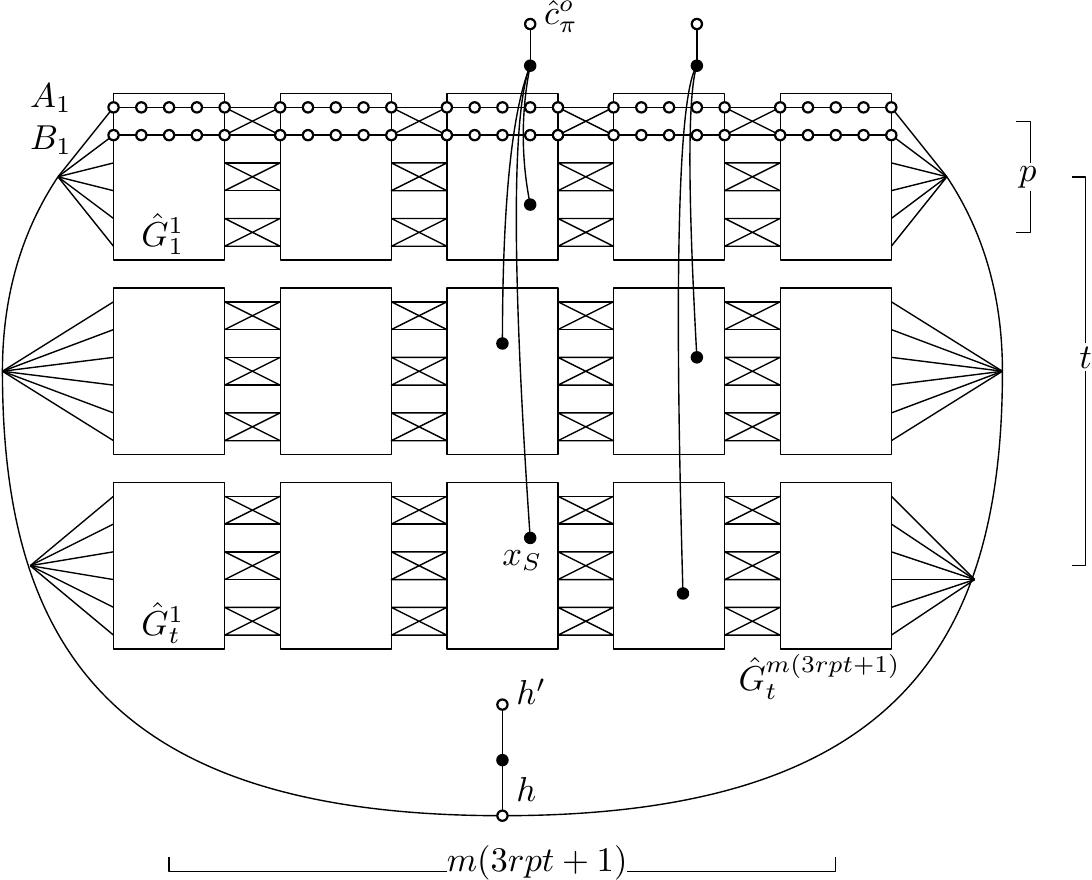}}
 \caption{A simplified picture of the complete construction. Note boxes indicate block gadgets $\hat{G}$, while there is no vertex anywhere between $h$ and the first/last vertices of the long paths.}
 \label{fig:global_construction}
\end{figure}
\subparagraph{Global construction:} Graph $G$ is then constructed as follows. For every group $F_{\tau}$ of variables of $\phi$ with $\tau\in[1,t]$, we make $m(3r p t+1)$ copies of the gadget $\hat{G}$ that we call $\hat{G}_{\tau}^{\mu}$ for $1\le \mu\le m(3r p t+1)$. Then, for each $\tau\in[1,t]$ we connect gadgets $\hat{G}_{\tau}^1,\hat{G}_{\tau}^2,\dots,\hat{G}_{\tau}^{m(3r p t+1)}$ in a path-like manner: for every $1\le\mu<m(3r p t+1)$ and $1\le i\le p$, we connect both vertices $a_i^{2r}$, $b_i^{2r}$ in $\hat{G}_{\tau}^{\mu}$ to both vertices $a_i^0$, $b_i^0$ in $\hat{G}_{\tau}^{\mu+1}$. We then make another vertex $h$ that we connect to all vertices $a_i^1,b_i^1$ in $\hat{G}_{\tau}^1$ for all $i\in[1,p]$ and $\tau\in[1,t]$, as well as to all vertices $a_i^{2r},b_i^{2r}$ in $\hat{G}_{\tau}^{m(3r p t+1)}$ for all $i\in[1,p]$ and $\tau\in[1,t]$, i.e.\ to all the first and last vertices on the long paths created upon connecting gadgets $\hat{G}_{\tau}^1,\hat{G}_{\tau}^2,\dots,\hat{G}_{\tau}^{m(3r p t+1)}$. We also attach a path of length $r$ to vertex $h$, the final vertex of this path named $h'$.

Next, for every $1\le\tau\le t$ we associate a set $S\subset U$ that contains exactly one vertex from each $U_i$ with an assignment to the variables in group $F_{\tau}$ and as there are at most $2^{\gamma}=2^{\lfloor\log_2(3r+1)^p\rfloor}$ assignments to the variables in $F_{\tau}$ and $(3r+1)^p\ge2^{\gamma}$ sets $S$, the association can be unique for each $\tau$. For each clause $C_{\pi}$ of $\phi$ with $\pi\in[1,m]$, we make $3r p t+1$ vertices $\hat{c}_{\pi}^o$ for $0\le o<3r p t+1$. Then, to each $\hat{c}_{\pi}^o$ we attach a path of length $r-1$ and we consider every assignment to the variables of group $F_{\tau}$ for every $\tau\in[1,t]$ that satisfies the clause $C_{\pi}$: a vertex $x_S$ in $\hat{G}_{\tau}^{m o+\pi}$, for every $0\le o<3rpt+1$, is adjacent to the endpoint of this path (thus being at distance $r$ from $\hat{c}_{\pi}^o$), where $S$ is the subset associated with this assignment within the gadget. 
This concludes our construction, while Figure \ref{fig:global_construction} provides an illustration.
\begin{lemma}\label{cw_SETH_LB_FWD_lem}
  If $\phi$ has a satisfying assignment, then $G$ has a $(k,r)$-center of size $k=((3r+3)p+1)m(3r p t+1)t+1$.
 \end{lemma}
 \begin{proof}
  Given a satisfying assignment for $\phi$ we show the existence of a $(k,r)$-center $K$ of $G$ of size $|K|=k=((3r+3)p+1)m(3r p t+1)t+1$. Set $K$ will include the vertex $h$ and $(3r+3)p+1$ vertices from each gadget $\hat{G}_{\tau}^{\mu}$, for $\tau\in[1,t]$ and $\mu\in[1,m(3rpt+1)]$. For each group $F_{\tau}$ of variables we consider the restriction of the assignment for $\phi$ to these variables and identify the set $S$ associated with this restricted assignment. We first add vertex $x_S$ to $K$ and then, for every $i\in[1,p]$, we also add vertex $q_i$ and all $u_i^z\in U_i\setminus S$, as well as the two vertices $a_i^{\alpha_y}$ and $b_i^{\beta_y}$ from the paths $A_i,B_i$ that are connected to the only $u_i^y\in S$ for this $i$ by paths of length $r+1$, i.e.\ the vertices whose indices $\alpha_y,\beta_y\in[0,2r]$ correspond to the pair $(\alpha_y,\beta_y)$ for the index $y\in[1,3r+1]$ of this vertex. In total, we have $3r$ vertices $u_i^y$, plus the three $a_i^{\alpha_y}$, $b_i^{\beta_y}$ and $q_i$ for each $i\in[1,p]$, with the addition of vertex $x_S$ completing the selection within each $\hat{G}_{\tau}^{\mu}$. Repeating the above for all $m(3rpt+1)t$ gadgets completes the selection for $G$ and what remains is to show that $K$ is indeed a $(k,r)$-center of $G$.
 
  First, our selection of one vertex from each path $A_i,B_i$ ensures that all guard vertices are within distance $r$ from some selected vertex. Next, for each $\tau\in[1,t]$ and $i\in[1,p]$, consider the ``long paths'' formed by joining both vertices $a_i^{2r},b_i^{2r}$ in $\hat{G}_{\tau}^{\mu}$ to both $a_i^0,b_i^0$ in $\hat{G}_{\tau}^{\mu+1}$. Since the associations between sets $S$ and partial assignments to variables of $F_{\tau}$ are consistent for each $\tau$, the patterns of selection of vertices from each $A_i$ and $B_i$ are repeating, i.e.\ $K$ contains every $(2r+1)$-th vertex on each path $A_i$ and $B_i$, meaning all vertices on these paths are within distance $r$ from some selected vertex, apart from the first/last $r-1$ depending on the actual selection pattern, yet these vertices are within distance $r$ from selected vertex $h$.
 
  Next, within each gadget $\hat{G}_{\tau}^{\mu}$, our selection of vertices $q_i$ for every $i\in[1,p]$ brings all $q_i^{1},\dots,q_i^{(2r-1)(3r+1)}$ and $u_i^j$ vertices within distance $r$ from $K$. Further, our selection of $a_i^{\alpha_y},b_i^{\beta_y}$ covers vertices $v_i^{y,1}\dots,v_i^{y,2r}$ on the two paths from $a_i^{\alpha_y},b_i^{\beta_y}$ to $u_i^y\in S$, while all other $v_i^{z,w}$ vertices are covered by our selection of each $u_i^z\in U_i\setminus S$. This selection of all $u_i^z\in U_i\setminus S$ also covers all vertices in the assignment gadgets as well as the guard gadgets $\hat{T}_p$ between the $u_i^z$ and $x_{S'}$ for all $S'\neq S$, while selection of $x_S$ covers all vertices in $X$ and all vertices in the guard gadget where $x_S$ is an input.
 
  Finally, concerning the clause vertices $\hat{c}_{\pi}^{o}$ for $\pi\in[1,m]$ and $o\in[0,3rpt]$, observe that if the given assignment for $\phi$ satisfies the clause $C_{\pi}$, there will be some literal contained therein that is set to true, that corresponds to a variable in a group $F_{\tau}$ for some $\tau\in[1,t]$ and a matching partial assignment for the variables in $F_{\tau}$ associated with some set $S^*$. Our set $K$ contains vertex $x_{S^*}$ in $\hat{G}_{\tau}^{mo+\pi}$ for every $o\in[0,3rpt]$ and vertices $\hat{c}_{\pi}^{o}$ are within distance $r$ from $x_{S^*}$, through a path whose vertices are also covered by $x_{S^*}$.
 \end{proof}
 \begin{lemma}\label{cw_SETH_LB_BWD_lem}
  If $G$ has a $(k,r)$-center of size $k=((3r+3)p+1)m(3r p t+1)t+1$, then $\phi$ has a satisfying assignment.
 \end{lemma}
 \begin{proof}
  Given a $(k,r)$-center $K$ of $G$ of size $|K|=k=((3r+3)p+1)m(3r p t+1)t+1$, we show the existence of a satisfying assignment for $\phi$. First, observe that $K$ must contain vertex $h$, as all vertices on the path attached to it and $h'$ must be within distance $r$ from $K$. Next we require an averaging argument: as the remaining number of vertices in $K$ is $((3r+3)p+1)m(3rpt+1)t$ and there are $m(3rpt+1)t$ gadgets $\hat{G}_{\tau}^{\mu}$, if there are more than $(3r+3)p+1$ vertices selected from some block gadget, then there will be less than these selected from some other block gadget. We will show that not all vertices within a block gadget can be covered by less than $(3r+3)p+1$ selected vertices, which implies that exactly this number is selected from each block gadget in any center-set of size $k$.
 
  Consider a gadget $\hat{G}_{\tau}^{\mu}$. First, observe that at least one of the $x_S$ vertices must be selected to cover all vertices in $X$ and that any such selection indeed covers all vertices in $X$, as well as all vertices in the guard gadget of which it is an input. This leaves $(3r+3)p$ vertices to cover all other $p$ groups of vertices in the gadget. Observe also that for every $i\in[1,p]$, any minimum-sized center-set must contain one vertex from each path $A_i,B_i$ to cover all the guards, as well as $q_i$ to cover $q_i^{1},\dots,q_i^{(2r-1)(3r+1)}$, as any single selection of some guard vertex will not be sufficient to cover all other guard vertices attached to the same path. This leaves $3rp$ vertices to cover the $v_i^{y,w}$ vertices on the paths between the $a_i^j,b_i^l$ that have not been selected from each pair $A_i,B_i$, as well as all vertices in the guard gadgets in which each $x_{S^*}$ that was not selected is an input. Due to the structure of the assignment gadgets $\hat{U}$, there must be $3r$ vertices selected from each $U_i$, thus completing the set $K$ (and the averaging argument). We then claim that the selections from the paths $A_i,B_i$ must match (complement) these selections from $U_i$, that in turn must match (also complement) the selection of $x_S$ from $X$.

  First, suppose that for some $i\in[1,p]$ the selections $a_i^j,b_i^l$ with $j,l\in[0,2r]$ from $A_i,B_i$ do not correspond to some canonical pair for some $y\in[1,3r+1]$ with $(\alpha_y=j,\beta_y=l)$.\footnote{In fact, Lemma \ref{canonical_pair_lemma} already shows that any $(k,r)$-center that does not make selections based only on canonical pairs can always be substituted for a $(k,r)$-center that does. Nevertheless, we show here that this requirement is also enforced by the structure of the graph, mostly for completeness.} Vertex $a_i^j$ will cover at most $2r$ vertices $v_i^{z,1},\dots,v_i^{z,r}$ and $v_i^{z',1},\dots,v_i^{z',r}$, for $z\neq z'$ (if index $j$ happens to be included in two pairs and only the first $r$ for inclusion in a single pair), but not vertices $v_i^{z,r+1}$ or $v_i^{z',r+1}$. Similarly, vertex $b_i^l$ will also cover at most $2r$ vertices $v_i^{w,r+1},\dots,v_i^{w,2r}$ and $v_i^{w',r+1},\dots,v_i^{w',2r}$, again for $w\neq w'$, but not vertices $v_i^{w,1}$ or $v_i^{w',1}$. Note that since the two choices do not correspond to some canonical pair all these indices will be different: $z\neq z'\neq w\neq w'$. Now, as there are no more selections from $A_i,B_i$, all vertices $v_i^{y',1}$ and $v_i^{y',r+1}$ are definitely not covered for $y'\neq z,z',w,w'$ (due to any adjacent selections being at distance at least $r+1$). In short, there is at least one vertex $v_i^{y,1}$ or $v_i^{y,r+1}$ (or both) that is not covered for each $y\in[1,3r+1]$. Since the number of selections from $U_i$ is $3r$ and no other selection would reach these vertices (e.g.\ from some other $U_{i'}$, due to the guard gadgets employed anywhere in between sets $U$ and $X$), there will be at least one such vertex that is not covered, implying the selections from every $A_i,B_i$ must match some canonical pair $y$ and the $3r$ selections from $U_i$ must complement this $y$. Further, suppose the selection $x_S$ from $X$, for $S=\{u_1^{y_1},\dots,u_p^{y_p}\}$, does not match the $3r$ selections from each $U_i$, i.e.\ that for some $i\in[1,p]$, it is $u_i^{y_i}\in K\cap U_i$ and $u_i^{z}\notin K\cap U_i$ for all $z\neq y_i$. Then for set $S^*=S\setminus\{u_i^{y_i}\}\cup\{u_i^z\}$ we have that $x_{S^*}\notin K$ and also $S^*\not\subset K$. This means all vertices in the guard gadget attached to $x_{S^*}$ are not covered.
 
 Next, we require that there exists at least one $o\in[0,3rpt]$ for every $\tau\in[1,t]$ for which $K\cap \{\bigcup_{i\in[1,p]} A_i\cup B_i\}$ is the same in all gadgets $\hat{G}_{\tau}^{mo+\pi}$ with $\pi\in[1,m]$, i.e.\ that there exists a number of successive copies of the gadget for which the pattern of selection of vertices from the paths $A_i,B_i$ does not change. As noted above, set $K$ must contain two vertices $a_i^{\alpha_y},b_i^{\beta_y}$ from each $A_i$ and $B_i$, such that the indices $\alpha_y,\beta_y$ of these two selections match the pair corresponding to the index $y$ of some $u_i^y$. Consider the ``long paths'' consisting of paths $A_i,B_i$ sequentially joined with their followers in the next gadget on the same row. Depending on the starting selection, observe that the pattern can ``shift towards the left'' a number of times in each pair of paths $A_i,B_i$, as the first and last $r-1$ vertices will be covered by $h$. That is, a pattern can be selected on some pair $A_i,B_i$ within some gadget $\hat{G}_{\tau}^{\mu}$ and a different pattern can be selected on the pair $A_i,B_i$ following it in gadget $\hat{G}_{\tau}^{\mu+1}$, without affecting whether all vertices on the long paths are covered. As shown by Lemma \ref{order_pair_lemma}, this can only happen if the  index $y'$ that gives the pair of indices $(\alpha_{y'},\beta_{y'})$ of the second pattern is smaller than or equal to the index $y$ that gives the pair of indices $(\alpha_y,\beta_y)$ of the first pattern, or $y'\le y$.
 
 As there are $3r+1$ different indices $y$ and pairs $(\alpha_y,\beta_y)$, the ``shift to the left'' can happen at most $3r$ times for each $i\in[1,p]$, thus at most $3rp$ times for each $\tau\in[1,t]$, or $3rpt$ times over all $\tau$. By the pigeonhole principle, there must thus exist an $o\in[0,3rpt]$ such that no such shift happens among the gadgets $\hat{G}_{\tau}^{mo+\pi}$, for all $\tau\in[1,t]$ and $\pi\in[1,m]$.
 
 Our assignment for $\phi$ is then given by the selections for $K$ in each gadget $\hat{G}_{\tau}^{mo+1}$ for this $o$: for every group $F_{\tau}$ we consider the selection of $x_S\in X$ that corresponds to a set $S\subset U$, that in turn is associated with a partial assignment for the variables in $F_{\tau}$. In this way we get an assignment to all the variables of $\phi$. To see why this also satisfies every clause $C_{\pi}$ with $\pi\in[1,m]$, consider clause vertex $\hat{c}_{\pi}^o$: this vertex is at distance $r$ from some selected vertex $x_S$ in some gadget $\hat{G}_{\tau}^{mo+\pi}$. Since the pattern for selection from paths $A_i,B_i$ remains the same in all gadgets $\hat{G}_{\tau}^{mo+1},\dots,\hat{G}_{\tau}^{mo+\pi}$, so does the set $U$ and also selection of vertices $x_S$, giving the same assignment for the variables of $F_{\tau}$ associated with $S$. 
\end{proof}
\begin{lemma}\label{cw_SETH_LB_cw_bound}
 Graph $G$ has clique-width $\textrm{cw}(G)\le tp+f(r,\epsilon)$, for $f(r,\epsilon)=O(r^p)$.
\end{lemma}
\begin{proof}
 We show how to construct graph $G$ using the clique-width operations introduce, join, relabel and at most $f(r,\epsilon)$ labels. We first introduce vertex $h$ and all vertices on the path of length $r$ from it to $h'$, using one label for each vertex, then consecutively join labels/vertices to form the path. The construction will then proceed in a vertical manner, successively constructing the gadgets $\hat{G}_{\tau}^{\mu}$ for each $\tau\in[1,t]$, before proceeding to repeat the process for each of the $m(3rpt+1)$ columns.
 
 To construct gadget $\hat{G}_{\tau}^{\mu}$ we do the following: we first introduce all guard vertices using one label for each vertex, subsequently applying the appropriate join operations. We do the same for all vertices $q_i,q_i^{1},\dots,q_i^{(2r-1)(3r+1)}$, as well as vertices $u_i^y$ in $U_i$ and all vertices in gadgets $\hat{U}_{3r+1}$ and $v_i^{y,1}\dots,v_i^{y,2r}$, along with all vertices in $X$ and the guard gadgets $\hat{T}$ attached to each $x_S$. We have also introduced the clause vertex $\hat{c}_{\pi}^{o}$ that corresponds to this column and all vertices on the path attached to it, the endpoint of which we now join with matching vertices $x_S$ (if any, for this $F_{\tau}$). We have thus created all vertices (and appropriate edges) in gadget $\hat{G}_{\tau}^{\mu}$, apart from the paths $A_i,B_i$ for $i\in[1,p]$, using one label per vertex. This accounts for the $f(r,\epsilon)$ labels, where function $f$ is $O(r^p)$.
 
 Now, for $i=1$, we introduce vertices $a_1^j,b_1^j$ in turn for each $j\in[0,2r]$, using one label for each vertex and appropriately join each with its previous in the path (for $j\ge2$), the endpoints of the guard gadgets, as well as the corresponding vertices $v_1^{y,w}$, completing the construction of paths $A_1,B_1$. We then relabel the vertex $b_1^{2r}$ with the same label as $a_1^{2r}$ (the last vertices of $A_1,B_1$) and repeat the process for $i=2,\dots,p$. When the above has been carried out for all $i$, the construction of gadget $\hat{G}_{\tau}^{\mu}$ has been completed and we can relabel all vertices to some ``junk'' label, apart from the two final vertices $a_i^{2r},b_i^{2r}$ for each $i\in[1,p]$ (that have the same label) and the endpoint of the path of length $r-1$ attached to clause vertex $\hat{c}_{\pi}^{o}$. This relabelling with some junk label will enable us to reuse the same labels when constructing the same parts of other gadgets. We then repeat the above for the following gadget in the column, until the column is fully constructed. When the column has thus been constructed, we can also relabel the endpoint of the path from clause vertex $\hat{c}_{\pi}^{o}$ to the junk label, thus reusing its former label for the endpoint of the path attached to the clause vertex of the subsequent column.
 
 During construction of the following column, the first vertices $a_i^0,b_i^0$ on each path $A_i,B_i$ will both be joined with the label that includes the last vertices $a_i^{2r},b_i^{2r}$ of the previous column's corresponding path $A'_i,B'_i$ and after each such join, we again relabel these two vertices with the junk label. The construction proceeds in this way until graph $G$ has been fully constructed. Note that during construction of the first and last columns, the first/last vertices of each path have also been joined with vertex $h$.
 
 In total, the number of labels used simultaneously by the above procedure are the $f(r,\epsilon)$ labels used each time for repeating constructions (also counting the constant number of ``outside'' labels for $h$, the clause vertices and the junk label), plus one label for each pair of paths $A_i,B_i$ (containing the last vertices of these paths) in each of the $t$ rows of the construction. The number of these being $tp$, the claimed bound follows. 
\end{proof}
\begin{theorem}\label{cw_SETH_LB}
 For any fixed $r\ge1$, if \textsc{$(k,r)$-Center} can be solved in $O^*((3r+1-\epsilon)^{\textrm{cw}(G)})$ time for some $\epsilon>0$, then \textsc{SAT} can be solved in $O^*((2-\delta)^n)$ time for some $\delta>0$. 
\end{theorem}
\begin{proof}
 Assuming the existence of some algorithm of running time $O^*((3r+1-\epsilon)^{\textrm{cw}(G)})=O^*((3r+1)^{\lambda\textrm{cw}(G)})$ for \textsc{$(k,r)$-Center}, where $\lambda=\log_{3r+1}(3r+1-\epsilon)$, we construct an instance of \textsc{$(k,r)$-Center}, given a formula $\phi$ of \textsc{SAT}, using the above construction and then solve the problem using the  $O^*((3r+1-\epsilon)^{\textrm{cw}(G)})$-time algorithm. Correctness is given by Lemma \ref{cw_SETH_LB_FWD_lem} and Lemma \ref{cw_SETH_LB_BWD_lem}, while Lemma \ref{cw_SETH_LB_cw_bound} gives the upper bound on the running time:  
 \begin{align}
  O^*((3r+1)^{\lambda\textrm{cw}(G)})&\le O^*\left((3r+1)^{\lambda(tp+f(r,\epsilon))}\right)\\
  &\le O^*\left((3r+1)^{\lambda p\left\lceil\dfrac{n}{\lfloor\log_2(3r+1)^p\rfloor}\right\rceil}\right)\label{cw_comp_2}\\
  &\le O^*\left((3r+1)^{\lambda p\dfrac{n}{\lfloor\log_2(3r+1)^p\rfloor}+\lambda p}\right)\\
  &\le O^*\left((3r+1)^{\lambda\dfrac{np}{\lfloor p\log_{2}(3r+1)\rfloor}}\right)\label{cw_comp_4}\\
  &\le O^*\left((3r+1)^{\delta'\dfrac{n}{\log_2(3r+1)}}\right)\label{cw_comp_5}\\
  &\le O^*(2^{\delta''n})=O((2-\delta)^n)
 \end{align} 
for some $\delta,\delta',\delta''<1$. Observe that in line (\ref{cw_comp_2}) the function $f(r,\epsilon)$ is considered constant, as is $\lambda p$ in line (\ref{cw_comp_4}), while in line (\ref{cw_comp_5}) we used the fact that there always exists a $\delta'<1$ such that $\lambda\dfrac{p}{\lfloor p\log_2(3r+1)\rfloor}=\dfrac{\delta'}{\log_2(3r+1)}$, as we have:
\begin{equation*}
 \begin{split}
    p\log_2(3r+1)-1&<\lfloor p\log_2(3r+1)\rfloor\\
    \Leftrightarrow\dfrac{\lambda p\log_2(3r+1)}{p\log_2(3r+1)-1}&>\dfrac{\lambda p\log_2(3r+1)}{\lfloor p\log_2(3r+1)\rfloor},\\
    \text{from which, by substitution, we get: }\; \dfrac{\lambda p\log_2(3r+1)}{p\log_2(3r+1)-1}&>\delta',\\
    \text{now requiring: }\; \dfrac{\lambda p\log_2(3r+1)}{p\log_2(3r+1)-1}&\le1,\\
    \text{or: }\; p\ge\dfrac{1}{(1-\lambda)\log_2(3r+1)},
 \end{split}
\end{equation*}
that is precisely our definition of $p$. This concludes the proof.
\end{proof}
From the above, we directly get the following corollary (for $r=1$):
\begin{corollary}\label{cw_SETH_DOM_cor}
 If \textsc{Dominating Set} can be solved in $O^*((4-\epsilon)^{\textrm{cw}(G)})$ time for some $\epsilon>0$, then \textsc{SAT} can be solved in $O^*((2-\delta)^n)$ time for some $\delta>0$.
\end{corollary}

\subsection{Dynamic programming algorithm}\label{sec_cw_DP} 

We now present an $O^*((3r+1)^{\cw})$-time dynamic
programming algorithm for unweighted \KC, using a given
clique-width expression $T_G$ for~$G$ with at most~$\cw$
labels. Even though the algorithm relies on standard techniques~(DP), there are several non-trivial, problem-specific observations that
we need to make to reach a DP table size of $(3r+1)^\cw$.

Our first step is to re-cast the problem as a \emph{distance-labeling} problem,
that is, to formulate the problem as that of deciding for each vertex what is
its precise distance to the optimal solution $K$. This is helpful because it
allows us to make the constraints of the problem local, and hence easier to
verify: roughly speaking, we say that a vertex is satisfied if it has a
neighbor with a smaller distance to $K$ (we give a precise definition below).
It is now not hard to design a clique-width based DP algorithm for this version
of the problem: for each label $l$ we need to remember two numbers, namely
the smallest distance value given to some vertex with label $l$, and the
smallest distance value given to a \emph{currently unsatisfied} vertex with
label $l$, if such a vertex exists.

The above scheme directly leads to an algorithm running in time (roughly)
$((r+1)^2)^\cw$. In order to decrease the size of this table, we now make the
following observation: if a label-set contains a vertex at distance $i$ from
$K$, performing a join operation will satisfy all vertices that expect to be at
distance $\ge i+2$ from $K$, since all vertices of the label-set will now be at
distance at most $2$. This implies that, in a label-set where the minimum
assigned value is $i$, states where the minimum unsatisfied value is between
$i+2$ and $r$ are effectively equivalent. With this observation we can bring
down the size of the table to $(4r)^\cw$, because (intuitively) there are four
cases for the smallest unsatisfied value: $i,i+1,\ge i+2$, and the case where all
values are satisfied.

The last trick that we need to achieve the promised running time departs
slightly from the standard DP approach. We will say that a label-set is
\emph{live} in a node of the clique-width expression if there are still edges
to be added to the graph that will be incident to its vertices. During the
execution of the dynamic program, we perform a ``fore-tracking'' step, by
checking the part of the graph that comes higher in the expression to determine
if a label-set is live. If it is, we merge the case where the smallest
unsatisfied value is $i+2$, with the case where all values are satisfied (since
a join operation will eventually be performed). Otherwise, a partial solution
that contains unsatisfied vertices in a non-live label-set can safely be
discarded. This brings down the size of the DP table to $(3r+1)^\cw$, and then
we need to use some further techniques to make the total running time quasi-linear in the size of the
table. This involves counting the number of
solutions instead of directly computing a solution of minimum size, as well as a non-trivial extension of fast subset convolution from \cite{BjorklundHKK07} for a $3\times(r+1)$-sized table (or \emph{state-changes}). See also \cite{RooijBR09,BodlaenderLRV10} and Chapter 11 of \cite{CyganFKLMPPS15}.
\subparagraph{Distance labeling:} We first require an alternative
formulation of the problem, based on the existence of a function $\dl$ that
assigns numbers $\dl(v)\in[0,r]$ to all vertices $v\in V$.

Let $\dl:V\mapsto[0,r]$ and $\dl^{-1}(i)$ be the set of all vertices with assigned number $i\in[0,r]$ by $\dl$. A function $\dl$ is then called \emph{valid}, if for all labels $\dl\in[1,\textrm{cw}]$ and nodes $t$ of the clique-width expression, at least one of the following conditions holds for all vertices $u\in V_l\cap G_t\setminus \dl^{-1}(0)$:
\begin{enumerate}[1)]
\item There is a neighbor $v$ of $u$ in $G_t$ with a strictly smaller number: $\exists v\in G_t:(u,v)\in E_t\wedge \dl(v)<\dl(u)$; \\
\item There is a vertex $v$ in the same label $l$ as $u$ and at distance 2 from it in $G_t$, while their difference in numbers is at least 2: $\exists v\in G_t\cap V_l\wedge \dl(v)\le \dl(u)-2\wedge d_t(u,v)=2$; \\
\item There is a vertex $v$ in the same label $l$ as $u$ in $G_t$ with their difference in numbers at least 2 and some vertex $w$ adjacent to it in the final graph $G$: $\exists v\in G_t\cap V_l\wedge \dl(v)\le \dl(u)-2\wedge\exists w\in G: (u,w)\notin G_t\wedge (u,w)\in G$.
\end{enumerate}
Note that in condition~3) above, vertex $w$ will also be adjacent to $u$ in $G$ (and thus $d(u,v)=2$), as both $u$ and $v$ belong to the same label. A vertex $u\in V_l\cap G_t$ with label $l\in[1,\textrm{cw}]$ is \emph{satisfied} by $\dl$ for node $t$, if $\dl(u)=0$, or either of the first two conditions 1)2) above holds for $u$. Let $U_l^t(\dl)$ be the set of vertices of label $l$ that are not satisfied by $\dl$ for node $t$ and $\DL_r(t)$ be the set of all possible valid functions $\dl$ for given $r$, restricted to the vertices of $G_t$ for node $t$ of clique-width expression $T_G$, with disjoint sets $\dl^{-1}(0)$. The following lemma shows the equivalence between the two formulations.
\begin{lemma}\label{lem:cw_exact_dl} A graph $G=(V,E)$ admits a
$(k,r)$-center if and only if it admits a valid distance-labeling function
$\dl: V\to \{0,\ldots,r\}$ with $|\dl^{-1}(0)|=k$.
\end{lemma}
 \begin{proof}
 To see why a valid function $\dl$ with $|\dl^{-1}(0)|=k$ represents a solution to the \textsc{$(k,r)$-Center} problem consider the following: first, given a $(k,r)$-center of $G$, let $\dl$ be the function that assigns to each vertex $v\in V$ a number equal to its distance from the closest center, i.e.\ number 0 to the centers, 1 to their immediate neighbors and so on. This function is valid as for every vertex $u$ there always exists some neighbor $v$ with $\dl(v)<\dl(u)$, being the neighbor that lies on the path between $u$ and its closest center, while also $|\dl^{-1}(0)|=k$. On the other hand, given such a valid function $\dl$, the set $\dl^{-1}(0)$ is indeed a $(k,r)$-center: we have $|\dl^{-1}(0)|=k$ and first, vertices in $\dl^{-1}(1)$ must have a neighbor in the center-set, while vertices in $\dl^{-1}(2)$ are at distance at most 2 from some center. Then, for $i\in[3,r]$, vertices $u$ in $\dl^{-1}(i)$ either have a neighbor in $\dl^{-1}(j)$ with $j<i$, or are at distance at most 2 from some vertex in $\dl^{-1}(j)$ with $j\le i-2$, that by induction must in both cases be at distance at most $j$ from some center, making the distance between $u$ and some vertex in $\dl^{-1}(0)$ at most $i\le r$.
 \end{proof}
\subparagraph{Table description:} There is a table $D_t$ associated with every node $t$ of the clique-width expression, while each table entry $D_t[\kappa,s_1,\dots,s_w]$, with $w\le\textrm{cw}$, is indexed by a number $\kappa\in[0,k]$ and a $w$-sized tuple $(s_1,\dots,s_w)$ of \emph{label-states}, assigning a state $s_l=(v_l,u_l)$ to each label $l\in[1,w]$, where $v_l=\min_{x\in V_l}\dl(x)\in[0,r]$ is the minimum number assigned to any vertex $x$ in label $l$, while $u_l\in\{0,1,2\}$ is the difference between $v_l$ and the minimum number assigned by $\dl$ to any vertex $y\in U_l^t(\dl)$ that is not satisfied by $\dl$ for this node: $u_l=0$ when $\min_{y\in U_l^t(\dl)}(\dl(y)-v_l)=0$, $u_l=1$ when $\min_{y\in U_l^t(\dl)}(\dl(y)-v_l)=1$ and $u_l=2$ when $\min_{y\in U_l^t(\dl)}(\dl(y)-v_l)\ge2$, or when $U_l^t(\dl)=\emptyset$. Note that states $(0,0)$ and $(r,1)$ do not signify any valid situation and are therefore not used.

There are thus $3r+1$ possible states for each label, each being a pair signifying the minimum number assigned to any vertex in the label and whether the difference between this and the minimum number of any vertex in the label that is not yet satisfied by $\dl$ is either exactly 0,1, or greater than 1, with absence of unsatisfied vertices also considered in the latter case.
For a node $t$ with $w$ involved labels, each table entry $D_t[\kappa,s_1,\dots,s_w]$ contains the number $|\DL_r(t)|$ of valid functions $\dl$ restricted to the vertices of $G_t$ with disjoint sets $\dl^{-1}(0)$ and $|\dl^{-1}(0)|=\kappa$, such that for each label $l\in[1,w]$, its state in the tuple gives the conditions that must be satisfied for this label by any such function $\dl$ that is to be counted in the entry's value. In particular, we have $\forall t\in T, D_t[\kappa,s_1,\dots,s_w]:\{\kappa\in[0,k]\}\times\{(1,0),\dots,(r,0),(0,1),\dots,(r-1,1),(0,2),\dots,(r,2)\}^w\mapsto\mathbb{N}^0$, where $w\in[1,\textrm{cw}]$.

The inductive computations of table entries for each type of node follows.
\subparagraph{Introduce node:} For node $t$ with operation $i(l)$ and $l\in[1,\textrm{cw}]$, we have:
\begin{equation*}
D_t[\kappa,s_l] \coloneqq
 \begin{cases}
  1, & \mbox{if } v_l=0,u_l=2,\kappa=1;\\
  1, & \mbox{if } v_l\neq0,u_l=0,\kappa=0;\\
			0, & \mbox{otherwise.}
 \end{cases}
\end{equation*}
\subparagraph{Join node:} For node $t$ with operation $\eta(a,b)$, child node $t-1$ and $a,b\in[1,w]$, let $Q(s'_a=(v'_a,u'_a),s'_b=(v'_b,u'_b))\coloneqq\{(s_a=(v_a,u_a),s_b=(v_b,u_b))|[v_a=v'_a\wedge v_b=v'_b]\wedge[((u'_a=2\wedge v_a+u_a>v_b)\vee(u_a=u'_a<2\wedge v_a+u_a\le v_b))\wedge((u'_b=2\wedge v_b+u_b>v_a)\vee(u_b=u'_b<2\wedge v_b+u_b\le v_a))]\}$. In words, $Q(s'_a,s'_b)$ is the set of all pairs of label states $(s_a,s_b)$, such that if label $a$ is joined with label $b$, their new states could be $s'_a,s'_b$, i.e.\ all pairs of states where the $v$ values remain the same for both $a,b$, as no new numbers are introduced within any label by a join operation, yet some vertices may become satisfied through the addition of new edges and thus the $u$ values of their label might change to 2. We then have:
\begin{equation*}
 D_t[\kappa,s_1,\dots,s'_a,s'_b,\dots,s_w]\coloneqq\sum_{(s_a,s_b)\in Q(s'_a,s'_b)}D_{t-1}[\kappa,s_1,\dots,s_a,s_b,\dots,s_w].
\end{equation*}
\subparagraph{Rename node:} For node $t$ with operation $\rho(w+1,w)$ and child node $t-1$ (we assume without loss of generality that the last label is renamed into the one preceding it), let $M(s=(v,u))\coloneqq\{(s_a=(v_a,u_a),s_b=(v_b,u_b))|[v=\min\{v_a,v_b\}]\wedge[[(u=0)\wedge((v_a=v\wedge u_a=0\wedge v_b\ge v\wedge u_b\ge0)\vee(v_b=v\wedge u_b=0\wedge v_a\ge v\wedge u_a\ge0))] \vee [(u=1)\wedge( (v_a=v\wedge u_a=1\wedge v_b=v\wedge u_b\ge1) \vee (v_a=v\wedge u_a\ge1\wedge v_b=v\wedge u_b=1) \vee (v_a=v\wedge u_a=1\wedge v_b> v\wedge u_b\ge0) \vee (v_b=v\wedge u_b=1\wedge v_a>v\wedge u_a\ge0) \vee (v_a=v\wedge u_a=2\wedge v_b=v+1\wedge u_b=0) \vee (v_b=v\wedge u_b=2\wedge v_a=v+1\wedge u_a\ge0))] \vee [((u=2)\wedge(u_a=u_b=2)\vee(v<v_a\wedge u_b=2\wedge((v_a-v\ge2\wedge u_a=0)\vee(v_a-v\ge1\wedge u_a=1)))\vee(v<v_b\wedge u_a=2\wedge((v_b-v\ge2\wedge u_b=0)\vee(v_b-v\ge1\wedge u_b=1)))]]\}$. In words, $M(s)$ is the set of all pairs of labels $(s_a,s_b)$ for two labels $a,b$, such that renaming one into the other could produce state $s$ for the resulting label, i.e.\ the pairs of states where the resulting $v$ value is the minimum of $v_a,v_b$, while $u$ comes from any of the appropriate combinations of $u_a,u_b$ for each case. We then have:
\begin{equation*}
 D_t[\kappa,s_1,\dots,s'_w]\coloneqq\sum_{(s_w,s_{w+1})\in M(s'_w)}D_{t-1}[\kappa,s_1,\dots,s_w,s_{w+1}].
\end{equation*}
\subparagraph{Union node:} For node $t$ with operation $G_{t-1}\cup G_{t-2}$ and children nodes $t-1,t-2$, we assume (again, without loss of generality) that all labels in $[1,y\le w]$ are involved in $t-1$ and all labels in $[1,z\le w]$ are involved in $t-2$, such that for some $i\in[1,w]$, all labels $1\le j\le i$ are involved in both nodes $t-1,t-2$, i.e.\ labels $1,\dots,i\le y,z\le w$ are common to both preceding nodes. Also observe that for any resulting state $s'_j$ of label $j$, following application of the union operation on two partial solutions where its states are $s_j$ and $\bar{s}_j$, the pair $(s_j,\bar{s}_j)$ would be included in $M(s'_j)$, similarly to renaming $j$ from one partial solution with state $s_j$ to the same label from the other partial solution with state $\bar{s}_j$. We then have:
\begin{align*}
 D_t[\kappa,s'_1,\dots,s'_i,\dots,s_w]\coloneqq\sum_{\substack{(s_j,\bar{s}_j)\in M(s'_j), j\in[1,i]\\ \kappa_1\in[0,\kappa]}}(&D_{t-1}[\kappa_1,s_1,\dots,s_i,\dots,s_y]\cdot\\
 \cdot&D_{t-2}[\kappa-\kappa_1,\bar{s}_1,\dots,\bar{s}_i,\dots,s_z]).
\end{align*}
\subparagraph{State changes:} Since the number of entries involved in the above computation of a union node's table can be exponential in the clique-width of the graph (due to the number of possible combinations of state pairs in the sum), in order to efficiently compute the tables of union nodes we will use \emph{state changes}: for a union node $t$ we will first transform the tables $D_{t-1},D_{t-2}$ of its children into tables $D^*_{t-1},D^*_{t-2}$ of a new type that employs a different state representation, for which the union operation can be efficiently performed to produce table $D^*_t$, that we finally will transform back to table $D_t$, thus progressing with our dynamic programming algorithm.

In particular, each entry of table $D^*_t[\kappa,s_1,\dots,s_w]$ of a node $t$ for $\kappa\in[0,k]$ and $w\in[1,\textrm{cw}]$ will be an aggregate of entries from $D_t[\kappa,s_1,\dots,s_w]$, with its value equal to the sum of the appropriate values of the original table. For label $l$, its state $s^*_l=(v^*_l,u^*_l)$ in the new state signification for table $D^*_t$ will correspond to all states where $v_l^*\ge v_l+i$ and $u_l^*\ge u_l-i$ in the previous state signification for all (applicable) $i\in\{0,1,2\}$. Observe that these correspondences exactly parallel the states described in the definition of set $M(s_l^*)$ given above for the computation of a rename node, i.e.\ all states that could combine to produce the resulting state $s^*$.

First, let $D'_t$ be a copy of table $D_t$. The transformation then works in two stages, of $w$ steps each, label-wise: we first produce the intermediate table $D'_t$ from $D_t$ and then table $D^*_t$ from $D'_t$. For table $D'_t$ we require that all entries $D'_t[\kappa,s'_1,\dots,s'_w]$ contain the sum of all entries of $D_t$ where $s'_l=(v'_l\ge v_l,u'_l=u_l)$ and all other label-states and $\kappa$ are fixed: at step $l$, we first add the entry where $s_l=(r,0)$ to the entry where $s_l=(r-1,0)$ and then the entry with $s_l=(r-2,0)$ to the previous result and so on until $s_l=(1,0)$. We then repeat this process for $u=1$ from $s_l=(r-1,1)$ until $s_l=(0,1)$ and then for $u=2$ from $s_l=(r,2)$ until $s_l=(0,2)$. We then proceed to the next label until table $D'_t$ is computed.

For the next stage, we initialize $D^*_t$ as a copy of $D'_t$ and also work on $w$ steps, again label-wise, fixing all other label-states and $\kappa$: at step $l$, we first add the entries where $s_l=(v_l,1)$ and $s_l=(v_l,2)$ from $D'_t$ to entries where $s_l=(v_l,0)$ from $D^*_t$, for each $v_l$ from $r-1$ to 1, in turn (and only the one with $s_l=(r,2)$ from $D'_t$ for the one where $s_l=(r,0)$ from $D^*_t$). We then add the entries where $s_l=(v_l+1,0)$ and $s_l=(v_l,2)$ from $D'_t$ to entries where $s_l=(v_l,1)$ from $D^*_t$, for each $v_l$ from $r$ to 0, in turn. Finally, we add the entries where $s_l=(v_l+1,1)$ and the entries where $s_l=(v_l+2,0)$ from $D'_t$ to entries where $s_l=(v_l,2)$ from $D^*_t$, for each $v_l$ from $r$ to 0, in turn. We then proceed to the next step for the following label until table $D^*_t$ is computed. See Figure \ref{fig:state_changes} for an illustration of the relationships between states/entries of these tables.

\begin{figure}[htbp]
 \centerline{\includegraphics[width=110mm]{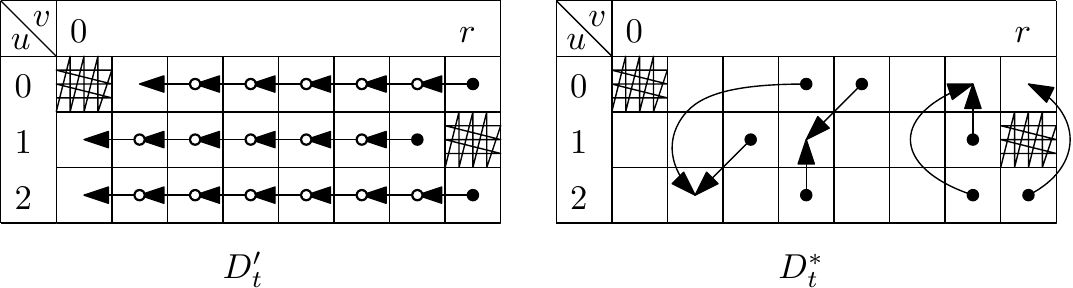}}
 \caption{A conceptual representation of the relationships between states/entries to be added for the computation of the alternative table types $D',D^*$. Arrows imply the transfer of values (addition), black endpoints indicate corresponding entries from a previous table type and white endpoints indicate previously computed entries in the same table.}
 \label{fig:state_changes}
\end{figure}

Note that the above procedure is fully reversible,\footnote{This is the reason for \emph{counting} the number of solutions for each $\kappa$, instead of finding the \emph{minimum} $\kappa$ for which at least one solution exists: there is no additive inverse operation for the min-sum semiring, yet the sum-product ring is indeed equipped with subtraction.} as for each label $l$, entries where $s_l=(r,2)$ are the same in all tables $D_t,D'_t,D^*_t$ and thus, to obtain $D'_t$ from $D^*_t$ we again work label-wise, fixing all other label-states and $\kappa$: at step $l$, we first compute the entry of $D'_t$ where $s_l=(r,0)$, by subtracting from its value the value of the entry where $s_l=(r,2)$ and we then do the same for the entry where $s_l=(r-1,2)$, moving to $s_l=(r-1,1)$ by subtracting the ones where $s_l=(r-1,2)$ and $s_l=(r,0)$, then to $s_l=(r-1,0)$ by subtracting the ones where $s_l=(r-1,1)$ and $s_l=(r-1,2)$ and so on, for each $v_l$ from $r-1$ to 0, in turn. After all $w$ steps we have obtained table $D'_t$ and from this we can similarly obtain table $D_t$, as again, for each label $l$ the entries where $s_l=(r,0),(r-1,1),(r,2)$ are the same: at step $l$, we compute the entries where $s_l=(v_l,i)$ by subtracting from the corresponding entries of $D'_t$ all entries where $s_l=(v_l+1,i)$, for $v_l$ from $r-1$ to 0 and $i\in\{0,1,2\}$ in turn. For both transformations and directions, we perform at most two additions/subtractions per entry for $\kappa\cdot(3r+1)^{\textrm{cw}}$ entries, for each step $l\in[1,w\le\textrm{cw}]$.

Thus we can compute table $D^*_t$ by simply multiplying the values of the two corresponding entries from $D^*_{t-1},D^*_{t-2}$, as they now contain all required information for this state representation, with the inverse transformation of the result giving table $D_t$:
\begin{equation*}
 D^*_t[\kappa,s_1,\dots,s_w]\coloneqq\sum_{\kappa_1=0}^{\kappa_1=\kappa}D^*_{t-1}[\kappa_1,s_1,\dots,s_w]\cdot D^*_{t-2}[\kappa-\kappa_1,s_1,\dots,s_w].
\end{equation*}
\subparagraph{Fore-tracking:} As is already apparent, our algorithm actually solves the counting version of \textsc{$(k,r)$-Center}, by reading the values of entries from table $D_z$ of the final node $z$, where all labels $l$ are of state $s_l=(v_l,u_l=2)$. Now, as already noted, these states actually correspond to both situations where either all unsatisfied vertices in the label are of number $\ge2$ than the label's $v$ value, or where there are no unsatisfied vertices in the label for this $\dl$. To ensure our algorithm only counts valid functions $\dl$, we employ a ``fore-tracking'' policy: whenever some entry is being computed for the table of some rename or union node $t$, where some label $l$ has state $s_l=(v_l,u_l=2)$, we establish whether condition 3) given in the definition of a valid function $\dl$ is also satisfied by all counted functions $\dl$ for all vertices in $l$, by verifying that some join operation $\eta(l,m)$ is applied between $l$ and another label $m$ in some subsequent node $t'$ of the clique-width expression (even after $l$ is potentially renamed to some other label). If such a join operation is indeed to be applied (and the label-set is live), there will be a vertex $w\in V_m$ that is adjacent to all vertices $u,v\in V_l$ in the final graph $G$ and as $u_l=2$, there must also be some vertex $v\in V_l$ with $\dl(v)\le \dl(u)-2$, for any unsatisfied vertex $u$. Since it will be $d(u,v)=2$, all such vertices $u$ will be satisfied in $G$ for all such $\dl$, that will in fact be valid.

On the other hand, in the absence of such a join operation and the case where $t$ is a rename node, we consider the definition of set $M(s)$ above, where there are three options for a resulting state $s_l=(v_l,u_l)$ with $u_l=2$ from two preceding states $s_a=(v_a,u_a),s_b=(v_b,u_b)$ (corresponding to the last bracket of clauses in the set's definition): we must have either $u_a=u_b=2$, i.e.\ that both states have a difference of at least 2 between the minimum number of any vertex and the minimum of any unsatisfied vertex, or $v_l<v_a$ and $u_b=2$, while either $v_a-v_l\ge2$ and $u_a=0$, or $v_a-v_l\ge1$ and $u_a=1$ (or vice-versa), i.e.\ that one label must have state $(>v_l,0)$ and the other can have any state from $(\ge v_l+2,0)$ or $(\ge v_l+1,1)$, their combined numbers giving $(v_l,2)$. Now, if no join operation follows the rename node $t$ for this label, we simply disregard any options in the above computation (from the last two) that consider states where one of the preceding labels $a,b$ had $u_a,u_b<2$.

Similarly, for a union node $t$ we consider the state changes given above: if no join node follows $t$ for this label, we simply disregard the additions (and subsequently their corresponding subtractions) for this label in the table's transformation from $D'_t$ to $D^*_t$ in step $l$, where entries with $s_l=(v_l+1,1)$ and $s_l=(v_l+2,0)$ from $D'_t$ are added to entries with $s_l=(v_l,2)$ from $D^*_t$ and keep all such entries as they are (direct copies from $D'_t$). In this way, as any state where $u=2$ for some label must be ``the direct result'' of some join operation (validating condition 1), or there will be some subsequent join operation satisfying all vertices of number $\ge2$ than the minimum (condition 2), no non-valid functions $\dl$ for which unsatisfied vertices are infused in the label producing some ``false'' state where $u=2$ can be counted, as these vertices are not to be satisfied by some following join node (3).
\subparagraph{Correctness:} To show correctness of our algorithm we need to establish that for every node $t\in T_G$, each table entry $D_t[\kappa,s_1,\dots,s_w]$ contains the number of partial solutions to the sub-problem restricted to the graph $G_t$, i.e.\ the number of valid functions $\dl\in \DL_r(t)$ with $|\dl^{-1}(0)|=\kappa$ (and disjoint sets $\dl^{-1}(0)$), such that for each label $l\in[1,w]$ (being the labels involved in $t$) the conditions imposed by its state $s_l=(v_l,u_l)$ are satisfied by all such $\dl$, i.e.\ $\min_{x\in V_l}\dl(x)=v_l$ and $\min_{y\in U_l^t(\dl)}(\dl(y)-v_l)=0,1$ when $u_l=0,1$, while $\min_{y\in U_l^t(\dl)}(\dl(y)-v_l)\ge2$, or $U_l^t(\dl)=\emptyset$ when $u_l=2$. That is:
\begin{gather}
 \forall t\in T_G,w\in[1,\textrm{cw}],\forall(\kappa,s_1,\dots,s_w)\in\label{cw_cor_1}\\
 \in\{\kappa\in[0,k]\}\times\{(1,0),\dots,(r,0),(0,1),\dots,(r-1,1),(0,2),\dots,(r,2)\}^w:\label{cw_cor_2}\\
 \Big\{D_t[\kappa,s_1,\dots,s_w]=|\DL_r(t)|:\forall \dl\in \DL_r(t), \forall l\in[1,w]:\label{cw_cor_3}\\
 \big\{(\min_{x\in V_l}\dl(x)=v_l)\wedge\label{cw_cor_4}\\
 \wedge\big[(u_l=0\wedge \min_{y\in U_l^t(\dl)}(\dl(y)-v_l)=0)\vee\label{cw_cor_5}\\
 \vee(u_l=1\wedge \min_{y\in U_l^t(\dl)}(\dl(y)-v_l)=1)\vee\label{cw_cor_6}\\
 \vee(u_l=2\wedge \big(\min_{y\in U_l^t(\dl)}(\dl(y)-v_l)\ge2\vee U_l^t(\dl)=\emptyset\big)\big]\big\}\wedge\label{cw_cor_7}\\
 \wedge\{|\dl^{-1}(0)|=\kappa\}\Big\}\label{cw_cor_8}.
\end{gather}
This is shown by induction on the nodes $t\in T_G$:
\begin{itemize}
 \item Introduce nodes: This is the base case of our induction. For node $t$ with operation $i(l)$ and $l\in[1,\textrm{cw}]$, all entries are properly initialized as there is one function $\dl$ that includes the introduced vertex in the center-set and one that does not, thus $|\DL_r(t)|=1$ for $s_l=(0,2),\kappa=1$ and $s_l=(>0,0),\kappa=0$, while for any other configuration $\DL_r(t)=\emptyset$.  In the following cases, we assume (our induction hypothesis) that all entries of $D_{t-1}$ (and $D_{t-2}$ for union nodes) satisfy the above statement (\ref{cw_cor_1}-\ref{cw_cor_8}).
 \item Join nodes: For node $t$ with operation $\eta(a,b)$ and $a,b\in[1,w]$, all edges are added between vertices of labels $a$ and $b$. Thus for each entry $D_t[\kappa,s_1,\dots,s'_a,s'_b,\dots,s_w]$, all valid functions $\dl$ counted in the entries of the previous table $D_{t-1}[\kappa,s_1,\dots,s_a,s_b,\dots,s_w]$ remain valid if $v'_a=v_a,v'_b=v_b$ (\ref{cw_cor_4}) and $u'_a=u_a,u'_b=u_b$ (\ref{cw_cor_5}-\ref{cw_cor_7}), while for all entries in $D_t$ where $v'_a=v_a,v'_b=v_b$ and $u'_a=2$ (resp.\ $u'_b=2$), all functions counted for entries of $D_{t-1}$ where $v_a+u_a>v_b$ (resp.\ $v_b+u_b>v_a$), are valid as well (\ref{cw_cor_7}), since in the resulting graph $G_t$ of node $t$, vertices in $a$ (resp.\ $b$) can become satisfied by the addition of such edges (and validation of condition 1). This is precisely the definition of set $Q(s'_a,s'_b)$.
 \item Rename nodes: For node $t$ with operation $\rho(w+1,w)$, all vertices of label $w+1$ will now be included in label $w$. Thus for each entry $D_t[\kappa,s_1,\dots,s'_w]$, we require all valid functions $\dl$ counted in the entries of the previous table $D_{t-1}[\kappa,s_1,\dots,s_w,s_{w+1}]$, where states $s_w,s_{w+1}$ describe the possible situations for labels $w,w+1$ that could combine to give the new state $s'_w$ for label $w$: first, value $v'_w$ denotes the minimum allowed number in any vertex that was previously in $w$ or $w+1$, yet at least one of these labels must have had some vertex with exactly this number and thus $v'_w=\min\{v_w,v_{w+1}\}$ (\ref{cw_cor_4}). Next, value $u'_w$ denotes whether the difference between this number and the minimum number of any unsatisfied vertex is 0,1, or greater (if any) and the possible combinations of states $s_w,s_{w+1}$ are fully described in the definition of set $M(s'_w)$: for $u'_w$ to be 0 we must have $v_w=v'_w$ and $u_w=0$, while $v_{w+1}\ge v_w$ and $u_{w+1}\ge0$ (or vice-versa), i.e.\ at least one of $u_w,u_{w+1}$ must be 0 with its corresponding value $v_w,v_{w+1}$ being the minimum, while the state of the other can have any at least equal values, as only the minimum is retained (\ref{cw_cor_5}). For $u'_w$ to be 1 we must have either $v_w=v'_w$ and $u_w=1$, while also $v_{w+1}=v'_w$ and $u_{w+1}\ge1$, or $v_w=v'_w$ and $u_w=1$, while $v_{w+1}>v'_w$ and $u_{w+1}\ge0$ (or vice-versa), i.e.\ one label must have the same target state, while the other must also have the same $v$ with a $u$ greater than 1, or a greater $v$ and any $u$ so as not to make $u'_w$ smaller than the target 1 (\ref{cw_cor_6}). Finally, for $u'_w$ to be 2 we must have either $u_w=u_{w+1}=2$, i.e.\ that both states have a difference of at least 2 between the minimum number of any vertex and the minimum of any unsatisfied vertex, or $v'_w<v_w$ and $u_{w+1}=2$, while either $v_w-v'_w\ge2$ and $u_w=0$, or $v_w-v'_w\ge1$ and $u_w=1$ (or vice-versa), i.e.\ that one label must have state $(>v'_w,0)$ and the other can have any state from $(\ge v'_w+2,0)$ or $(\ge v'_w+1,1)$, their combined numbers inducing $(v'_w,2)$ (\ref{cw_cor_7}).
 \item Union nodes: For node $t$ with operation $G_{t-1}\cup G_{t-2}$ and children nodes $t-1,t-2$, the resulting graph $G_t$ is the disjoint union of graphs $G_{t-1},G_{t-2}$, where all common labels $j\in[1,i]$ involved in both $t-1,t-2$ will now include all vertices that were included in label $j$ in one of the previous graphs, while nothing changes for all other labels. This can be seen to be similar to the result of a rename operation between the label $j$ from graph $G_{t-1}$ and the same label from $G_{t-2}$. Thus for each entry $D_t[\kappa,s'_1,\dots,s'_i,\dots,s_w]$, we require the \emph{product} of the numbers of all valid functions $\dl$ counted in entries $D_{t-1}[\kappa_1,s_1,\dots,s_i,\dots,s_y]$ for $t-1$ and those counted in entries $D_{t-2}[\kappa-\kappa_1,\bar{s}_1,\dots,\bar{s}_i,\dots,s_z]$ for $t-2$, summed over all possible combinations of states $s_j,\bar{s}_j$ that belong to some pair in $M(s'_j)$ for every $j\in[1,i]$, as well as over all values of $\kappa_1$ from 0 to $\kappa$, since the number of valid functions in $G_t$ respecting the conditions of the target entry is equal to a number of valid functions in $G_{t-1}$ (counted in entries of the first type) for each valid function that complements each of them in $G_{t-2}$ (counted in entries of the second type): the sizes of $\dl^{-1}(0)$ always add up to $\kappa$ (\ref{cw_cor_8}), while the states $s_j,\bar{s}_j$ of all common labels $j$ in any pair of entries from $D_{t-1},D_{t-2}$ would result in the desired target state $s_j$ for the same label in the entry from $D_t$ in question, as they belong in the set $M(s'_j)$ (\ref{cw_cor_4}-\ref{cw_cor_7}).
\end{itemize}
\begin{theorem}\label{thm:cw_dp}
 Given graph $G$, along with $k,r\in\mathbb{N}^+$ and clique-width expression $T_G$ of clique-width $\textrm{cw}$ for $G$, there exists an algorithm to solve the counting version of the \textsc{$(k,r)$-Center} problem in $O^*((3r+1)^{\textrm{cw}})$ time.
\end{theorem}
\begin{proof}
 Correctness of the dynamic programming algorithm is given above, while for the final computation at the root $z$ of $T_G$, all entries $D_z[k,s_1,\dots,s_w]$ with $u_l=2,\forall l\in[1,w]$ and $w\in[1,\textrm{cw}]$ can be considered for identification of the number of valid functions $\dl$, or $(k,r)$-centers of graph $G_z=G$. For the decision version of the problem, the algorithm can output YES if any of these entries' value is $>0$ and NO otherwise.

 For the algorithm's complexity, there are at most $k(3r+1)^{\textrm{cw}}$ entries in each table $D_t$ of any node $t$ and any entry of the $O(n\cdot\textrm{cw}^2)$ join/rename nodes can be computed in $O(r)$ time, while for the $O(n)$ union nodes, table transformations require $O^*(k\cdot\textrm{cw}\cdot(3r+1)^{\textrm{cw}})$ time and any entry of the transformed tables can be computed in $O(k)$ time.
\end{proof}

%% file: vc-fvs-td.tex
\section{Vertex Cover, Feedback Vertex Set and Tree-depth}\label{vcsec}

\subsection{Parameterization by Vertex Cover and Feedback Vertex Set numbers}

In this section we first show that the edge-weighted variant of the \KC\ problem
parameterized by $\vc+k$ is W[1]-hard, and more precisely, that the problem
does not admit a $n^{o(\vc+k)}$ algorithm under the ETH (Theorem~\ref{thm:W_hard_VCk}).  We give a reduction
from \textsc{$k$-Multicolored Independent Set}. This is a well-known W[1]-hard problem that cannot be solved in $n^{o(k)}$
under the ETH \cite{CyganFKLMPPS15}.

\subparagraph{Construction:} Given an instance $[G=(V,E),k]$ of
\textsc{$k$-Multicolored Independent Set}, we will construct an instance
$[G'=(V',E'),k]$ of edge-weighted \textsc{$(k,r)$-Center}, where $r=4n$. First,
for every set $V_i\subseteq V$, we create a set $P_i\subseteq V'$ of $n$
vertices $p_l^i,\forall l\in[1,n]$, $\forall i\in[1,k]$ (that directly
correspond to the vertices of $V_i$) and two \emph{guard} vertices
$g_i^1,g_i^2$, attaching them to all vertices in $P_i$ by edges of weight
$r=4n$.  Next, for each $i\in[1,k]$, we create another pair of vertices
$a_i,b_i$ and connect $a_i$ to each vertex $p_l^i$ by an edge of weight $n+l$,
while $b_i$ is connected to each vertex $p_l^i$ by an edge of weight $2n-l+1$.
Now each $P_i$ contains all vertices $a_i,b_i,g_i^1,g_i^2$ and each
$p_l^i,\forall l\in[1,n]$.  

Finally, for each edge $e\in E$ with endpoints in $V_{i_1}, V_{i_2}$ and
$i_1\neq i_2$ (not part of a clique), we create a vertex $u_e$ that we connect
to vertices $a_{i_1},b_{i_1}$ and $a_{i_2},b_{i_2}$. We set the weights of
these edges as follows: suppose that $e$ connects the $j_1$-th vertex of
$V_{i_1}$ to the $j_2$-th vertex of $V_{i_2}$. Then we set $w(u_e, a_{i_1}) =
3n - j_1+1$, $w(u_e, b_{i_1}) = 2n+j_1$, $w(u_e, a_{i_2}) = 3n - j_2+1$, $w(u_e,
b_{i_2}) = 2n+j_2$. 
\begin{figure}[htbp]
 \centerline{\includegraphics[width=120mm]{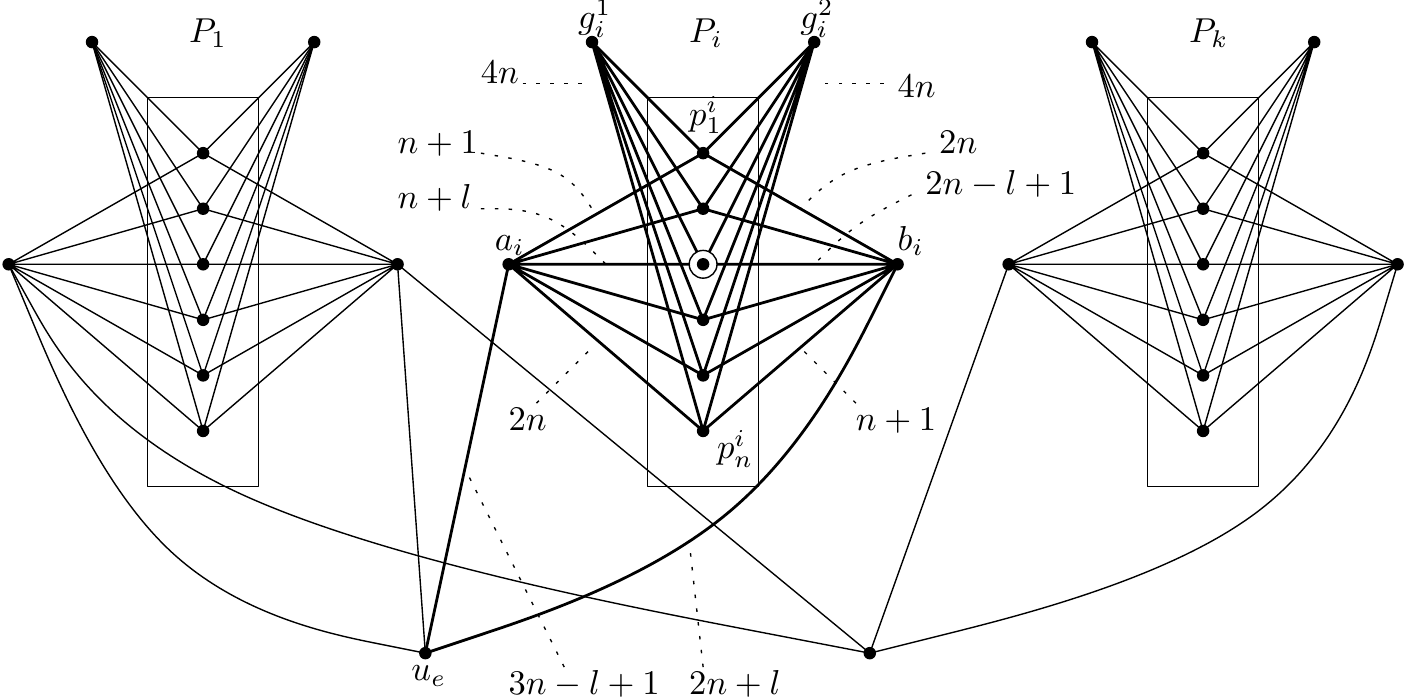}}
 \caption{A general picture of graph $G'$. The circled vertex is $p_l^i$, while dotted lines match edges to weights.}
 \label{fig:Whard_VC}
\end{figure}
\begin{lemma}\label{thm:W_FWD}
 If $G$ has a $k$-multicolored independent set, then $G'$ has a $(k,4n)$-center.
\end{lemma}
\begin{proof}
 Let $I\subseteq V$ be a multicolored independent set in $G$ of size $k$ and $v_{l_i}^i$ denote the vertex selected from each $V_i$, or $I\coloneqq\{v_{l_1}^1,\dots,v_{l_i}^i,\dots,v_{l_k}^k\}$. Let $S\subseteq V'$ be the set of vertices $p_{l_i}^i$ in $G'$ that correspond to each $v_{l_i}^i$. We claim $S$ is a $(k,r)$-center of $G'$: since one vertex is selected from each $P_i$ in $G'$, all the guards $g_i^1,g_i^2$ and vertices $a_i,b_i$ are within distance $r=4n$ from selected vertex $p_{l_i}^i$ for some $l_i\in[1,n]$ and every $i\in[1,k]$, as well as all other vertices $p_{l'}^i$. For vertices $u_e$, observe that selected vertex $p_{l_i}^i$ is at distance $n+l_i+3n-l_i+1=4n+1$ from $u_e$, through either $a_i$ or $b_i$, if its corresponding vertex in $G$ is an endpoint of $e$, or $e=(v_{l_i}^i,w)$ for some $w\in V\setminus V_i$.
 As $I$ is an independent set of $G$, however, there can be no edge between any two selected vertices $v_{l_i}^i,v_{l_j}^j\in I$ with $i\neq j$ and thus, we know that for every $u_e$ at least one of the vertices $p_{l_i}^i,p_{l_j}^j$ selected for $S$ from the two components $P_i,P_j$ to which $u_e$ is connected will not be one of the vertices corresponding to the endpoints of edge $e$, or $e\neq(v_{l_i}^i,v_{l_j}^j)$. Let $e\coloneqq(v_{x}^i,v_{y}^j)$, with $x\neq l_i$ and/or $y\neq l_j$. Assuming without loss of generality that $x\neq l_i$ is the case, we have that the distances from $u_e$ to $p_{l_i}^i$ are $n+l_i+3n-x+1$ via $a_i$ and $3n-l_i+n+x$ via $b_i$. If $x>l_i$ then distance via $a_i$ is $4n+1+l_i-x\le4n$, while if $x<l_i$ the distance via $b_i$ is $4n+x-l_i<4n$.
\end{proof}
\begin{lemma}\label{thm:W_BWD}
 If $G'$ has a $(k,4n)$-center, then $G$ has a $k$-multicolored independent set.
\end{lemma}
\begin{proof}
 Let $S\subseteq V'$ be the $(k,4n)$-center in $G'$. As $|S|=k$ and all guard vertices $g_i^1,g_i^2,\forall i\in[1,k]$ must be within distance $4n$ from some selected vertex, set $S$ must contain exactly one vertex from each $P_i$, or $S=\{p_{l_1}^1,\dots,p_{l_k}^k\}$ for some $l_i\in[1,n]$ for each $i\in[1,k]$. We let $I\subseteq V$ be the set of vertices $v_{l_i}^i\in V_i,i\in[1,k]$ that correspond to each $p_{l_i}^i$ and claim that $I$ is an independent set in $G$: suppose there is an edge $e=(v_{l_i}^i,v_{l_j}^j)\in E, i\neq j$ and $v_{l_i}^i,v_{l_j}^j\in I$. Then there must be a vertex $u_e\in G'$ with edges to $a_i,b_i\in P_i$ and $a_j,b_j\in P_j$, where we have $p_{l_i}^i,p_{l_j}^j\in S$. The distance from $u_e$ to $p_{l_i}^i$ is $3n-l_i+1+n+l_i>4n$, via either $a_i$ or $b_i$, while the distance from $u_e$ to $p_{l_j}^j$ is also $3n-l_j+1+n+l_j>4n$, via either $a_j,b_j$, meaning $u_e$ is not covered by $S$, giving a contradiction.
\end{proof}
\begin{theorem}\label{thm:W_hard_VCk} The weighted \KC\ problem is W[1]-hard
parameterized by $\vc+k$. Furthermore, if there is an algorithm for weighted
\KC\ running in time $n^{o(\vc+k)}$ then the ETH is false.  \end{theorem} 
\begin{proof} Observe that the set $Q\subset V'$ that includes all guard
vertices $g_i^1,g_i^2$ and $a_i,b_i$, $\forall i\in[1,k]$, is a vertex
cover of $G'$, as all edges have exactly one endpoint in $Q$.  This means
$\vc(G')\le4k$. In addition, parameter $k$ remains the same in both the
instances of \textsc{$k$-Multicolored Independent Set} and
\textsc{$(k,r)$-Center}. Thus, the construction along with Lemmas
\ref{thm:W_FWD} and \ref{thm:W_BWD}, indeed imply the statement.  \end{proof}
Using essentially the same reduction we also obtain a similar hardness result for unweighted \KC\ parameterized by $\fvs$.
\begin{corollary}\label{thm:W_hard_FVSk} The \KC\ problem is W[1]-hard when
parameterized by $\fvs+k$. Furthermore, if there is an algorithm for
\KC\ running in time $n^{o(\fvs+k)}$, then the ETH is false. \end{corollary}
\begin{proof}
We use the same reduction as in Theorem \ref{thm:W_hard_VCk}, except that we
replace all weighted edges by unweighted paths through new vertices in such a
way that distances between original vertices are preserved. It is not hard to
see that any set that was a vertex cover of the previous graph is a feedback
vertex set of the new graph, hence $\fvs = O(k)$. Lemma \ref{thm:W_FWD} goes
through unchanged, while for Lemma \ref{thm:W_BWD} it suffices to observe that,
because of the guard vertices, no valid solution can be using one of the new
vertices as a center.  \end{proof}
We now show that unweighted \KC\ admits an algorithm running in
time $O^*(5^\vc)$, in contrast to its weighted version (Theorem \ref{thm:W_hard_VCk}). We devise an algorithm that operates in two
stages: first, it guesses the intersection of the optimal solution with the
optimal vertex cover, and then it uses a reduction to \textsc{Set Cover} to
complete the solution. 
\begin{theorem}\label{thm:vc_algo}
 Given graph $G$, along with $k,r\in\mathbb{N}^+$ and a vertex cover of size
$\vc$ of $G$, there exists an algorithm solving unweighted \KC\ in
$O^*(5^{\vc})$ time.  \end{theorem}
\begin{proof}
Let $C$ be the given vertex cover of $G$ and $I=V\setminus C$ be the remaining
independent set. We assume without loss of generality that the graph is
connected (otherwise each component can be handled separately). We also assume
that $r\ge 2$, otherwise the problem reduces to \textsc{Dominating Set} which
is already known to be solvable faster than $O^*(5^\vc)$.

Let $K$ be some (unknown) optimal solution. Our algorithm first guesses a
partition of $C$ into three sets $C=S\cup R \cup Q$ such that $S=K\cap C$, $Q$
contains the vertices of $C$ which are at distance exactly $r$ from $K$, and
$R$ the rest (which are at distance $<r$ from $K$). Our algorithm first guesses
this partition by trying out all $3^\vc$ possibilities.

Suppose we are given a partition $C=S\cup R \cup Q$ as above. We would like to
check if this is the correct partition and then find a set $Z\subseteq I$ such
that $S\cup I$ is an optimal solution.  First, we verify that all vertices of
$Q$ are at distance $\ge r$ from $S$ (otherwise we already know this is not a
correct partition).  Second, we check if there exists $v\in I$ such that
$N(v)\subseteq Q$. In such a case, we again know that this is not a correct
partition, since such a $v$ would need to be included in $K$, which would imply
that its neighbors in $Q$ are at distance $<r$ from $K$. We can therefore
assume that all vertices in $I$ have some neighbor in $S\cup R$. 

We now formulate an instance of \textsc{Set Cover} as follows: the universe
contains all vertices of $R\cup Q$ which are not already covered, that is, all
vertices of $R$ which are at distance $\ge r$ from $S$ and all vertices of $Q$
which are at distance $>r$ from $S$. We construct a set for each vertex of
$v\in I$ and we place in it all vertices $u\in R$ such that $d(v,u)<r$ and all
vertices $u\in Q$ such that $d(v,u)\le r$. We solve this \textsc{Set Cover}
instance in time $O^*(2^{|R\cup Q|})$ using dynamic programming, and this gives
us a set $Z\subseteq I$. We output $S\cup Z$ as a solution to \KC. We observe
that this is always a valid solution because by construction all vertices of
$R$ are at distance $<r$ from $S\cup Z$, and all vertices of $I$ have a
neighbor in $S\cup R$. If we started with the correct partition of $C$ into
$S,R,Q$ then this solution is optimal because $K\cap I$ must give a feasible
set cover of the instance we constructed.

To analyze the running time, observe that if $|S|=i$, then
the algorithm goes through $2^{\vc-i}$ possible partitions of $C\setminus S$
into $R,Q$, and then for each partition spends $2^{\vc-i}n^{O(1)}$ to solve
\textsc{Set Cover}.  Hence, the algorithm runs in time $\sum_{i=0}^\vc {vc \choose
i} 4^{\vc-i}n^{O(1)} = \sum_{i=0}^\vc {\vc \choose i} 4^{i}n^{O(1)} = 5^\vc
n^{O(1)}$.  \end{proof}

\subsection{Parameterization by Tree-depth}

We now consider the unweighted version of \KC\ parameterized by~$\td$. Theorem~\ref{thm:W_hard_VCk} has established that weighted \KC\ is
W[1]-hard parameterized by $\vc$ (and hence also by $\td$ by Lemma
\ref{lem:relations}), but the complexity of unweighted \KC\ parameterized by~$\td$ does not follow from this theorem, since~$\td$ is incomparable
to $\fvs$. Indeed, we will show that \KC\ is FPT parameterized by~$\td$
and precisely determine its parameter dependence based on the ETH.

We begin with a simple upper bound argument. We will make use of the following
known fact on tree-depth (see also Corollary~6.1 of \cite{nesetril12}), while the algorithm then follows from the dynamic programming procedure of \cite{BorradaileL16} and the relationship between $r,\td$ and $\tw$: % (see Appendix \ref{append_TD}):
\begin{lemma}\label{lem:td-diam} For any graph $G=(V,E)$ we have $d(G)\le
2^{td+1}-2$, where $d(G)$ denotes the graph's diameter.  \end{lemma}
\begin{proof}
We use the following equivalent inductive definition of tree-depth:
$\td(K_1)=0$ while for any other graph $G=(V,E)$ we set $\td(G) = 1+ \min_{u\in
V} \td(G\setminus u)$ if $G$ is connected, and $\td(G) = \max_C \td(G[C])$ if
$G$ is disconnected, where the maximum ranges over all connected components of
$G$.

We prove the claim by induction. The inequality is true for $K_1$, whose
diameter is $0$. For the inductive step, the interesting case is when $G=(V,E)$
is connected, since otherwise we can assume that the claim has been shown for
each connected component and we are done. Let $u\in V$ be such that $\td(G) =
1+\td(G\setminus u)$. Consider two vertices $v_1,v_2\in V\setminus\{u\}$ which
are at maximum distance in $G$. If $v_1,v_2$ are in the same connected
component of $G':=G\setminus u$, then $d_{G}(v_1,v_2) \le d_{G'}(v_1,v_2) \le
d(G') \le 2^{\td(G')+1}-2 \le 2^{\td(G)+1} -2$, where we have used the
inductive hypothesis on $G'$. So, suppose that $v_1,v_2$ are in different
connected components of $G'$. It must be the case that $u$ has a neighbor in
the component of $v_1$ (call it $v_1'$) and in the component of $v_2$ (call it
$v_2'$), because $G$ is connected. We have $d_G(v_1,v_2) \le d_G(v_1,v_1') + 2
+ d_G(v_2,v_2') \le d_{G'}(v_1,v_1') + 2 + d_{G'}(v_2,v_2') \le 2d(G') +2 \le
2\cdot 2^{\td(G')+1} - 2 = 2^{\td(G)+1}-2$.  \end{proof}
\begin{theorem}\label{thm:td-alg}
Unweighted \KC\ can be solved in time $O^*(2^{O(\td^2)})$.
\end{theorem}
\begin{proof} The main observation is that we can assume that $r\le d(G)$,
because otherwise the problem is trivial. Hence, by Lemma \ref{lem:td-diam} we
have $r\le 2^{\td+1}$. We can now rely on Lemma \ref{lem:relations} to get
$\tw\le\td$, and the algorithm of \cite{BorradaileL16} which runs in time
$O^*((2r+1)^\tw)$ gives the desired running time.  \end{proof}
The main contribution concerning the tree-depth parameter is to establish a tight ETH-based lower
bound, matching Theorem \ref{thm:td-alg}. 

\subparagraph{Construction:} Given an instance $\phi$ of \textsc{3-SAT} on $n$
variables and $m$ clauses, where we can assume that $m=O(n)$ (by the
Sparsification Lemma, see \cite{ImpagliazzoPZ01}), we will create an instance
$[G=(V,E),k]$ of the unweighted \KC\ problem, where $k=\sqrt{n}$ and $r=4\cdot
c^{\sqrt{n}}$ for an appropriate constant $c$. (To simplify notation, we assume
without loss of generality that $\sqrt{n}$ is an integer). 

We first group the clauses of $\phi$ into $\sqrt{n}$ equal-sized groups
$F_1,\dots,F_{\sqrt{n}}$. As a result, each group involves $O(\sqrt{n})$
variables, so there are $2^{O(\sqrt{n})}$ possible assignments to the variables
of each group. Select $c$ appropriately so that  each group $F_i$ has at most
$c^{\sqrt{n}}$ possible partial assignments $\phi_j^i$ for the variables of
clauses in $F_i$.  

We then create for each $i\in\{1,\ldots,n\}$, a set $P_i$ of at most
$c^{\sqrt{n}}$ vertices $p_1^i,\dots,p_{c^{\sqrt{n}}}^i$, such that each vertex
of $P_i$ represents a partial assignment to the variables of $F_i$ that
satisfies all clauses of $F_i$. We add two \emph{guard} vertices $g_i^1,g_i^2$,
attaching them to all vertices of $P_i$ by paths of length $r=4\cdot
c^{\sqrt{n}}$.  Next, for each $i\in\{1,\ldots,\sqrt{n}\}$, we create another
pair of vertices $a_i,b_i$, connecting $a_i$ to each vertex $p_l^i$ by a path
of length $c^{\sqrt{n}}+l$, while $b_i$ is connected to each vertex $p_l^i$ by
a path of length $2\cdot c^{\sqrt{n}}-l+1$. Now each $P_i$ contains all
vertices $a_i,b_i,g_i^1,g_i^2$ and each $p_l^i,\forall
l\in\{1,\ldots,c^{\sqrt{n}}\}$.

Finally, for any two \emph{conflicting} partial assignments
$\phi_l^i,\phi_o^j$, with $l,o\in\{1,c^{\sqrt{n}}\}$ and
$i,j\in\{1,\sqrt{n}\}$, i.e.\ two partial assignments that assign conflicting
values to at least one variable, we create a vertex $u_{l,o}^{i,j}$ that we
connect to vertices $a_i,b_i$ and $a_j,b_j$: if $p_l^i\in P_i$ is the vertex
corresponding to $\phi_l^i$ and $p_o^j\in P_j$ is the vertex corresponding to
$\phi_o^j$, then vertex $u_{l,o}^{i,j}$ is connected to $a_i$ by a path of
length $3\cdot c^{\sqrt{n}}-l+1$ and to $b_i$ by a path of length
$2\cdot c^{\sqrt{n}}+l$, as well as to $a_j$ by a path of length $3\cdot
c^{\sqrt{n}}-o+1$ and $b_j$ by a path of length $2\cdot c^{\sqrt{n}}+o$. See Figure
\ref{fig:eth_TD} for an illustration.  
\begin{figure}[htbp]
 \centerline{\includegraphics[width=120mm]{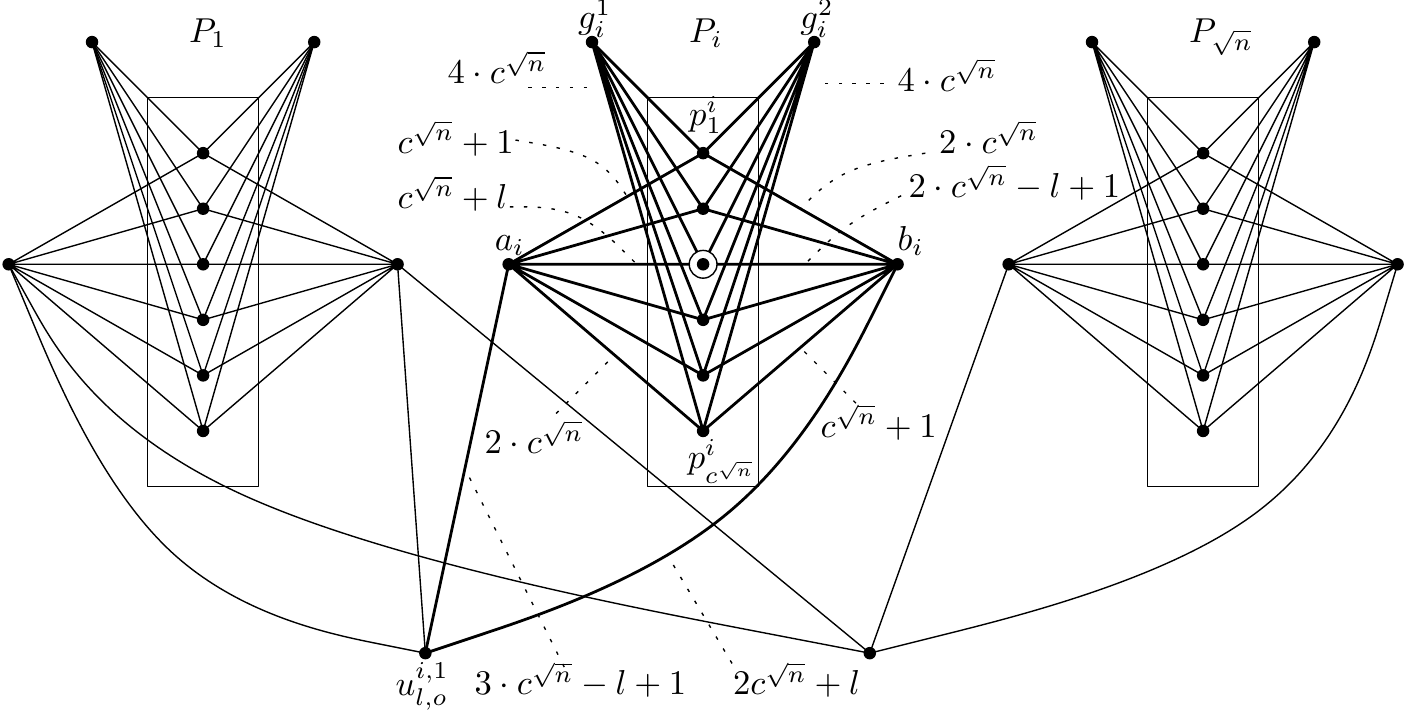}}
 \caption{A general picture of graph $G$. Note all straight lines denote paths of length equal to the number shown by dotted lines, while the circled vertex is $p_l^i$.}
 \label{fig:eth_TD}
\end{figure}
\begin{lemma}\label{thm:eth_TD_FWD}
 If $\phi$ has a satisfying assignment, then there exists a $(k,r)$-center in $G$ of size $k=\sqrt{n}$ and $r=4\cdot c^{\sqrt{n}}$.
\end{lemma} 
\begin{proof} Consider the satisfying assignment for $\phi$ and let
$\phi_{l_i}^i$, with $l_i\in\{1,\ldots,c^{\sqrt{n}}\}$ and
$i\in\{1,\ldots,\sqrt{n}\}$, be the restriction of that assignment for all
variables appearing in clauses of group $F_i$. We claim the set $K$, consisting
of all vertices $p_{l_i}^i$ corresponding to $\phi_{l_i}^i$, is a
$(k,r)$-center for $G$ with $k=\sqrt{n}$ and $r=4\cdot c^{\sqrt{n}}$. Since one
vertex is selected from each $P_i$ in $G$, all the guards $g_i^1,g_i^2$ and
vertices $a_i,b_i$ are within distance $4\cdot c^{\sqrt{n}}$ from selected
vertex $p_{l_i}^i$, as well as all other vertices $p_{l'}^i$. For vertices $u_{l_i,l_j}^{i,j}$, observe that selected
vertex $p_{l_i}^i$ is at distance $c^{\sqrt{n}}+l_i+3\cdot
c^{\sqrt{n}}-l_i+1=4\cdot c^{\sqrt{n}}+1$ from $u_{l_i,l_j}^{i,j}$, through
either $a_i$ or $b_i$, while vertex $p_{l_j}^j$ is at distance
$c^{\sqrt{n}}+l_j+3\cdot c^{\sqrt{n}}-l_j+1=4\cdot c^{\sqrt{n}}+1$ from
$u_{l_i,l_j}^{i,j}$, through either $a_j$ or $b_j$, only if the corresponding
partial assignments $\phi_{l_i}^i$ and $\phi_{l_j} ^j$ are conflicting.  As
$\phi_{l_i}^i$, $\phi_{l_j}^j$ are parts of a satisfying assignment for $\phi$,
however, this will not be the case and at least one path from
$u_{l_i,l_j}^{i,j}$ to either $p_{l_i}^i$, or $p_{l_j}^j$, will be of length
$\le4\cdot c^{\sqrt{n}}$.  \end{proof}
\begin{lemma}\label{thm:eth_TD_BWD}
 If there exists a $(k,r)$-center in $G$ of size $k=\sqrt{n}$ and $r=4\cdot c^{\sqrt{n}}$, then $\phi$ has a satisfying assignment.
\end{lemma}
\begin{proof} Let $K\subset V$ be the $(k,r)$-center in $G$, with $r=4\cdot
c^{\sqrt{n}}$ and $k=\sqrt{n}$. As $|K|=k=\sqrt{n}$ and all guard vertices
$g_i^1,g_i^2,\forall i\in\{1,\ldots,\sqrt{n}\}$, must be within distance at
most $4\cdot c^{\sqrt{n}}$ from some vertex selected in $K$, the set $K$ must
contain exactly one vertex from each $P_i$, or
$K=\{p_{l_1}^1,\dots,p_{l_{\sqrt{n}}}^{\sqrt{n}}\}$, for some
$l_i\in\{1,\ldots,c^{\sqrt{n}}\}$ and for each $i\in\{1,\ldots,\sqrt{n}\}$.  We
set the assignment for $\phi$ to consist of all partial assignments
$\phi_{l_i}^i$, with $i\in\{1,\ldots,\sqrt{n}\}$, corresponding to vertices
$p_{l_i}^i\in K$ and claim that this is a valid assignment that satisfies
$\phi$.  It is not hard to see that, if the assignment is valid, then it
satisfies the clause, as we have only listed partial assignments that satisfy
all clauses of each group. To see that the assignment does not assign
conflicting values to any variable, suppose for contradiction that there are
two conflicting partial assignments $\phi_{l_i}^i,\phi_{l_j}^j$ and a vertex
$u_{l_i,l_j}^{i,j}\in G$ with paths to $a_i,b_i\in P_i$ and $a_j,b_j\in P_j$,
where we have $p_{l_i}^i,p_{l_j}^j\in K$. The distance from $u_{l_i,l_j}^{i,j}$
to $p_{l_i}^i$ is $3\cdot c^{\sqrt{n}}-l_ i+1+c^{\sqrt{n}}+l_i>4\cdot
c^{\sqrt{n}}$, via either $a_i$ or $b_i$, while its distance to $p_{l_j}^j$ is
also $3\cdot c^{\sqrt{n}}-l_j+1+c^{\sqrt{n}}+l_j>4\cdot c^{\sqrt{n}}$, via
either $a_j,b_j$, meaning $u_{l_i,l_j}^{i,j}$ is not covered by $K$, giving a
contradiction.  \end{proof}
\begin{lemma}\label{thm:eth_TD_treedepth} The tree-depth of $G$ is
$4\sqrt{n}+\lceil\log(4\cdot c^{\sqrt{n}})\rceil+1 = O(\sqrt{n})$.
\end{lemma}
\begin{proof} Again we use the definition of tree-depth used in Lemma
\ref{lem:td-diam}. Consider the graph $G$ after removal of all guard vertices
$g_{i}^1,g_{i}^2$ and all $a_i,b_i,\forall i\in[1,\sqrt{n}]$.  The graph now
consists of a number of sub-divided stars, centered either on some vertex
$u_{l,o}^{i,j}$, or some $p_{l}^i$, while the maximum distance from each such
center to a leaf vertex in the star is $4\cdot c^{\sqrt{n}}-1$, for a path
connecting $p_l^i$ to a guard $g_i^1,g_i^2$, omitting the final edge due to
removal of $g_i^1,g_i^2$.  The claim follows since paths of length $n$ have
tree-depth exactly $\lceil\log(n+1)\rceil$ (this can be shown by repeatedly
removing the middle vertex of a path). By the definition of tree-depth, after
removal of $4\sqrt{n}$ vertices from $G$, the maximum tree-depth of each
resulting disconnected component is $\lceil\log(4\cdot c^{\sqrt{n}})\rceil =
\lceil\sqrt{n}\cdot\log(c)+\log(4)\rceil$.  \end{proof}
\begin{theorem}\label{thm:eth_TD_LB}
 If \textsc{$(k,r)$-Center} can be solved in $2^{o(\textrm{td}^2)}\cdot n^{O(1)}$ time, then \textsc{3-SAT} can be solved in $2^{o(n)}$ time.
\end{theorem}
\begin{proof}
 Suppose there is an algorithm for the \textsc{$(k,r)$-Center} problem with running time $2^{o(\textrm{td}^2)}$. Given an instance $\phi$ of \textsc{3-SAT}, we use the above construction to create an instance $[G,k,r]$ of \textsc{$(k,r)$-Center}, with $k=\sqrt{n}$ and $r=4\cdot c^{\sqrt{n}}$, in time $O(\sqrt{n}\cdot c^{\sqrt{n}} + c^{2\sqrt{n}})$. As, by Lemma \ref{thm:eth_TD_treedepth}, we have $\textrm{td}(G)=O(\sqrt{n})$, using the supposed algorithm for \textsc{$(k,r)$-Center}, we can decide whether $\phi$ has a satisfying assignment in time $2^{o(\textrm{td}^2)}\cdot n^{O(1)}=2^{o(n)}$.
\end{proof}

%% file: tw-approx.tex
\section{Treewidth: FPT approximation scheme}\label{sec:fptas} 

In this section we present an FPT approximation \emph{scheme} (FPT-AS) for \KC\
parameterized by $\tw$. Given as input a weighted graph $G=(V,E)$, $k,r\in\mathbb{N}^+$
and an arbitrarily small error parameter $\epsilon>0$, our algorithm is able to
return a solution that uses a set of $k$ centers $K$, such that all other
vertices are at distance at most $(1+\epsilon)r$ from $K$, or to correctly
conclude that no $(k,r)$-center exists. The running time of the algorithm is
$O^*((\tw/\epsilon)^{O(\tw)})$, which (for large $r$) significantly
out-performs any \emph{exact} algorithm for the problem
(even for the unweighted case and more restricted parameters, as in Theorems~\ref{thm:W_hard_VCk} and~\ref{thm:W_hard_FVSk}), while only sacrificing a small
$\epsilon$ error in the quality of the solution.

Our algorithm will rely heavily on a technique introduced in \cite{Lampis14}
to approximate problems which are W-hard by
treewidth (see also \cite{AngelBEL16}, as well as \cite{Bodlaender2012} for a similar approach). The idea is that, if the hardness of the problem is due to the fact
that the DP table needs to store $\tw$ large numbers (in our
case, the distances of the vertices in the bag from the closest center), we can
significantly speed up the algorithm if we replace all entries by the closest
integer power of $(1+\delta)$, for some appropriately chosen $\delta$. This will reduce the table size from (roughly) $r^\tw$ to $(\log_{(1+\delta)} r)^\tw$.
The problem now is that a DP performing calculations on its entries will, in
the course of its execution, create values which are not integer powers of
$(1+\delta)$, and will therefore have to be ``rounded'' to retain the table
size. This runs the risk of accumulating rounding errors, but we manage to show that the error on any entry of the rounded table can be bounded by a function of the height of its corresponding bag, then using a
theorem of \cite{BodlaenderH98} stating that any tree decomposition can be
balanced so that its width remains almost unchanged, yet its total
height becomes $O(\log n)$. Beyond these ideas, which are for the most part
present in \cite{Lampis14}, we will also need a number of problem-specific
observations, such as the fact that we can pre-process the input by taking the
metric closure of each bag, and in this way avoid some error-prone
arithmetic operations.

To obtain the promised algorithm we thus do the following:
first we re-cast the problem as a distance-labeling problem (similarly to Section \ref{sec_cw_DP}) and formulate an exact
treewidth-based DP algorithm running in time $O^*(r^{O(\tw)})$. This algorithm
follows standard techniques (indeed, a faster version is given in
\cite{BorradaileL16}) but we give it here to ensure that we have a solid foundation upon which to build the
approximation algorithm.  We then apply the rounding procedure sketched above,
and prove its approximation ratio by using the balancing theorem of
\cite{BodlaenderH98} and indirectly comparing the value produced by the
approximation algorithm with the value that would have been produced by the
exact algorithm.
\subparagraph{Distance-labeling:} We give an equivalent formulation of \KC\
that will be more convenient to work with in the remainder, similarly to
Section \ref{sec_cw_DP}. For an edge-weighted graph $G=(V,E)$, a
distance-labeling function is a function  $\dl: V\to \{0,\ldots,r\}$. We say
that $u\in V$ is \emph{satisfied} by $\dl$, if $\dl(u)=0$, or if there exists
$v\in N(u)$ such that $\dl(u)\ge dl(v) + w((v,u))$. We say that $\dl$ is \emph{valid}
if all vertices of $V$ are satisfied by $\dl$, and we define the \emph{cost} of $\dl$
as $|\dl^{-1}(0)|$. The following lemma , 
shows the equivalence between the two formulations:
\begin{lemma}\label{lem:dl} An edge-weighted graph $G=(V,E)$ admits a
$(k,r)$-center if and only if it admits a valid distance-labeling function
$\dl: V\to \{0,\ldots,r\}$ with cost $k$.  \end{lemma}
\begin{proof} For one direction, if there is a $(k,r)$-center $K$, we construct
a function $\dl$ by assigning as value to each vertex $v\in V$ the distance
from the closest vertex of $K$ to $v$.  This function is valid and has a cost
of $|K|=k$. For the converse direction, if such a function exists, we set
$K=\dl^{-1}(0)$. It is now not hard to verify that all vertices $v$ are at
distance at most $\dl(v)$ from a vertex of $K$. \end{proof}
\subparagraph{Exact-algorithm:} We formulate an exact algorithm which,
given an edge-weighted graph $G=(V,E)$, finds the minimum cost of a
distance-labeling function for $G$. By Lemma \ref{lem:dl}, this is equivalent
to \KC. We remark that the algorithm essentially reproduces the ideas of \cite{BorradaileL16}, and can be made to run in $O^*((2r+1)^\tw)$ if one uses fast subset convolution for the Join nodes (the naive
implementation would need time $O^*((2r+1)^{2\tw})$).
\begin{theorem}\label{thm:tw-exact} There is an algorithm which, given an
edge-weighted graph $G=(V,E)$ and $r\in\mathbb{N}^+$, computes the minimum cost of
any valid distance labeling of $G$ in time $O^*(r^{O(\tw)})$.  \end{theorem}
\begin{proof} The algorithm uses standard techniques, so we sketch some of the
details. The idea is that for any bag $B_t$ of a tree decomposition of $G$ we
maintain a table $D_t: \left( (B_t\to \{0,\ldots,r\}) \times (B_t\to \{0,1\})
\right) \to \{0,\ldots,n\}\cup\{\infty\}$. Let $B_t^\downarrow$ denote the
vertices of $G$ contained in bags in the sub-tree rooted at $B_t$. Then, the
meaning of $D_t$ is that for each distance labeling $\dl_t: B_t \to
\{0,\ldots,r\}$ of the vertices of $B_t$ and for each subset $S\subseteq B$ it
contains the minimum cost of any distance labeling of
$G[B_t^\downarrow]$ which satisfies all vertices of $B_t^\downarrow$, except
perhaps those of $B_t\setminus S$, and which agrees with $\dl_t$ on $B_t$.

Using this table we can now formulate a DP algorithm as follows:
\begin{itemize}
\item For a Leaf node $B_t$ that only contains a single vertex $u$, we set
$D_t(\dl,\{u\})=1$, if $\dl(u)=0$ and $D_t(\dl,\{u\})=\infty$, otherwise.  We set
$D_t(\dl,\emptyset)=0$ for any other distance labeling of $u$.
\item For a Forget node $B_t$ that forgets a vertex $u$ and has a child node
$B_{t'}$, we set $D_t(\dl,S) = \min_{\dl'} D_{t'}(\dl',S\cup\{u\})$, where the
minimum ranges over all labeling functions $\dl'$ on $B_t'$ that agree with
$\dl$ on $B_t=B_{t'}\setminus\{u\}$.  In other words, we only retain partial
solutions which have already satisfied the vertex $u$ that we forget.
\item For an Introduce node $B_t$ that introduces a vertex $u$ and has a child
node $B_{t'}$ (so $B_t=B_{t'}\cup\{u\}$), we consider every labeling $\dl'$ of
$B_{t'}$ and every $S'\subseteq B_{t'}$ such that $D_{t'}(\dl',S')\neq \infty$.
For each such $\dl', S'$, and for each $i\in\{0,\ldots,r\}$, we construct a
$\dl$ function for $B_t$ which agrees with $\dl'$ on $B_{t'}$ and sets
$\dl(u)=i$. We compute $S$ to be the set that contains $S'$ as well as all
vertices of $B_{t'}$ satisfied by $\dl(u)$, and $u$ if it is satisfied by a
vertex of $B_{t'}$ or $\dl(u)=0$. More precisely, we add $v\in B_{t'}$ to $S$
if $\dl(v) \ge \dl(u)+w((u,v))$, and we add $u$ to $S$ if $\dl(u)=0$, or if
there exists $v\in B_{t'}$ for which $\dl(u) \ge \dl(v) + w((v,u))$.  We set
$D_t(\dl,S) = D_{t'}(\dl',S')+1$ if $\dl(u)=0$, and $D_t(\dl,S) =
D_{t'}(\dl',S')+1$ otherwise. All other entries of $D_t$ are set to $\infty$.
\item For a Join node $B_t$ with two children $B_{t_1},B_{t_2}$, for each
labeling function $\dl$ on $B_t$ and each $S\subseteq B_t$ we set $D_t(\dl,S) =
\min_{S_1\subseteq S} D_{t_1}(\dl,S_1)+D_{t_2}(\dl,S\setminus S_1) -
|\dl^{-1}(0)\cap B_t|$. Note that we subtract $|\dl^{-1}(0)\cap B_t|$ to
avoid double-counting the vertices of $B_t$.\end{itemize} 
\end{proof}
We now describe an approximation algorithm based on the exact DP algorithm of Theorem~\ref{thm:tw-exact}. We
make use of a result of~\cite{BodlaenderH98} (Theorem~\ref{thm:balance}) and of Lemma~\ref{lem:metric}, both given in the sequel.
\begin{theorem}\label{thm:balance}\cite{BodlaenderH98} There is an algorithm
which, given a tree decomposition of width $w$ of a graph on $n$ nodes,
produces a decomposition of the same graph with width at most $3w+2$ and height
$O(\log n)$ in polynomial time. \end{theorem} 
\begin{lemma}\label{lem:metric} Let $G=(V,E)$ be an edge-weighted graph,~$\mathcal{T}$ a tree decomposition of~$G$, and $u,v\in V$ two vertices that
appear together in a bag of~$\mathcal{T}$. Let~$G'$ be the graph obtained from~$G$ by adding~$(u,v)$ to~$E$ (if it does not already exist) and setting
$w((u,v)) = d_G(u,v)$. Then~$\mathcal{T}$ is a valid decomposition of ~$G'$, and $\forall k,r$, $G'$ admits a $(k,r)$-center if and only if~$G$ does.
\end{lemma}
 \begin{proof}
 To see that $G'$ admits a $(k,r)$-center if and only if $G$ does, observe that
 adding the weighted edge $(u,v)$ did not change the distances between any two
 vertices. To see that $\mathcal{T}$ remains a valid decomposition we recall
 that $u,v$ appear together in a bag of $\mathcal{T}$.  
 \end{proof}
Let us also give an approximate version of the distance labeling problem we
defined above, for a given error parameter $\epsilon>0$.  Let $\delta>0$ be
some appropriately chosen secondary parameter (we will eventually set $\delta\approx
\frac{\epsilon}{\log n}$). We define a \emph{$\delta$-labeling} function $\dl_\delta$
as a function from $V$ to $\{0\}\cup\{ (1+\delta)^i\ |\ i\in \mathbb{N},
(1+\delta)^i\le (1+\epsilon) r\}$. In words, such a function assigns (as
previously) a distance label to each vertex, with the difference that now all
values assigned are integer powers of $(1+\delta)$, and the maximum value is at
most $(1+\epsilon)r$.  We now say that a vertex~$u$ is \emph{$\epsilon$-satisfied} if
$\dl_\delta(u)=0$ or, for some $v\in N(u)$ we have $\dl_\delta(u)\ge
\dl_\delta(v)+ \frac{w((v,u))}{1+\epsilon}$. As previously, we say that
$\dl_\delta$ is \emph{valid} if all vertices are $\epsilon$-satisfied, and define its
cost as $|\dl_\delta^{-1}(0)|$. The following Lemma~\ref{lem:approx-label} shows the equivalence of a valid $\delta$-labeling function of cost $k$ and a $(k,(1+\epsilon)^2r)$-center for $G$ and using it we conclude the proof of Theorem~\ref{thm:tw_approx},  stating the main result of this section.
\begin{lemma}\label{lem:approx-label} If for a weighted graph $G=(V,E)$ and any
$k,r,\delta,\epsilon>0$, there exists a valid $\delta$-labeling function with
cost $k$, then there exists a $(k,(1+\epsilon)^2r)$-center for $G$.
\end{lemma}
\begin{proof}
Again we set $K=\dl_\delta^{-1}(0)$, where $\dl_\delta$ is a valid
$\delta$-labeling function. Recall that $\dl_\delta$ assigns value
$(1+\delta)^i$, for $i\in\mathbb{N}$ to all vertices of $V\setminus K$. We
prove by induction on $i$, that a vertex $u$ for which
$\dl_\delta(u)=(1+\delta)^i$ is at distance at most $(1+\epsilon)\dl_\delta(u)$
from $K$.  First, if $\dl_\delta(u)=(1+\delta)$ and $u$ is satisfied, then $u$
has a neighbor $v$ with $\dl_\delta(u) \ge \dl_\delta(v)+
\frac{w((v,u))}{1+\epsilon}$. Because $w((v,u))>0$, we must have $v\in K$, so
we conclude that $d(K,u) \le (1+\epsilon)\dl_\delta(u)$. Similarly, if we have
$\dl_\delta(u)=(1+\delta)^{i+1}$ and $u$ is satisfied, there must exist $v\in
N(u)$ with $\dl_\delta(u) \ge \dl_\delta(v) + \frac{w((v,u))}{1+\epsilon}$, so
$\dl_\delta(v) \le (1+\delta)^i$. By the inductive hypothesis, it is $d(K,v)\le
(1+\epsilon) \dl_\delta(v)$. So $d(K,u) \le w((v,u)) + d(K,v) \le
(1+\epsilon)\left(\dl_\delta(v) + \frac{w((v,u))}{1+\epsilon} \right)\le
(1+\epsilon)\dl_\delta(u)$. Since for all $u\in V$ we have $\dl_\delta(u)\le
(1+\epsilon)r$, all vertices are at distance at most $(1+\epsilon)^2r$ from
$K$.  \end{proof}
We are now ready to state the main result of this section:
\begin{theorem}\label{thm:tw_approx}
There is an algorithm which, given a weighted instance of \KC, $[G,k,r]$, a
tree decomposition of $G$ of width $\tw$ and a parameter $\epsilon>0$, runs in
time $O^*((\tw/\epsilon)^{O(\tw)})$ and either returns a
$(k,(1+\epsilon))$-center of $G$, or correctly concludes that $G$ has no
$(k,r)$-center.  
\end{theorem}
\begin{proof} Our algorithm will follow along the same lines as the algorithm
of Theorem \ref{thm:tw-exact}, with the major difference that the labeling
functions we consider are $\delta$-labeling functions, for a value of $\delta$
that we define further below. More precisely, if $\Sigma_r:= \{0\}\cup\{
(1+\delta)^i\ |\ i\in \mathbb{N}, (1+\delta)^i\le (1+\epsilon)r\}$, then for
each bag $B_t$ of $\mathcal{T}$ we define a table as a function $D^\delta_t:
\left( (B_t\to \Sigma_r) \times (B_t\to \{0,1\}) \right) \to
\{0,\ldots,n\}\cup\{\infty\}$.  Let $B_t^\downarrow$ denote the vertices of $G$
contained in bags in the sub-tree rooted at $B_t$. The meaning of $D^\delta_t$
here is similar to that of $D_t$ in the algorithm of Theorem
\ref{thm:tw-exact}. For each $\delta$-distance labeling $\dl_\delta: B_t\to
\Sigma_r$, and for each $S\subseteq B_t$, $D^\delta_t$ tells us what is the
minimum cost of any $\delta$-distance labeling of $B_t^\downarrow$ which
satisfies all vertices $B_t^\downarrow$, except perhaps those of $B_t\setminus
S$, and which agrees with $\dl_\delta$ on $B_t$.

Before we proceed, we pre-process the input using Theorem \ref{thm:balance} and
Lemma \ref{lem:metric}. In particular, we construct a tree decomposition
$\mathcal{T'}$ of width at most $3\tw+2$ and height at most $H\le c\log n$ for
some constant $c$, and then for any two vertices $u,v$ which appear together in a bag of
$\mathcal{T'}$, we add the edge $(u,v)$ to $G$, as in Lemma
\ref{lem:metric}. We define the height of a bag $B_t$ of the decomposition
recursively as follows: the height of a Leaf bag is $1$, for any other bag
$B_t$ its height is $1$ plus the maximum of the heights of its children.
Clearly, under this definition the root bag has height $H$ and all other bags
have height $< H$. We now set $\delta = \frac{\epsilon}{2H} =
\Theta\left(\frac{\epsilon}{\log n}\right)$. Observe that because of this
setting we have for all $h\le H$ that $(1+\delta)^h \le
(1+\frac{\epsilon}{2H})^H \le e^{\epsilon/2} \le (1+\epsilon)$ for sufficiently
small $\epsilon$ (e.g. it suffices to assume without loss of generality
$\epsilon<1/4$). Our goal will be to return a $(k,(1+\epsilon)^2r)$-center, if
a $(k,r)$-center exists by producing a $\delta$-labeling and invoking Lemma
\ref{lem:approx-label}. The approximation ratio can then be brought down to
$(1+\epsilon)$ by adjusting $\epsilon$ appropriately.

Our algorithm now follows the algorithm of Theorem \ref{thm:tw-exact}, with the
only difference that the satisfaction of a vertex follows the definition of
$\epsilon$-satisfaction (that is, we effectively divide edge weights by
$(1+\epsilon)$). Specifically, we have the following:
\begin{itemize}
\item For a Leaf node $B_t$ that only contains a single vertex $u$, we set
$D_t^\delta(\dl_\delta,\{u\})=1$, if $\dl_\delta(u)=0$ and
$D_t^\delta(\dl_\delta,\{u\})=\infty$, otherwise.  We set
$D_t^\delta(\dl_\delta,\emptyset)=0$ for any other distance labeling of $u$.
\item For a Forget node $B_t$ that forgets a vertex $u$ and has a child node
$B_{t'}$, we set $D_t^\delta(\dl_\delta,S) = \min_{\dl_\delta'}
D_{t'}^\delta(\dl_\delta',S\cup\{u\})$ where the minimum ranges over all
labeling functions $\dl_\delta'$ on $B_t'$ that agree with $\dl_\delta$ on
$B_t=B_{t'}\setminus\{u\}$.  In other words, we only retain partial solutions
which have already satisfied the vertex $u$ that we forget.
\item For an Introduce node $B_t$ that introduces a vertex $u$ and has a child
node $B_{t'}$ (so $B_t=B_{t'}\cup\{u\}$), we consider every labeling
$\dl_\delta'$ of $B_{t'}$ and every $S'\subseteq B_{t'}$ such that
$D_{t'}^\delta(\dl_\delta',S')\neq \infty$.  For each such $\dl_\delta', S'$,
and for each $i\in\Sigma_r$ we construct a $\dl_\delta$ function for $B_t$
which agrees with $\dl_\delta'$ on $B_{t'}$ and sets $\dl_\delta(u)=i$. We
compute $S$ to be the set that contains $S'$ as well as all vertices of $B_{t'}$
satisfied by $\dl(u)$, and $u$ if it is satisfied by a vertex of $B_{t'}$ or
$\dl_\delta(u)=0$.  More precisely, we add $v\in B_{t'}$ in $S$ if
$\dl_\delta(v) \ge \dl_\delta(u)+\frac{w((u,v))}{1+\epsilon}$, and we add $u$
to $S$ if $\dl_\delta(u)=0$, or if there exists $v\in B_{t'}$ for which
$\dl_\delta(u) \ge \dl_\delta(v) + \frac{w((v,u))}{1+\epsilon}$.  We set
$D_t^\delta(\dl_\delta,S) = D_{t'}^\delta(\dl_\delta',S')+1$ if $\dl(u)=0$, and
$D_t^\delta(\dl,S) = D_{t'}^\delta(\dl',S')+1$ otherwise. All other entries of
$D_t$ are set to $\infty$.
\item For a Join node $B_t$ with two children $B_{t_1},B_{t_2}$, for each
labeling function $\dl_\delta$ on $B_t$ and each $S\subseteq B_t$, we set
$D_t(\dl_\delta,S) = \min_{S_1,S_2\subseteq S: S_1\cup S_2 = S}
D_{t_1}(\dl_\delta,S_1)+D_{t_2}(\dl_\delta,S_2) - |\dl_\delta^{-1}(0)\cap
B_t|$. 
\end{itemize}
To establish correctness of the above algorithm we proceed in two directions.
First, we show that for any bag $B_t$ we have $D_t^\delta(\dl_\delta,S)\le i$,
if and only if there exists a $\delta$-labeling of $B_t^\downarrow$ which
satisfies all vertices of $B_t^\downarrow$, except perhaps those of
$B_t\setminus S$, and has cost $i$. This can be shown by a standard inductive
argument: the statement is clearly true on Leaf bags, and it can be shown to be
true on other bags if we assume it to be true on their children. As a result,
by examining the table of the root bag at the end of the algorithm's execution,
we can correctly compute the minimum cost of a valid $\delta$-labeling. If
that cost is at most $k$, by Lemma \ref{lem:approx-label} we obtain a
$(k,(1+\epsilon)^2r)$-center of $G$.

It is now more interesting to establish the converse direction: we would like
to show that if a $(k,r)$-center exists then there exists a $\delta$-labeling
of $G$, which will therefore be found by the algorithm. The reason that this
direction is more difficult is that the converse of Lemma
\ref{lem:approx-label} does not hold for any $\delta$; indeed if $\delta$ is
too small it could be the case that a $(k,r)$-center exists, but the graph does
not admit a $\delta$-labeling.  We will thus need to prove that the $\delta$ we
have selected is sufficient.  In the remainder, suppose there exists a $\dl:
V\to \{0,\ldots,r\}$ which encodes a $(k,r)$-center of $G$. 

We will establish by induction the following property: for any bag $B_t$ of the
decomposition such that the height of $B_t$ is $h$, let $\dl_t$ denote the
restriction of $\dl$ to $B_t$ and $S$ the vertices of $B_t$ which are satisfied
in $G[B_t^\downarrow]$ by $\dl$. We will show that there always exists
$\dl_\delta: B_t\to \Sigma_r$, $S_\delta\supseteq S$ such that
$D_t^\delta(\dl_\delta,S_\delta) \le |\dl^{-1}(0)\cap B_t^\downarrow|$ and for
all $u\in B_t$ we have $\frac{\dl(u)}{(1+\delta)^h} \le \dl_\delta(u) \le
(1+\delta)^h\dl(u)$.

We observe that the property holds for Leaf bags $B_t=\{u\}$, because the
algorithm considers both the case where $u\in K$ (in which
$\dl(u)=\dl_\delta(u)=0$), as well as every possible labeling of $u$ with an
integer power of $(1+\delta)$, hence it considers a labeling such that
$\dl(u),\dl_\delta(u)$ are at most a multiplicative factor $(1+\delta)$ apart.
The property can easily be shown inductively for Forget and Join bags, if we
assume the property for their children, because whenever $\dl_\delta(u)\in
[\frac{\dl(u)}{(1+\delta)^h}, (1+\delta)^h\dl(u)]$, this implies
$\dl_\delta(u)\in [\frac{\dl(u)}{(1+\delta)^{h+1}}, (1+\delta)^{h+1}\dl(u)]$.
Hence, the interesting case is Introduce bags.

Consider an Introduce bag $B_t$ with child $B_{t'}$, so that $B_t =
B_{t'}\cup\{u\}$. Let $S\subseteq B_t$ be the set of vertices satisfied by
$\dl$ in $G[B_t^\downarrow]$, and similarly, let $S'\subseteq B_{t'}$ be the
corresponding set in $B_{t'}$. Clearly, $S'\subseteq S$. We now claim that at
least one of the following three must be true: (i) $\dl(u)=0$, (ii)
$S=S'\cup\{u\}$, (iii) $u\not\in S$. To see this, suppose for contradiction
that $\dl(u)\neq 0$ and $S$ contains all of $S'$, $u$ and some vertex $v_1\in
B_t\setminus S'$. Then $\dl(v_1) \ge \dl(u) + w((u,v_1))$, because $v_1\in S$
but $v_1\not\in S'$.  Furthermore, there exists $v_2\in B_{t'}$ such that
$\dl(u) \ge \dl(v_2) + w((v_2,u))$, because $u\in S$ and $\dl(u)\neq 0$.
Hence, $\dl(v_1) \ge \dl(v_2) + w((v_2,u)) + w((u,v_1)) \ge \dl(v_2) +
w((v_2,v_1))$, where the last inequality follows because we have used Lemma
\ref{lem:metric} to take the metric closure of $B_{t'}$. But the last
inequality implies $v_1\in S'$, contradiction. 

We therefore need to establish that, for each of the three cases above, the
algorithm will produce an entry $D^\delta_t(\dl_\delta,S_\delta)$ with
$S\subseteq S_\delta$ and $\dl_\delta(u)$ which is at most a factor of
$(1+\delta)^h$ apart from $\dl(u)$. Assume by the inductive hypothesis that
$D^\delta_{t'}(\dl_\delta',S_\delta')\le i$, for some $S_\delta'\supseteq S'$,
$\dl_\delta'$ which has $\forall v\in B_{t'}, \dl_\delta'(v)\in
[\frac{\dl(u)}{(1+\delta)^{h-1}}, (1+\delta)^{h-1}\dl(u) ] $, and for
$i=|\dl^{-1}(0)\cap B_{t'}^\downarrow|$.  

We have:
\begin{enumerate}
\item[(i)] If $\dl(u)=0$, the algorithm will consider a $\dl_\delta$ which sets
$\dl_\delta(u)=0$ and agrees with $\dl'_\delta$ on $B_{t'}$. From the entry
$D_{t'}^\delta(\dl'_\delta,S'_\delta)\le i$ it will construct an entry
$D_{t}(\dl_\delta,S_\delta)\le i+1$, with $S\subseteq S_\delta$, because for
each $v\in S\setminus S'$, we have $\dl(v)\ge w((u,v))$, therefore,
$\dl_\delta(v) \ge \frac{\dl(v)}{(1+\delta)^{h-1}} \ge
\frac{w((u,v))}{1+\epsilon}$, where we have used the fact that $(1+\delta)^h\le
1+\epsilon$.
\item[(ii)] Here we assume $\dl(u)\neq 0$ and $u\in S$. There exists therefore
$v\in B_{t'}$ with $\dl(u) \ge \dl(v) + w((v,u))$. Let $r:=
(1+\delta)^{h-1}\dl(u)$. The algorithm will consider the function $\dl_\delta$
which agrees with $\dl'_\delta$ on $B_{t'}$ and sets $\dl_\delta(u) =
(1+\delta)^{\lceil \log_{(1+\delta)} r \rceil}$. We have $\dl_\delta(u) \ge
(1+\delta)^{h-1}\dl(u) \ge (1+\delta)^{h-1}(\dl(v)+w((v,u))\ge \dl_\delta(v) +
\frac{w((v,u))}{1+\epsilon}$ , thus the algorithm will (correctly) add $u$
to $S_\delta'$ to obtain $S_\delta\supseteq S$.  Furthermore,
$(1+\delta)^{h-1}\dl(u) \le \dl_\delta(u) \le (1+\delta)^{h}\dl(u)$.
\item[(iii)] Here the interesting case is if $S\setminus S'\neq\emptyset$.
Consider a $v\in S\setminus S'$. We have $\dl(v) \ge \dl(u) + w((u,v))$, since
$v\in S\setminus S'$. Let $r:= \frac{\dl(u)}{(1+\delta)^h}$. The algorithm will
consider the $\dl_\delta$ which agrees with $\dl'_\delta$ on $B_{t'}$ and sets
$\dl_\delta(u) = (1+\delta)^{\lceil \log_{(1+\delta)} r \rceil}$. We have
$\dl_\delta(v) \ge \frac{\dl(v)}{(1+\delta)^{h-1}} \ge
\frac{\dl(u)}{(1+\delta)^{h-1}} + \frac{w((u,v))}{1+\epsilon} \ge \dl_\delta(u)
+  \frac{w((u,v))}{1+\epsilon}$. Hence, the algorithm will extend $S_\delta'$
by adding to it all elements of $S\setminus S'$. Furthermore,
$(1+\delta)^{-h+1}\dl(u) \ge \dl_\delta(u) \ge (1+\delta)^{-h}\dl(u)$.
\end{enumerate}
We can therefore conclude that whenever a $(k,r)$-center exists, we will be
able to recover from the root bag of the algorithm's DP table a
$(k,(1+\epsilon)^2r)$-center with the same or lower cost. In particular, by the
property we established above, if $B_r$ is the root bag, there exists
$\dl_\delta$ such that $D_r^\delta(\dl_\delta,B_r) \le k$, and for all $u\in
B_r$ we have $\dl_\delta(u) \le (1+\delta)^H \dl(u) \le (1+\epsilon)r$.

Finally, let us consider the algorithm's running time. Observe that $|\Sigma_r|
= O(\log_{(1+\delta}r) = O(\frac{\log r}{\log(1+\delta)}) = O(\frac{\log
r}{\delta})$, where we have used that fact that $\ln(1+\delta) \approx \delta$
for sufficiently small $\delta$ (that is, sufficiently large $n$). By recalling
that $\delta = \Theta(\log n/\epsilon)$ we get a running time of $O(\log
n/\epsilon)^{O(\tw)}$, which implies the promised running time by Lemma
\ref{lem:fpt-logn}. 
\end{proof}

%% file: cw-approx.tex
\section{Clique-width revisited: FPT approximation scheme}\label{sec:fptas_CW} 

We give here an FPT-AS for \KC\ parameterized by~$\cw$, both
for unweighted and for weighted instances (for a weighted definition of~$\cw$ which we explain below). Our algorithm builds on the algorithm of
Section \ref{sec:fptas}, and despite the added generality of the
parameterization by~$\cw$, we are able to obtain an algorithm with
similar performance: for any $\epsilon>0$, our algorithm runs in time
$O^*((\cw/\epsilon)^{O(\cw)})$ and produces a $(k,(1+\epsilon)r)$-center if the
input instance admits a $(k,r)$-center. 

Our main strategy, which may be of independent interest, is to pre-process the
input graph $G=(V,E)$ in such a way that the answer does not change, yet producing a graph whose~$\tw$ is bounded by $O(\cw(G))$. The main insight
that we rely on, which was first observed by~\cite{GurskiW00},
is that a graph of low~$\cw$ can be transformed into a graph of low~$\tw$ if one removes all large bi-cliques. Unlike previous applications of
this idea (e.g. \cite{LampisMMU14}), we do not use the main theorem of
\cite{GurskiW00} as a ``black box'', but rather give an explicit construction of a
tree decomposition, exploiting the fact that \KC\ allows us to relatively
easily eliminate complete bi-cliques. As a result, we obtain a tree
decomposition of width not just bounded by some function of $\cw(G)$, but
actually $O(\cw(G))$.  

In the remainder we deal with the weighted version of \KC. To allow
clique-width expressions to handle weighted edges, we interpret the
clique-width join operation $\eta$ as taking three arguments. The
interpretation is that $\eta(a,b,w)$ is an operation that adds (directed) edges
from all vertices with label $a$ to all vertices with label $b$ and gives
weight $w$ to all these edges. It is not hard to see that if a graph has a
(standard) clique-width expression with $\cw$ labels, it can also be
constructed with $\cw$ labels in our context, if we replace every standard join
operation $\eta(a,b)$ with $\eta(a,b,1)$ followed by $\eta(b,a,1)$. Hence, the
algorithm we give also applies to unweighted instances parameterized by
(standard) clique-width.
We will also deal with a generalization of \KC, where we are also supplied, along with
the input graph $G=(V,E)$, a subset $I\subseteq V$ of \emph{irrelevant} vertices. In
this version, a $(k,r)$-center is a set $K\subseteq V\setminus I$, with $|K|=k$, such that
all vertices of $V\setminus I$ are at distance at most~$r$ from~$K$. Clearly,
the standard version of \KC\ corresponds to $I=\emptyset$. As we explain in the proof of 
Theorem \ref{thm_cw_approx}, this generalization does not make the problem
significantly harder.
In addition to the above, in this section we will allow edge weights to be
equal to $0$. This does not significantly alter the problem, however, if we are
interested in approximation and allow $r$ to be unbounded, as the following lemma shows: %(proof in Appendix \ref{append_cw_fptas}):
\begin{lemma}\label{lem:zero-ok}
There exists a polynomial algorithm which, for any $\epsilon>0$,  given an
instance $I = [G,w,k,r]$ of \KC, with weight function $w:V\to \mathbb{N}$,
produces an instance $I'=[G,w',k,r']$ on the same graph with weight function
$w':V\to\mathbb{N}^+$, such that we have the following: for any $\rho\ge1$, any
$(k,\rho r')$-center of~$I'$ is a $(k,\rho r)$-center of~$I$; any $(k,\rho r)$-center of~$I$ is a
$(k,(1+\epsilon)\rho r')$-center of~$I'$.  \end{lemma}
\begin{proof}
We define a scaling factor $B:= \frac{n}{\epsilon}$.   We set $w'((u,v)) =
B\cdot w((u,v)) + 1$, for all $(u,v)\in E$. We set $r'=B\cdot r$. Suppose that
we have a $(k,\rho r)$-center $K$ for $w$. We use the same solution for $w'$
and the maximum distance from $K$ to any vertex is at most $B\rho r + n = \rho
r' + n$. However, $n\le \epsilon r'$, hence this solution has maximum distance
at most $(1+\epsilon)\rho r'$. For the converse direction, suppose that we have
a solution for the function $w'$ with maximum distance $\rho r'$. We use the
same solution and now the cost is at most $\rho r'/B \le \rho r$, since
$w((u,v))\le w'((u,v))/B$.  
\end{proof}
Our main tool is the following lemma, whose strategy is to replace every large
label-set by two ``representative'' vertices, in a way that retains the same
distances among all vertices of the graph. Applying this transformation
repeatedly results in a graph with small treewidth. The main theorem of this section then follows from the above.
\begin{lemma}\label{cw_fptas_transform_lem} Given a \KC\ instance $G=(V,E)$
along with a clique-width expression $T$ for $G$ on $\textrm{cw}$ labels, we
can in polynomial time obtain a \KC\ instance $G'=(V',E')$ with $V\subseteq
V'$, and a tree decomposition of $G'$ of width $\textrm{tw}=O(\textrm{cw})$,
with the following property: for all $k,r$, $G$ has a $(k,r)$-center if and
only if $G'$ has a $(k,r)$-center. \end{lemma}
\begin{proof}
Suppose we are given a clique-width expression of $G$, represented as a binary
tree $T$, using $\cw$ labels. We will first construct a new clique-width
expression $T'$, using $\cw+3$ labels in a way that preserves all
$(k,r)$-centers. The end result will be that $T'$ does not contain any join
operations involving large label-sets. We will use this property to construct a
tree decomposition of the resulting graph.

Let $c$ be an  appropriate small constant (for concreteness, let $c=10$). We
will say that a label $l\in\{1,\ldots,\cw\}$ is \emph{big} in a node $x\in T$, if the
graph described by the sub-expression rooted at $x$ has at least $c$ vertices
with label $l$. We say that $l$ is \emph{newly big} in $x\in T$, if $l$ is big in $x$
but it is not big in any of the children of $x$ in $T$. Finally, we say that
$l$ is \emph{active inside} (respectively \emph{active outside}) in $x\in T$ if
$l$ is newly big in $x$ and, furthermore, there exists an arc $e\in E$ that is
not yet present in $x$, whose head (respectively tail) is incident on a vertex
with label $l$ in $x$.  In other words, a label is active (inside or outside)
in a sub-expression, if it contains many vertices and there is a join operation
(of the respective direction) that is applied to its vertices somewhere higher
in $T$.

Suppose that the nodes of $T$ are ordered in some arbitrary way, so that a
child is ordered before its parent. Our transformation is the following: as
long as there exists $l\in\{1,\ldots,\cw\}$ and $x\in T$ such that $l$ is
active (inside or outside) in $x$, we select the first such $x$ in the ordering
and let $l$ be an active label in $x$.  We insert directly above $x$ in $T$ the
following operations: if $l$ is active inside we introduce a vertex $s$ with
label $\cw+1$, and if $l$ is active outside we introduce a vertex $t$ with
label $\cw+2$; we then add the operations $\eta(\cw+1,l,0)$, $\eta(l,\cw+2,0)$,
followed by the rename operations $\rho(l,\cw+3), \rho(\cw+1,l),
\rho(\cw+2,l)$. In words, this transformation (potentially) adds a ``source''
$s$, or a ``sink'' $t$ to the graph (or both), gives a junk label $\cw+3$ to
all vertices that previously had label $l$ in $x$, and uses $s,t$ as the new
representatives of this class of vertices. We repeat this process until no
$x\in T$ contains an active label; this finishes in polynomial time because
once we eliminate all active labels from a node $x$, we do not re-visit it.  If
the original instance had a set of irrelevant vertices $I\subseteq V$, we
construct a set of irrelevant vertices $I'$ for $G'$ by adding to $I$ all the
$s,t$ vertices created in the above procedure.

Let us first argue that any $(k,r)$-center of the original instance can be
transformed into a $(k,r)$-center of the new instance. Consider one step of the
above procedure applied on a node $x\in T$. The main observation is that such a
step does not change the distance between any two vertices $u,v\in V$ in the
final graph. Suppose $u$ has label $l$ in $x$. Any edge incident on $u$ in $G$
that was removed as a result of this operation does not yet appear in $x$ and
hence can be replaced by a path of the same cost through $s$ or $t$. Similarly,
any $u,v$ path in the new graph that uses $s$ or $t$ in the new graph, must
also use an edge (that is not present in $x$) which was removed by the
operation. Such a path can be replaced by this edge in $G$. 

As a result of the above argument, all distances between vertices of $V$ are
preserved after the end of the procedure. Since all vertices of $V'\setminus V$
are irrelevant in $G'$, any $(k,r)$-center of $G$ is also a $(k,r)$-center of
$G'$ and vice-versa, as any solution can ignore these vertices and no solution can contain any of them. 

To see that $G'$ has small treewidth, observe that we now have a clique-width
expression for $G'$ that uses $\cw+3$ labels and has the following property:
for any join operation $\eta(a,b,w)$ applied on a node $x$, there are at most
$2c$ vertices with labels $a$ or $b$ in $x$. To see this, observe that this
operation is either an operation that was originally in $T$, or an operation
added during the transformation procedure. In the former case, neither label is
big in $x$, otherwise, there would be a descendant of $x$ where a label is
active, and we would have applied the operation. In the latter case, there is a
single vertex with label $b$, and $a$ was newly active, hence it cannot contain
more than $2c-2$ vertices (if it was the result of a rename or union operation
on two previously small label-sets).

We can therefore construct a tree decomposition of $G'$ by taking its
clique-width expression $T'$ and making a bag for each of its nodes. For the
bag $B_x$ that corresponds to node $x$ of $T'$, we find all labels
$l\in\{1,\ldots,\cw+3\}$ which are small in $x$ (that is, there are at most
$2c$ vertices of this label) and place all vertices with that label in $x$ in
$B_x$. This is a valid tree decomposition because: (i) as argued above all join
operations involve two small label-sets, hence any edge of the graph is
contained in some bag, (ii) if $u\in B_x$ but $u\not\in B_y$, for $B_y$ being
the parent of $B_x$, we know that $u$ also does not appear in any ancestor of
$B_y$, because once a vertex is part of a large label, it remains so and it
does not appear in the other child of $B_y$ (if it exists), or its descendants,
because if $B_y$ has two children it must correspond to a disjoint
union operation in $T'$. To complete the proof we observe that we have placed
at most $2c(\cw+3)$ vertices in each bag.  \end{proof}
\begin{theorem}\label{thm_cw_approx} Given $G=(V,E)$, $k,r\in\mathbb{N}^+$, clique-width expression $T$ for $G$ on $\textrm{cw}$
labels and $\epsilon>0$, there exists an algorithm that runs in
time $O^*((\cw/\epsilon)^{O(\cw)})$ and either produces a
$(k,(1+\epsilon)r)$-center, or correctly concludes that no $(k,r)$-center exists.
\end{theorem}
\begin{proof}
Given the \KC\ instance, we use Lemma \ref{cw_fptas_transform_lem} to obtain an
instance $G'$ and its tree decomposition of width $O(\cw)$. We then invoke
Lemma \ref{lem:zero-ok} to make all edge weights strictly positive, and then
use the algorithm of Section \ref{sec:fptas}.  We remark that this algorithm
can be easily adjusted to handle the more general case of the problem where
some vertices are irrelevant: we simply need to modify the computation on Forget
nodes to allow the algorithm to retain a solution that has not satisfied a
vertex $u$ if $u$ is irrelevant and the computation on Introduce nodes to disregard a solution that considers $u$ for inclusion in the center-set.  
\end{proof}